\documentclass{article}

\usepackage{arxiv}

\usepackage[utf8]{inputenc} 
\usepackage[T1]{fontenc}    
\usepackage{hyperref}       
\usepackage{url}            
\usepackage{booktabs}       
\usepackage{amsfonts}       
\usepackage{nicefrac}       
\usepackage{microtype}      
\usepackage{lipsum,enumerate,bm,mathrsfs,yfonts}
\usepackage{graphicx,amsmath,multicol,multirow}
\usepackage{url}
\usepackage{amsthm}
\usepackage[style=apa, backend=biber, natbib=true]{biblatex}
\usepackage{adjustbox,xr,comment} 

\usepackage{amsmath}
\usepackage{amssymb}
\usepackage{amsthm,tikz}
\usepackage{colortbl,caption,xr}
\usepackage{adjustbox}  
\usepackage{graphicx}
\usepackage{lscape}
\usepackage{rotating}
\usepackage{graphicx}
\usepackage{float}
\usepackage{url}
\usepackage{subfigure}
\usepackage{dcolumn}
\usepackage{multirow}
\usepackage{verbatim,booktabs, array,arydshln}
\usepackage{color}
\usepackage{setspace,amsthm}
\usepackage{bbm}
\usepackage{hhline}
\usepackage{bm}
\usepackage{dsfont,import}
\usepackage{enumerate}
\usepackage{graphicx}
\usetikzlibrary{shapes.geometric, arrows, shadows,positioning,arrows.meta,calc}
\usepackage{graphicx}
\usepackage{booktabs,multicol,multirow,tabularx,array}
\usepackage{longtable}
\usepackage{caption}
\usepackage{graphicx}

\newtheorem{theorem}{Theorem}[section]   
\newtheorem{corollary}[theorem]{Corollary}

\theoremstyle{definition}

\theoremstyle{remark}

\theoremstyle{assumption}

\theoremstyle{proposition}

\newcommand{\Ind}[1]{\mathrm{I}\!\left\{#1\right\}}
\newcommand{\E}{\mathrm{E}}
\newcommand{\Var}{\mathrm{Var}}

\newcommand{\WR}{\mathrm{WR}}
\newcommand{\NB}{\mathrm{NB}}
\newcommand{\WO}{\mathrm{WO}}

\newcommand{\bT}{\mathbf{T}}
\newcommand{\bY}{\mathbf{Y}}
\newcommand{\bYobs}{\widetilde{\mathbf{Y}}}
\newcommand{\bDelta}{\boldsymbol{\Delta}}
\newcommand{\bZ}{\mathbf{Z}}
\newcommand{\bpi}{\boldsymbol{\pi}}

\newcommand{\bOmega}{\boldsymbol{\Omega}}

\newcommand{\mbf}[1]{\mathbf{#1}}

\newcommand{\bz}{\mbf{z}}

\newcommand{\dee}{\,\mathrm{d}}

\title{Making censored pairs count: conditional tie weighting for win statistics with composite survival endpoints}

\author{
 Xi Fang \\
    Data Science Institute \\
    Medical College of Wisconsin \\
    Milwaukee, WI, USA\\
\And
 Fan Li \\
Department of Biostatistics \\
Yale School of Public Health \\
New Haven, CT, USA\\
  \texttt{fan.f.li@yale.edu} \\
}

\begin{document}
\maketitle
\begin{abstract}
Hierarchical composite endpoints are increasingly used in clinical trials to compare patients first on the most clinically important outcome and then, only when that comparison is tied, on lower priority outcomes. Under right censoring, a lower priority comparison may already be observed but still cannot contribute because the higher priority genuine tie required for descent through the hierarchy is not confirmed. Existing restricted win-statistic estimators address censoring by requiring such ties from higher priority to be observed as genuine ties. This all-or-nothing rule preserves the restricted-time estimand, but excludes pairs with censoring-induced ties even when their lower priority comparisons contain useful information. We propose conditional tie weighting, which replaces the unavailable higher priority genuine-tie indicator by its conditional probability given the observed pairwise data. The resulting estimator targets the same restricted-time win probabilities while allowing partially observed pairs to contribute fractionally when their lower priority comparison is informative. We establish identification and large-sample theory for the resulting two-sample U-statistics with estimated nuisance functions, and derive sandwich variance estimators for the win ratio, net benefit, and win odds. Simulations show substantial efficiency gains, especially under heavier censoring and longer restriction horizons. A reanalysis of the HF-ACTION trial illustrates how conditional tie weighting recovers information from censoring-induced ties in death-first hospitalization comparisons.
\end{abstract}


\section{Introduction} \label{sec:intro}

Composite endpoints are increasingly common in clinical trials when treatment benefit cannot be adequately summarized by any single outcome alone. In cardiovascular trials and heart failure trials, where mortality and major nonfatal clinical events may all contribute to the overall disease burden, a primary analysis that reflects multiple dimensions of patient experience is often desirable for both efficiency and interpretation. At the same time, component outcomes within a composite endpoint commonly differ in clinical importance, frequency, and responsiveness to treatment. This complexity has motivated the evolving development of the hierarchical, or prioritized, composite endpoints, in which case the outcomes are ordered according to clinical importance and compared sequentially based on that priority ordering \citep{buyse2010generalized}.  Building on the foundational work of \citet{finkelstein1999combining} and \citet{buyse2010generalized}, \citet{pocock2012win} introduced the win ratio as a summary measure for hierarchical composite endpoints, and subsequent developments have broadened this framework to include related measures such as the net benefit and the win odds \citep{dong2023stratified}. 

Our work is motivated by clinical trials with prioritized time-to-event outcomes. The Heart Failure: A Controlled Trial Investigating Outcomes of Exercise Training (HF-ACTION) trial \citep{o2009efficacy} randomized 2,331 medically stable outpatients with heart failure and reduced ejection fraction at 82 centers in the United States, Canada, and France, with a median follow-up of approximately 30 months. The original primary endpoint was all-cause mortality or all-cause hospitalization. This endpoint reflects a common clinical tension in chronic disease trials, where death and hospitalization are both clinically meaningful, but death is naturally prioritized over hospitalization, whereas hospitalization occurs more frequently and can carry substantial information about treatment benefit. In this study, mortality occurred in about 16--17\% of patients, whereas at least one hospitalization occurred in about 63--65\% of patients. Thus, a hierarchical comparison that prioritizes death while allowing hospitalization to contribute when death does not distinguish the pair is clinically appealing. Although HF-ACTION provides a two component example, the same inferential problem arises more generally in any prioritized composite survival endpoint when a lower priority component can contribute only after all higher priority components form genuine ties. The same feature that makes the hierarchical comparison clinically interpretable also creates a statistical difficulty under right censoring. To contribute a lower priority comparison, all higher priority ties must be genuine ties through the restriction time, which censoring may prevent. When censoring prevents a genuine tie from being determined, the pair may have only a censoring-induced tie on the higher priority component, even when the lower priority comparison itself has already been observed. Thus, censoring can block descent through the hierarchy, not because the lower priority outcome is unobserved, but because the required higher priority genuine tie is not confirmed. Such pairs are not useless. Instead, they contain partial information about whether the latent higher priority outcomes would have formed genuine ties through the restriction time.

This issue is closely connected to the estimand question for win statistics. A win estimand is determined not only by the component hierarchy, but also by the rule used to define ties and by the time frame over which pairs are compared. However, the conventional, unrestricted win ratio need not correspond to a well-defined population quantity, because different patient pairs may be compared over different follow up intervals, making the resulting estimand inherently dependent on the unknown censoring distribution and specific follow up time rather than solely on the underlying outcome process \citep{oakes2016win,dong2020win,mao2024defining}. Such dependence is difficult to reconcile with the ICH E9(R1) estimands framework, which requires that the target of inference be specified separately from the censoring process \citep{ich2019e9r1}. A natural resolution is to fix a clinically meaningful restriction time and define wins, losses, and ties using the latent component outcomes restricted at that horizon. Under incomplete follow-up, right censoring can prevent the observed data from resolving whether a pair is a win, loss, or tie by the restriction time, and estimators that treat indeterminate comparisons as ties or discard them without correction are generally biased for the restricted win probabilities \citep{mao2024defining}. Existing work has addressed this by applying inverse probability of censoring weighting (IPCW) to recover the target estimand. For example, \citet{dong2020inverse} and \citet{dong2021adjusting} proposed IPCW adjusted win statistics under independent and covariate-dependent censoring; \citet{cui2025ipcw} broadened IPCW estimation under explicit tie rules and \citet{cao2026generalized} applied IPCW in the regression setting based on win measures.

For hierarchical composite endpoints, however, the main unresolved issue is not bias correction itself, but also the efficiency with which censored pairwise information, especially from lower priority comparisons, is utilized. A lower priority comparison conributes only when the higher priority components are tied. At the restricted horizon, this tie can be genuine, where neither individual experiences the higher priority event by that horizon, or censoring-induced, where early censoring leaves the status of those components indeterminate so that the pair only appears tied. Existing IPCW methods require this tie to be genuine before a lower priority comparison can contribute, so pairs with a censoring-induced tie before the restriction time are excluded from those comparisons. Although the this IPCW estimator corrects the resulting bias, the excluded pairs still represent a source of efficiency loss, which becomes more pronounced as censoring increases and the restriction horizon increases. 

\begin{figure}
    \centering
    \includegraphics[width=0.8\linewidth]{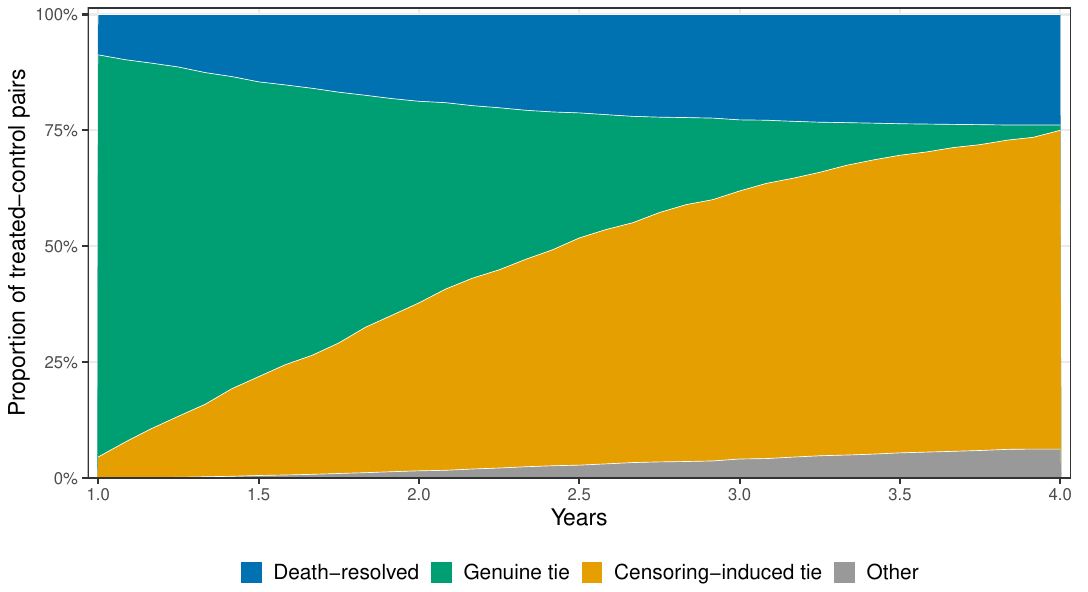}
    \caption{Composition of the first priority death gate in the HF-ACTION data. At each restriction time $\tau$, all exercise training versus usual care pairs are classified as resolved at death, genuine ties, censoring-induced ties, or other non-descending states. Pairs resolved at death are determined by an observed death before $\tau$. Genuine ties are pairs in which both patients are observed alive through $\tau$. Censoring-induced ties have no observed death by $\tau$ but at least one patient is censored before $\tau$, so descent to the lower priority hospitalization comparison cannot be verified by the standard IPCW rule. The other category contains the remaining non-descending death comparisons, such as pairs in which an observed death occurs only after the other patient has already been censored, as well as exact same-day deaths arising from day-level event recording.}
    \label{fig:hfaction_composition}
\end{figure}

This information loss is visible in the motivating HF-ACTION data. Figure \ref{fig:hfaction_composition} classifies all exercise training versus usual care pairs according to the observed status of the first-priority death comparison as the restriction time varies. As time increases, the proportion of genuine ties on death decreases, while the proportion of censoring-induced ties increases. These censoring-induced ties are precisely the pairs for which descent to the hospitalization comparison cannot be verified by the standard IPCW estimator\citep{cui2025ipcw}, even though the observed follow-up still contains information about the probability that the latent death outcomes would have formed a genuine tie through that restriction time. This observed data pattern motivates conditional tie weighting. Rather than treating the higher priority tie requirement as an all-or-nothing observed quantity, we replace the unavailable genuine-tie indicator by its conditional probability given the observed pairwise data. This converts censoring-induced ties from zero contributions into fractional contributions when the lower priority comparison is informative. The proposed estimator therefore targets the same restricted-time win probabilities as the standard IPCW estimator, but uses more of the partially observed pairwise information created by right censoring. The main contribution of this paper is to show that, for hierarchical composite survival endpoints, efficiency loss under censoring arises not only because event times are missing, but also because censoring prevents genuine higher priority ties from being confirmed. Conditional tie weighting addresses this specific source of information loss while preserving the original clinical hierarchy. We develop large-sample theory for the resulting two-sample U-statistics with estimated nuisance functions, derive sandwich variance estimators for the win ratio, net benefit, and win odds, evaluate finite-sample performance under nuisance-model misspecification and copula sensitivity, and apply the method to the HF-ACTION trial.

\section{Conditional tie weighting for restricted win statistics} \label{sec:method}

\subsection{Restricted win estimands} \label{sec:setup}
Consider a two-arm randomized clinical trial comparing a treatment condition ($A=1$) with a control condition ($A=0$). The trial enrolls $n_1$ individuals in the treatment arm and $n_0$ individuals in the control arm, with $n=n_1+n_0$. For individual $i$, write $A_i\in\{1,0\}$ for the treatment assignment. The clinical outcome is a hierarchical composite endpoint consisting of $Q$ prioritized time-to-event components, ordered $q=1,2,\dots,Q$ from the most clinically important to the least. A typical example in cardiovascular trials is $Q=2$ with cardiovascular death as the first priority component ($q=1$) and nonfatal hospitalization as the second-priority component ($q=2$). For individual $i$ and component $q$, let $T_{qi}$ denote the latent event time, that is, the time at which the component-$q$ event would occur in the absence of censoring or any other form of incomplete observation. We write $\bT_i=(T_{1i},\dots,T_{Qi})$ for the full latent outcome vector. As explained by \citet{mao2024defining}, the ordinary unmatched win ratio defined without a fixed time horizon does not target a well-defined treatment effect, because different pairs are compared over different time frames depending on when follow-up ends.
To avoid this, we fix a clinically meaningful restriction horizon $\tau>0$ and define every pairwise comparison in terms of the restricted latent outcomes $Y_{qi}=T_{qi}\wedge\tau$, which equals the event time if the event occurs by $\tau$ and equals $\tau$ otherwise. We denote $\bY_i=(Y_{1i},\dots,Y_{Qi})$ as the vector of restricted outcomes. For a treated individual $i$ ($A_i=1$) and a control individual $j$ ($A_j=0$), the hierarchical comparison at time $\tau$ proceeds through the ordered components, where the pair is first compared on the most important component, and only if that comparison results in a tie does the pair move to the second component, and so on. Under the continuity assumption, a tie on component $q$ at horizon $\tau$ means $Y_{qi}=Y_{qj}$, which can happen only if both event times exceed $\tau$, namely $\{Y_{qi}=Y_{qj}\}=\{T_{qi}>\tau,\;T_{qj}>\tau\}$. This continuity assumption rules out exact event time ties within the observation window. In trials where event times are recorded on a coarse time scale, or where clinically meaningful equivalence margins are used to define ties, the pairwise win and loss functions in \eqref{eq:winfun} and \eqref{eq:lossfun} would require corresponding modification. We focus on continuous event times and zero equivalence margins. The hierarchical comparison can be written as a pairwise win function
\begin{equation}\label{eq:winfun}
  W_{ij}(\tau) =\sum_{q=1}^{Q} \left[\prod_{k=1}^{q-1} \Ind{Y_{ki}=Y_{kj}} \right] \Ind{Y_{qi}>Y_{qj}},
\end{equation}
with the convention that the empty product (when $q=1$) equals one. 
Since the hierarchical cascade terminates as soon as a win or loss is identified, at most one term in the sum \eqref{eq:winfun} is nonzero for any given pair. The pairwise loss function can be defined analogously as
\begin{equation}\label{eq:lossfun}
  L_{ij}(\tau)=\sum_{q=1}^{Q} \left[\prod_{k=1}^{q-1} \Ind{Y_{ki}=Y_{kj}} \right] \Ind{Y_{qj}>Y_{qi}}.
\end{equation}
Since for each component either $Y_{qi}>Y_{qj}$, $Y_{qj}>Y_{qi}$, or $Y_{qi}=Y_{qj}$, the three pairwise outcomes are mutually exclusive and satisfy $W_{ij}(\tau)+L_{ij}(\tau)+U_{ij}(\tau)=1$, where $U_{ij}(\tau)=\prod_{q=1}^Q\Ind{Y_{qi}=Y_{qj}}$ is the overall tie indicator. 

We define the target estimands as the restricted-time win probabilities:
\begin{equation}\label{eq:win-prob}
  \pi_t(\tau)=\E\left[W_{ij}(\tau)\right],
  \quad
  \pi_c(\tau)=\E\left[L_{ij}(\tau)\right],
  \quad
  \pi_u(\tau)=1-\pi_t(\tau)-\pi_c(\tau).
\end{equation}
These are well-defined population quantities that depend only on the joint distribution of the latent event times and on the restriction horizon $\tau$. By the additive structure of $W_{ij}(\tau)$, the win and loss probabilities decompose by component:
\begin{equation}\label{eq:pi-decomp}
  \pi_t(\tau)=\sum_{q=1}^Q \pi_{tq}(\tau),
  \quad
  \pi_c(\tau)=\sum_{q=1}^Q \pi_{cq}(\tau),
\end{equation}
where $\pi_{tq}(\tau)=\E\bigl\{\bigl[\prod_{k<q}\Ind{Y_{ki}=Y_{kj}}\bigr]\Ind{Y_{qi}>Y_{qj}}\bigr\}$ and $\pi_{cq}(\tau)$ is defined symmetrically. The component-specific probabilities can quantify which part of the hierarchy is driving the overall effect. Standard summary measures built from $(\pi_t,\pi_c)$ include the win ratio $\WR(\tau)=\pi_t(\tau)/\pi_c(\tau)$ \citep{pocock2012win}, the net benefit $\NB(\tau)=\pi_t(\tau)-\pi_c(\tau)$ \citep{buyse2010generalized}, and the win odds $\WO(\tau)=\{{\pi_t(\tau)+\frac{1}{2}\pi_u(\tau)}\}/\{{\pi_c(\tau)+\frac{1}{2}\pi_u(\tau)}\}$ \citep{dong2020inverse}. Because $\tau$ is part of the estimand, it should be chosen to correspond to a clinically meaningful comparison horizon rather than to the realized censoring distribution. In applications, natural choices include a planned follow-up time or one or more clinically interpretable landmark times at which the endpoint hierarchy is considered meaningful. When several restriction times are scientifically relevant, the estimands $\pi_t(\tau)$, $\pi_c(\tau)$, $\WR(\tau)$, $\NB(\tau)$, and $\WO(\tau)$ may be reported over a prespecified set of values of $\tau$, or displayed as functions of $\tau$, with the understanding that changing $\tau$ changes the target comparison. The component-specific probabilities in \eqref{eq:pi-decomp} are also important for interpretation, because an overall win ratio or net benefit may be driven primarily by the highest-priority outcome or by more frequent lower priority outcomes. We therefore view $\{\pi_{tq}(\tau),\pi_{cq}(\tau):q=1,\ldots,Q\}$ not as secondary technical quantities, but as part of the clinical reporting of a hierarchical composite endpoint.

\subsection{Censoring-induced ties and the standard IPCW estimator} \label{sec:censor_tie_ipcw}

In practice, the latent event times are not fully observed due to right censoring. Let $C_i$ denote the censoring time for individual $i$. We focus on the common censoring setting, so that an individual who is lost to follow-up at time $C_i$ has all component outcomes censored simultaneously. We define the observed time at restriction time \(\tau\) and the event indicator by $\widetilde{Y}_{qi}(\tau)=T_{qi}\wedge C_i\wedge\tau$ and $\delta_{qi}(\tau)=\Ind{T_{qi}\le C_i\wedge\tau}$. In particular, when $\delta_{qi}(\tau)=0$, the event is not observed by $\tau$, which may occur either because the individual is censored before the event and before $\tau$, or because the individual remains event free through $\tau$. For every individual $i$, we let $\bZ_i$ denote a vector of baseline covariates and assume covariate-dependent censoring such that $C_i\perp \bT_i\mid (A_i,\bZ_i)$. We further define the arm-specific censoring survival function $G_a(t\mid \bm{Z})=\Pr(C_i>t\mid A_i=a,\bZ_i=\bm{Z})$, with the positivity condition $G_a(\tau\mid \bZ_i)>0$ satisfied almost surely for $a=0,1$. Under right censoring, the observed tie indicator for component $k$ at time $\tau$ is $$\widetilde{U}_{ij,k} = 1-\Ind{\widetilde{Y}_{ki}(\tau)>\widetilde{Y}_{kj}(\tau)}\,\delta_{kj}(\tau)
    -\Ind{\widetilde{Y}_{kj}(\tau)>\widetilde{Y}_{ki}(\tau)}\,\delta_{ki}(\tau), \nonumber$$
which equals one not only for genuine ties but also for every indeterminate comparison scenario due to right censoring. In this case, the $q$th observed win kernel is 
$$\widetilde{W}_{ij,q}(\tau)=\left\{\prod_{k<q}\widetilde{U}_{ij,k}\right\}\Ind{\widetilde{Y}_{qi}(\tau)>\widetilde{Y}_{qj}(\tau)}\,\delta_{qj}(\tau),$$ 
with the overall $\widetilde{W}_{ij}(\tau)=\sum_{q=1}^Q\widetilde{W}_{ij,q}(\tau)$. The observed loss kernel $\widetilde{L}_{ij}(\tau)$ is defined analogously. Although the partition $\widetilde{W}_{ij}(\tau)+\widetilde{L}_{ij}(\tau)+\widetilde{U}_{ij}(\tau)=1$ mirrors the original partition in Section \ref{sec:setup}, the two identities are not equal, because $\widetilde{U}_{ij}(\tau)$ absorbs both genuine ties and censoring-induced ties. As a result, $\E[\widetilde{W}_{ij}(\tau)]\neq\pi_t(\tau)$ in general, and the ordinary sample average of $\widetilde{W}_{ij}(\tau)$ is not an unbiased estimator of $\pi_t(\tau)$.

To construct an unbiased estimator, let $O_{ij,q}(\tau)$ denote the indicator that the component $q$ comparison is ascertainable from the observed data, with observation probability $p_{ij,q}(\tau)=\Pr\{O_{ij,q}(\tau)=1\mid W_{ij,q}(\tau),L_{ij,q}(\tau),\bZ_i,\bZ_j\}>0$ almost surely. We can show that 
$$\E\!\left[\{O_{ij,q}(\tau)\,W_{ij,q}(\tau)\}/\{p_{ij,q}(\tau)\}\right]=\pi_{tq}(\tau),$$ 
and therefore an unbiased estimator $\widehat{\pi}_t(\tau)=(n_1 n_0)^{-1}\sum_{i:A_i=1}\sum_{j:A_j=0}\sum_{q=1}^Q O_{ij,q}(\tau)\,W_{ij,q}(\tau)/p_{ij,q}(\tau)$. This construction is generic and holds for any choice of observation event satisfying the positivity condition. The practical question is therefore not only how to correct bias from censoring, but also which observation event should be used. Different choices can target the same restricted-time win probabilities while using different amounts of censored pairwise information.
%
Given a restricted time horizon $\tau$, \citet{cui2025ipcw} proposed an IPCW estimator for the win estimands. Specifically, for $q=1$, a win is ascertainable whenever the control individual's event is observed ($\delta_{1j}(\tau)=1$) and the treated individual remains at risk beyond that time ($C_i>T_{1j}$), giving $p_{ij,1}(\tau)=G_1(\widetilde{Y}_{1j}(\tau)\mid\bZ_i)\,G_0(\widetilde{Y}_{1j}(\tau)\mid\bZ_j)$. For $q\ge 2$, the higher priority components must also form a genuine tie $\prod_{k<q}\Ind{T_{ki}>\tau,T_{kj}>\tau}$, 
which requires $\widetilde{Y}_{ki}(\tau)=\tau$ for both individuals on every higher priority component. The resulting estimator is given by
\begin{align} \label{eq:cuiA}
\widehat{\pi}_{tq}^{(\text{ipcw})}(\tau)&=\frac{1}{n_1 n_0}\sum_{i:A_i=1}\sum_{j:A_j=0}
\begin{cases}
\displaystyle\frac{\Ind{\widetilde{Y}_{1i}(\tau)>\widetilde{Y}_{1j}(\tau)}\,\delta_{1j}(\tau)}
{\widehat{G}_1\!\bigl(\widetilde{Y}_{1j}(\tau)\mid\bZ_i\bigr)\,\widehat{G}_0\!\bigl(\widetilde{Y}_{1j}(\tau)\mid\bZ_j\bigr)}, & q=1, \\[16pt]
\displaystyle\frac{\left[\prod_{k<q}\Ind{\widetilde{Y}_{ki}(\tau)=\tau,\,\widetilde{Y}_{kj}(\tau)=\tau}\right]
\Ind{\widetilde{Y}_{qi}(\tau)>\widetilde{Y}_{qj}(\tau)}\,\delta_{qj}(\tau)}
{\widehat{G}_1(\tau\mid\bZ_i)\,\widehat{G}_0(\tau\mid\bZ_j)}, & q\ge 2,
\end{cases}
\end{align}
with $\widehat{\pi}_t^{(\text{ipcw})}(\tau)=\sum_{q=1}^Q\widehat{\pi}_{tq}^{(\text{ipcw})}(\tau)$. The loss estimator $\widehat{\pi}_c^{(\text{ipcw})}(\tau)$ is defined symmetrically. The corresponding estimator is denoted as \(\widehat{\WR}^{(\text{ipcw})}(\tau)\), \(\widehat{\WO}^{(\text{ipcw})}(\tau)\), and \(\widehat{\NB}^{(\text{ipcw})}(\tau)\) of WR, WO, and NB, respectively. In principle, the censoring survival function $G_a(t\mid\bm{Z})$ can be estimated nonparametrically by the Kaplan-Meier estimator under completely independent censoring, or by a Cox proportional hazards model under covariate dependent censoring. 
For completeness, we show the consistency of this estimator \eqref{eq:cuiA} in \ref{supp:consistent_cui},.


Despite its simplicity, \eqref{eq:cuiA} uses an all-or-nothing rule for ties on the higher priority components at $q\ge 2$. The genuine tie requirement $\prod_{k<q}\Ind{\widetilde{Y}_{ki}(\tau)=\tau,\;\widetilde{Y}_{kj}(\tau)=\tau}$ discards every pairs with a censoring-induced tie, that is, every pair whose status on the higher priority components is obscured by censoring before \(\tau\), even when the lower priority comparison itself is observed. A pair satisfying $T_{ki}\wedge T_{kj}>\tau$ for all $k<q$ but $C_i\wedge C_j\le\tau$ are excluded from the component $q$ comparison. This IPCW scheme corrects the bias from this exclusion, but at the cost of efficiency due to removing such censoring-induced ties. More concretely, consider \(Q=2\) with death prioritized over hospitalization. Suppose the control individual has an observed hospitalization at time \(t<\tau\) and is subsequently censored at \(u\in(t,\tau)\), whereas the treated individual remains under observation beyond \(t\). Although the hospitalization comparison is observed at time \(t\), the death comparison cannot be confirmed as a genuine tie through the restriction time \(\tau\). The observed data therefore contain a censoring-induced tie on death. The IPCW estimator of \citet{cui2025ipcw} assigns zero contribution to this hospitalization comparison, whereas the observed survival information through \(u\) still provides partial information about the probability that the latent death outcomes would have formed a genuine tie through \(\tau\). 
Exploiting this mechanistic insight, our method aims to convert that partial information into a fractional contribution to the final win statistic (Figure \ref{fig:two_illustration}). 
\begin{figure}
    \centering
    \includegraphics[width=1\linewidth]{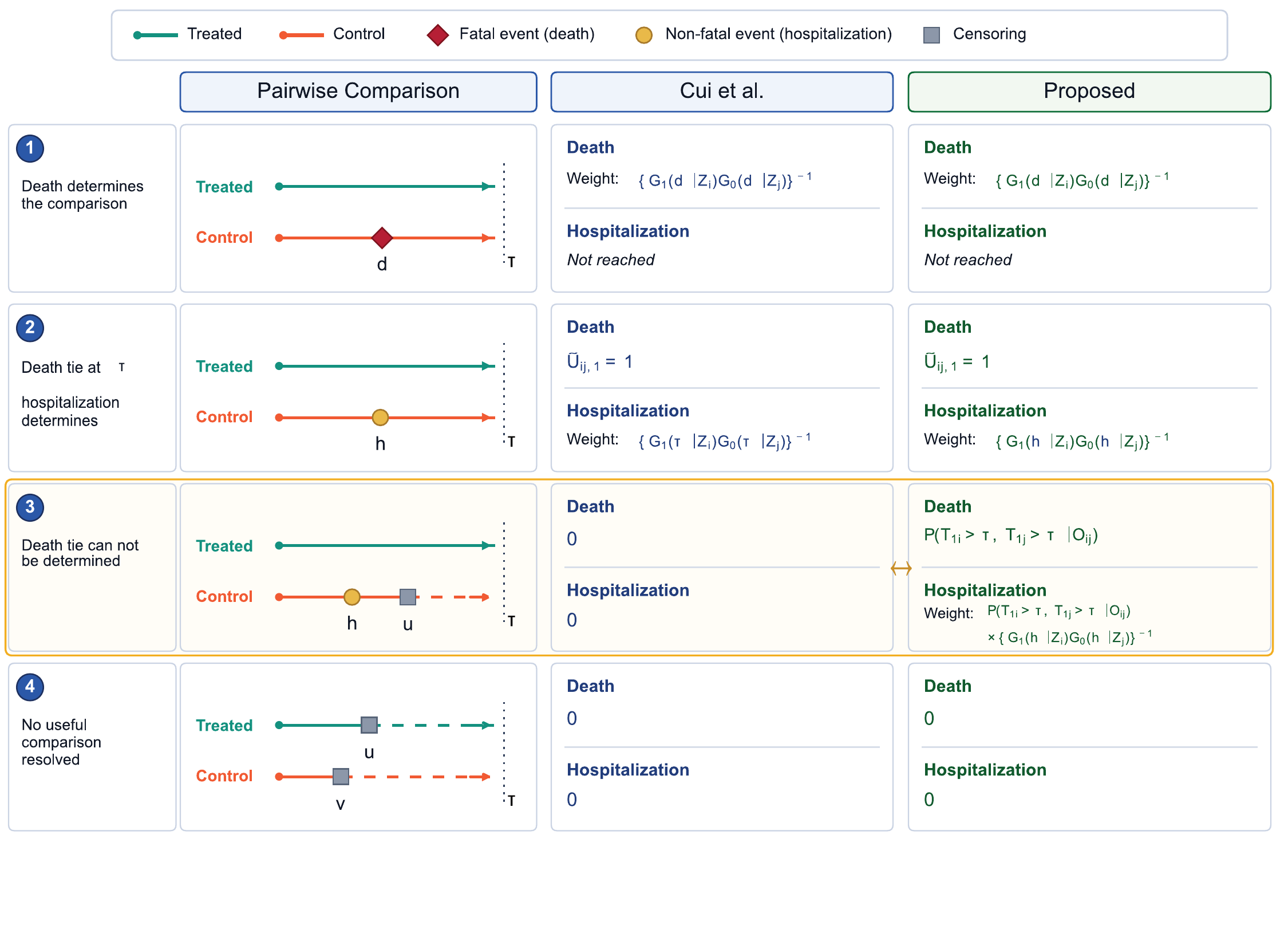}
    \caption{Illustration of pairwise weighting rules under the IPCW estimator of \citet{cui2025ipcw} and the proposed conditional tie-weighting estimator for a two component hierarchical endpoint restricted at time $\tau$, with death as the first priority component and hospitalization as the second priority component. In Scenario 1, death determines the comparison before $\tau$, and both estimators use the same first priority IPCW weight evaluated at the control death time $d$. In Scenario 2, the death comparison is tied through $\tau$, and the pair descends to hospitalization, where the IPCW estimator of \citet{cui2025ipcw} and the proposed estimator weight the hospitalization comparison at $\tau$. In Scenario 3, hospitalization is observed at $h$, but the death tie can not be determined because the control patient is censored at $u<\tau$. The IPCW estimator of \citet{cui2025ipcw} assigns zero weight to the hospitalization comparison, while the proposed estimator replaces the undetermined death tie gate by the conditional probability $P(T_{1i}>\tau,T_{1j}>\tau\mid O_{ij})$ and retains a fractional hospitalization weight. In Scenario 4, no useful comparison is resolved before censoring, so both estimators assign zero contribution. Here $O_{ij}$ denotes the observed pairwise data for patients $i$ and $j$.}
    \label{fig:two_illustration}
\end{figure}


\subsection{Conditional tie weighting} \label{sec:ctw}

To recover information from lower priority comparisons with censoring-induced higher priority ties, we modify the observation event underlying \eqref{eq:cuiA}. Instead of requiring both individuals to be followed event free through $\tau$ on the higher priority components, we weight the observed lower priority comparison at the time it is resolved and incorporate the genuine tie requirement on the higher priority components through a conditional tie probability given the observed data. Specifically, let $\overline{\delta}_{qi}(\tau)=\prod_{k<q}\{1-\delta_{ki}(\tau)\}$ indicate that no higher priority event is observed for individual $i$ by time $\tau$. Under the common censoring assumption, $\overline{\delta}_{qi}(\tau)=1$ implies $\widetilde{Y}_{1i}(\tau)=\cdots=\widetilde{Y}_{(q-1)i}(\tau)=C_i\wedge\tau$, a common value we denoted by \(u_{iq}\). When $\overline{\delta}_{qi}(\tau)=0$, we set $u_{iq}=\tau$ so that the ratios introduced below remain well-defined. Intuitively, the quantity $u_{iq}$ is the last time up to which individual $i$ is known to be event free on all higher priority components. In \eqref{eq:cuiA}, \(u_{iq}\) enters only through the binary indicator $\mathrm{I}\{u_{iq}=\tau\}$, so a pair contributes only when individual \(i\) is observed event free through \(\tau\). In our proposed estimator, \(u_{iq}\) enters continuously, through the conditional probability that individual \(i\) remains event free on the higher priority components from \(u_{iq}\) to \(\tau\). To proceed, for each arm $a\in\{0,1\}$ and $q\ge 2$, we first define the joint survival function
\begin{equation}\label{eq:jointSq}
  \mathcal{S}_{q,a}(u,t\mid \bm{Z})=\Pr\!\bigl(T_{1i}>u,\dots,T_{(q-1)i}>u,\,T_{qi}>t\bigm| A_i=a,\,\bZ_i=\bm{Z}\bigr), \quad 0\le t\le u\le \tau,
\end{equation}
and let $\mathcal{H}_{q,a}(u,t\mid \bm{Z})=-\partial\,\mathcal{S}_{q,a}(u,t\mid \bm{Z})/\partial t$ denote the associated subdensity, so that $\mathcal{H}_{q,a}(u,t\mid\bm{Z})\,\mathrm{d}t$ is the conditional probability that the component $q$ event occurs in $[t,t+\mathrm{d}t)$ while all higher priority components remain event free beyond $u$. The roles of $u$ and $t$ are asymmetric by design, where the higher priority components are required to remain event free through the horizon $u$, while the $q$th component is evaluated at the local time $t$ that resolves the pairwise comparison. Next, the joint survival function and subdensity allow us to define two conditional tie probabilities,
\begin{equation}\label{eq:Rq}
  \mathcal{R}_{q,a}^{>}(u,t\mid \bm{Z})=\frac{\mathcal{S}_{q,a}(\tau,t\mid \bm{Z})}{\mathcal{S}_{q,a}(u,t\mid \bm{Z})},
  \qquad
  \mathcal{R}_{q,a}^{=}(u,t\mid \bm{Z})=\frac{\mathcal{H}_{q,a}(\tau,t\mid \bm{Z})}{\mathcal{H}_{q,a}(u,t\mid \bm{Z})},
\end{equation}
corresponding to the two distinct roles an individual can play in the comparison for the $q$th component. The first ratio gives the conditional probability that all higher priority components remain event free through $\tau$, given that they are event free through $u$ and the component $q$ event has not yet occurred by time $t$. The second is the analogous probability given that the component $q$ event occurs at time $t$. Under the positivity conditions stated in Section \ref{sec:asymptotics}, both ratios lie in $[0,1]$. These two ratios arise from conditioning on the observed information available at the time the component-$q$ comparison is resolved. For a win on component $q$ resolved by the control individual's component-$q$ event at time $t$, the treated individual is known to have no observed higher priority event through $u_{iq}$ and to remain event free on component $q$ beyond $t$. Thus,
\[
\E\!\left[\prod_{k<q}\Ind{T_{ki}>\tau}\mid T_{1i}>u_{iq},\ldots,T_{(q-1)i}>u_{iq},\,T_{qi}>t,\,A_i=1,\,\bZ_i\right]=\mathcal{R}_{q,1}^{>}(u_{iq},t\mid\bZ_i).\]
For the control individual whose component-$q$ event occurs at time $t$, the corresponding conditional expectation conditions on the event time itself,
\[
\E\!\left[\prod_{k<q}\Ind{T_{kj}>\tau}\mid T_{1j}>u_{jq},\ldots,T_{(q-1)j}>u_{jq},\,T_{qj}=t,\,A_j=0,\,\bZ_j\right]=\mathcal{R}_{q,0}^{=}(u_{jq},t\mid\bZ_j).
\]
The product of these two conditional expectations is the conditional probability that the higher priority components would have formed genuine ties through $\tau$, given the observed pairwise information used to resolve the lower priority comparison. This is the quantity that is unavailable under censoring and is replaced probabilistically in the proposed estimator. For $q=1$, conditional tie weighting coincides with the IPCW estimator in \eqref{eq:cuiA}. For $q\ge 2$, the proposed conditional tie weighting estimator is
\begin{align}\label{eq:ctw-win}
  \widehat{\pi}_{tq}^{(\text{ctw})}(\tau)=
  \frac{1}{n_1 n_0}\sum_{i:A_i=1}\sum_{j:A_j=0}&
  \frac{\overline{\delta}_{qi}(\tau)\,\overline{\delta}_{qj}(\tau)\;\widehat{\mathcal{R}}_{q,1}^{>}\!\bigl(u_{iq},\,\widetilde{Y}_{qj}(\tau)\bigm|\bZ_i\bigr)\;\widehat{\mathcal{R}}_{q,0}^{=}\!\bigl(u_{jq},\,\widetilde{Y}_{qj}(\tau)\bigm|\bZ_j\bigr)\;
  }{\widehat{G}_1\!\bigl(\widetilde{Y}_{qj}(\tau)\bigm|\bZ_i\bigr)\;\widehat{G}_0\!\bigl(\widetilde{Y}_{qj}(\tau)\bigm|\bZ_j\bigr)}\times\nonumber \\
& \Ind{\widetilde{Y}_{qi}(\tau)>\widetilde{Y}_{qj}(\tau)}\,\delta_{qj}(\tau)
\end{align}
and the component $q$ loss estimator $\widehat{\pi}_{cq}^{(\text{ctw})}(\tau)$ is defined symmetrically. The overall conditional tie weighting estimators are then given by \(\widehat{\pi}_{t}^{(\text{ctw})}(\tau)=\sum_{q=1}^{Q}\widehat{\pi}_{tq}^{(\text{ctw})}(\tau)\), \(\widehat{\pi}_{c}^{(\text{ctw})}(\tau)=\sum_{q=1}^{Q}\widehat{\pi}_{cq}^{(\text{ctw})}(\tau)\), with $\widehat{\pi}_{t1}^{(\text{ctw})}(\tau)=\widehat{\pi}_{t1}^{(\text{ipcw})}(\tau)$ and $\widehat{\pi}_{c1}^{(\text{ctw})}(\tau)=\widehat{\pi}_{c1}^{(\text{ipcw})}(\tau)$. The corresponding win ratio, net benefit, and win odds are denoted by $\widehat{\WR}^{(\text{ctw})}(\tau)$, $\widehat{\NB}^{(\text{ctw})}(\tau)$, and $\widehat{\WO}^{(\text{ctw})}(\tau)$.

The kernel in \eqref{eq:ctw-win} has a direct pairwise interpretation. The lower priority comparison becomes observable once the control individual's component-$q$ event is observed at $t=\widetilde Y_{qj}(\tau)$ and the treated individual remains under observation beyond that time. The inverse censoring weight therefore corrects for observing the lower priority comparison at its resolving time $t$, rather than requiring both individuals to be followed completely through $\tau$. The factor $\overline\delta_{qi}(\tau)\overline\delta_{qj}(\tau)$ removes pairs for which a higher priority event is already known to have occurred, since such pairs cannot descend to component $q$ under the hierarchy. For the remaining pairs, the genuine tie indicator $\prod_{k<q}\Ind{T_{ki}>\tau,T_{kj}>\tau}$ is not always observable. The proposed estimator replaces this latent indicator by the conditional probability represented by $\mathcal R_{q,1}^{>}(u,t\mid \bm{Z})$ and $\mathcal R_{q,0}^{=}(u,t\mid \bm{Z})$. This probability equals one when the higher priority genuine tie is verified through $\tau$, lies between zero and one when the tie is censoring-induced, and is zeroed out by $\overline\delta_{qi}(\tau)\overline\delta_{qj}(\tau)$ when a higher priority event is observed. Thus the estimator preserves the same restricted-time hierarchy but allows partially observed lower priority comparisons to contribute fractionally. Consistency of \eqref{eq:ctw-win} is proved in \ref{supp:consistent_proposed}. For the first-priority component, the proposed estimator and \eqref{eq:cuiA} are identical because there are no higher priority components. If there is no censoring before \(\tau\), then \(u_{iq}=\tau\) whenever \(\overline\delta_{qi}(\tau)=1\), the conditional tie ratios equal one, and the proposed estimator reduces to the usual uncensored U-statistic. The key statistical distinction is that lower priority comparison ascertainment and higher priority descent eligibility are separated. In the IPCW estimator for genuine time in \eqref{eq:cuiA}, these two requirements are tied to the same horizon: a pair can contribute to component \(q\) only if the component-\(q\) comparison is observed and all higher priority ties can be verified through \(\tau\). The proposed estimator uses a less restrictive observation time for the lower priority comparison, namely the time at which the component-\(q\) ordering is resolved, while retaining the restricted horizon \(\tau\) for the target hierarchy. Thus, the lower priority event ordering is not imputed after censoring; it must already be observed. What is modeled is only the remaining eligibility to descend through the higher priority components. This separation is important because it preserves the clinical priority structure of the estimand while avoiding the unnecessary loss of lower priority information caused by requiring complete higher priority follow-up through \(\tau\). It also clarifies the expected efficiency pattern: the proposed estimator should behave similarly to the IPCW estimator in \eqref{eq:cuiA} when most descending pairs have verified higher priority ties, and should gain precision when many lower priority comparisons are resolved before censoring but their higher priority descent status remains unresolved at \(\tau\).

The proposed estimator is implementable with standard survival-modeling tools. It requires the same censoring model $\widehat G_a$ as in \eqref{eq:cuiA}, together with an arm-specific event-time model used only to evaluate the conditional tie probabilities in \eqref{eq:Rq}. We may use a survival-copula construction because it separates marginal event-risk modeling from the dependence model among endpoint components. For each treatment arm $a\in\{0,1\}$ and component $r\in\{1,\ldots,Q\}$, let \(S_{r,a}(s\mid\bz)=\Pr(T_r>s\mid A=a,\bZ=\bz)\) denote the marginal survival function. For a current priority level $q\ge2$, let $C_{q,a}(\cdot;\theta_{q,a})$ be a $q$-dimensional survival copula. We model the joint survival function in \eqref{eq:jointSq} by
\[
\mathcal{S}_{q,a}(u,t\mid\bz)=C_{q,a}\{S_{1,a}(u\mid\bz),\ldots,S_{q-1,a}(u\mid\bz),S_{q,a}(t\mid\bz);\theta_{q,a}\}.\]
If $\partial_q C_{q,a}$ denotes the partial derivative of the copula with respect to its $q$th argument and $f_{q,a}(t\mid\bz)=-\partial S_{q,a}(t\mid\bz)/\partial t$, then
\[
\mathcal{H}_{q,a}(u,t\mid\bz)=\partial_q C_{q,a}\{S_{1,a}(u\mid\bz),\ldots,S_{q-1,a}(u\mid\bz),S_{q,a}(t\mid\bz);\theta_{q,a}\}f_{q,a}(t\mid\bz).
\]
Substitution into \eqref{eq:Rq} gives
\begin{align}
\mathcal{R}_{q,a}^{>}(u,t\mid\bz)&=\frac{C_{q,a}\{S_{1,a}(\tau\mid\bz),\ldots,S_{q-1,a}(\tau\mid\bz),S_{q,a}(t\mid\bz);\theta_{q,a}\}}{C_{q,a}\{S_{1,a}(u\mid\bz),\ldots,S_{q-1,a}(u\mid\bz),S_{q,a}(t\mid\bz);\theta_{q,a}\}},
\label{eq:copula-Rgt-main}
\\
\mathcal{R}_{q,a}^{=}(u,t\mid\bz)&=\frac{\partial_q C_{q,a}\{S_{1,a}(\tau\mid\bz),\ldots,S_{q-1,a}(\tau\mid\bz),S_{q,a}(t\mid\bz);\theta_{q,a}\}}{\partial_q C_{q,a}\{S_{1,a}(u\mid\bz),\ldots,S_{q-1,a}(u\mid\bz),S_{q,a}(t\mid\bz);\theta_{q,a}\}}.
\label{eq:copula-Req-main}
\end{align}
The marginal density $f_{q,a}(t\mid\bz)$ cancels from \eqref{eq:copula-Req-main}. This cancellation is useful in practice because Cox proportional hazards margins can be used without estimating a smooth event-time density.

In the semiparametric implementation, the component-specific margins can be estimated by separate Cox proportional hazards models within each arm,
\[
\lambda_{r,a}(s\mid\bz)=\lambda_{0r,a}(s)\exp(\bm\beta_{r,a}^{\top}\bz),\qquad
\widehat S_{r,a}(s\mid\bz)=\exp\{-\widehat\Lambda_{0r,a}(s)\exp(\widehat{\bm\beta}_{r,a}^{\top}\bz)\},
\]
where $\widehat{\bm\beta}_{r,a}$ is obtained from partial likelihood and $\widehat\Lambda_{0r,a}$ is the Breslow estimator. After the marginal models are fitted, the copula parameter is estimated by pseudo-likelihood. In the bivariate case, define $x_{ri}=\widetilde Y_{ri}(\tau)$, $d_{ri}=\delta_{ri}(\tau)$, and $\widehat U_{ri}=\widehat S_{r,a}(x_{ri}\mid\bZ_i)$ for $r=1,2$. Let $c_\theta(u,v)=\partial^2 C_\theta(u,v)/\partial u\,\partial v$, and let $C_{\theta,1}$ and $C_{\theta,2}$ denote the first-order partial derivatives. Treating the fitted margins as fixed, the arm-specific copula pseudo log-likelihood is
\begin{align}
\ell_a(\theta)=
\sum_{i:A_i=a}
\bigl\{
&d_{1i}d_{2i}\log c_\theta(\widehat U_{1i},\widehat U_{2i})
+d_{1i}(1-d_{2i})\log C_{\theta,1}(\widehat U_{1i},\widehat U_{2i}) \nonumber\\
&+(1-d_{1i})d_{2i}\log C_{\theta,2}(\widehat U_{1i},\widehat U_{2i})
+(1-d_{1i})(1-d_{2i})\log C_\theta(\widehat U_{1i},\widehat U_{2i})
\bigr\}.
\label{eq:copula-pseudolik-main}
\end{align}
We estimate $\theta_a$ by maximizing \eqref{eq:copula-pseudolik-main} over the admissible parameter space of the selected copula family. Once $\widehat S_{1,a}$, $\widehat S_{2,a}$, and $\widehat\theta_a$ are available, the pairwise conditional tie probabilities are computed by direct substitution into \eqref{eq:copula-Rgt-main} and \eqref{eq:copula-Req-main}. Table \ref{tab:copula-main} summarizes the four bivariate copula families used in the simulations. Each family has closed-form expressions for both $C_\theta(u,v)$ and $C_{\theta,2}(u,v)$, so that the conditional tie probabilities can be evaluated without numerical differentiation.

For $q>2$, the same construction applies with a $q$-dimensional copula. The fitted model only needs to provide $C_{q,a}$ and $\partial_q C_{q,a}$ at the relevant marginal survival values, so the conditional tie win weighting structure of \eqref{eq:ctw-win} is unchanged. In numerical implementation, fitted marginal survival probabilities are truncated to a small interval $[\epsilon,1-\epsilon]$ before evaluating logarithms and copula derivatives. The fitted ratios $\widehat{\mathcal R}_{q,a}^{>}$ and $\widehat{\mathcal R}_{q,a}^{=}$ are conditional probabilities and should lie in $[0,1]$ under a valid fitted joint distribution; small numerical deviations can be truncated back to this range. The identities $\mathcal R_{q,a}^{>}=1$ and $\mathcal R_{q,a}^{=}=1$ when $u=\tau$ provide simple diagnostic checks. When a higher priority event is observed before $\tau$, the factor $\overline\delta_{qi}(\tau)$ in \eqref{eq:ctw-win} is zero, so the assigned value of $u_{iq}$ has no effect on the estimator. These implementation checks are useful in practice because the fitted event time model affects only the conditional descent probabilities, while the observed lower priority event ordering and the weighting structure of \eqref{eq:ctw-win} remain unchanged. Additional copula families, influence-function details, and software implementation are given in \ref{supp:copula} and \ref{supp:ifdetails}.

%
%

\begin{table}[htbp]
\centering
\caption{Bivariate copula families used to model the joint survival function in the proposed estimator. Here $C_{2,\theta}(u,v)=\partial C_\theta(u,v)/\partial v$ is the partial derivative needed for $\mathcal{R}_{2,a}^{=}(u,t\mid\bZ)$.}
\label{tab:copula-main}
\resizebox{0.65\textwidth}{!}{%
\begin{tabular}{lll}
\toprule
Family & Copula $C_\theta(u,v)$ & Partial derivative $C_{2,\theta}(u,v)$ \\
\midrule
Gumbel--Hougaard
&\shortstack[l]{
$\exp\{-A^{1/\theta}\}$, $\theta\ge 1$\\
$A=(-\log u)^\theta+(-\log v)^\theta$}
&
$C_\theta(u,v)A^{1/\theta-1}(-\log v)^{\theta-1}/v$
\\[3mm]
Clayton
&\shortstack[l]{
$B^{-1/\theta}$, $\theta>0$\\
$B=u^{-\theta}+v^{-\theta}-1$}
&
$v^{-\theta-1}B^{-1/\theta-1}$
\\[3mm]
Frank
&
\shortstack[l]{$-\theta^{-1}\log(B)$, $\theta\ne0$\\
$B=1+\dfrac{(e^{-\theta u}-1)(e^{-\theta v}-1)}{e^{-\theta}-1}$}
&
$\dfrac{e^{-\theta v}(e^{-\theta u}-1)}{(e^{-\theta}-1)B}$
\\[4mm]
Plackett
&\shortstack[l]{
$\dfrac{B-D^{1/2}}{2(\psi-1)}$, $\psi>0$, $\psi\ne1$\\
$B=1+(\psi-1)(u+v)$\\
$D=B^2-4\psi(\psi-1)uv$}
&
$\dfrac{1}{2}\left\{1-\dfrac{B-2\psi u}{D^{1/2}}\right\}$
\\
\bottomrule
\end{tabular}%
}
\end{table}

\section{Asymptotic properties}\label{sec:asymptotics}

Both the estimator of \eqref{eq:cuiA} and the proposed conditional tie weighting estimator \eqref{eq:ctw-win} are two sample U-statistics with estimated nuisance functions, and in what follows, we treat them in a unified way. Let $\nu\in\{\text{ipcw},\text{ctw}\}$ denote the two estimators and write
\[
  \widehat{\bpi}^{(\nu)}(\tau)
  =\bigl\{\widehat{\pi}_t^{(\nu)}(\tau),\,\widehat{\pi}_c^{(\nu)}(\tau)\bigr\}^\top
  =(n_1n_0)^{-1}\sum_{i=1}^{n_1}\sum_{j=1}^{n_0}\bm{H}^{(\nu)}(O_{1i},O_{0j};\widehat{\eta}_\nu(\cdot)),
\]
where $O_{ai}=\{\bYobs_i(\tau),\bDelta_i(\tau),\bZ_i\}$ collects the observed record for individual $i$ in arm $a$, with $\bYobs_i(\tau)=\{\widetilde{Y}_{1i}(\tau),\ldots,\widetilde{Y}_{Qi}(\tau)\}$ the vector of component wise observed times and $\bDelta_i(\tau)=\{\delta_{1i}(\tau),\ldots,\delta_{Qi}(\tau)\}$ the corresponding vector of event indicators. The comparison kernel $\bm{H}^{(\nu)}(\cdot)=\{K^{(\nu)}(\cdot),L^{(\nu)}(\cdot)\}^\top$ collects the pairwise win kernel $K^{(\nu)}$ and loss kernel $L^{(\nu)}$. The nuisance functions are $\eta_{\text{ipcw}}(\cdot)=\{G_1(\cdot),G_0(\cdot)\}$ for the IPCW estimator and \(\eta_{\text{ctw}}(\cdot)= \{G_1(\cdot),G_0(\cdot),\{\mathcal{S}_{q,a}(\cdot),\mathcal{H}_{q,a}(\cdot)\}_{q\ge2,\,a=0,1}\}\) for the conditional tie-weighting estimator, with $\eta_{\nu0}(\cdot)$ denoting the corresponding true nuisance functions.We assume the following regularity conditions.
\begin{enumerate}
\renewcommand{\theenumi}{A\arabic{enumi}}%
\renewcommand{\labelenumi}{(\theenumi)}%
\item \label{ass:A1} The observations $\{O_{1i}\}_{i=1}^{n_1}$ and $\{O_{0j}\}_{j=1}^{n_0}$ are i.i.d.\ within each arm, the two arms are independent, and $n_1/n\to\rho\in(0,1)$ as $n=n_1+n_0\to\infty$.
\item \label{ass:A2} $C_i\perp\bT_i\mid(A_i,\bZ_i)$, and there exists $\epsilon_G>0$ such that $G_a(t\mid\bZ)\ge\epsilon_G$ a.s.\ for $a\in\{0,1\}$ and $0\le t\le\tau$. For the proposed estimator, $\mathcal{S}_{q,a}(u,t\mid\bZ)\ge\epsilon_S>0$ and $\mathcal{H}_{q,a}(u,t\mid\bZ)\ge\epsilon_H>0$ a.s.\ on $0\le t\le u\le\tau$ for $q\ge2$ and $a=0,1$.
\item \label{ass:A3} $\E\bigl[\bigl\|\bm{H}^{(\nu)}(O_1,O_0;\eta_{\nu0}(\cdot))\bigr\|^2\bigr]<\infty$.
\item \label{ass:A4} For each $a\in\{0,1\}$, $\sqrt{n_a}\bigl\{\widehat{\eta}_{a,\nu}(\cdot)-\eta_{a,\nu0}(\cdot)\bigr\}=n_a^{-1/2}\sum_{k=1}^{n_a}\kappa_{a,\nu}(O_{ak})+o_p(1)$ for some mean zero, square integrable influence function $\kappa_{a,\nu}(\cdot)$ of the nuisance estimator.
\item \label{ass:A5} $\Psi_\nu\{\eta_\nu(\cdot)\}=\E\{\bm{H}^{(\nu)}(O_1,O_0;\eta_\nu(\cdot))\}$ is Fr\'{e}chet differentiable at $\eta_{\nu0}(\cdot)$ with derivative $\Gamma_\nu(\cdot)$, and $\Psi_\nu\{\widehat{\eta}_\nu(\cdot)\}-\Psi_\nu\{\eta_{\nu0}(\cdot)\}-\Gamma_\nu\{\widehat{\eta}_\nu(\cdot)-\eta_{\nu0}(\cdot)\}=o_p(n^{-1/2})$.
\end{enumerate}

Condition (\ref{ass:A1}) is the usual independent and identical distribution assumption to describe the sampling model. Condition (\ref{ass:A2}) imposes covariate-dependent censoring together with the positivity assumption required for weighting. The positivity on $G_a(\cdot)$ bounds the inverse probability of censoring weights by $\epsilon_G^{-2}$, and the additional positivity on $\mathcal{S}_{q,a}(\cdot)$ and $\mathcal{H}_{q,a}(\cdot)$ introduced by the proposed weighting ensures the conditional tie ratios in \eqref{eq:Rq} are well defined and lie in $[0,1]$, making (\ref{ass:A3}) bounded. Condition (\ref{ass:A4}) holds for the Kaplan-Meier estimator under independent censoring, and for a correctly specified Cox model under covariate-dependent censoring. Condition \ref{ass:A5} assumes that the plug-in remainder from the smooth nuisance functional is asymptotically negligible. Under these mild regularity conditions, we then obtain the following result that details the asymptotic property of the restricted-time win estimator.

\begin{theorem}\label{thm:asymp}
Under regularity conditions \ref{ass:A1}--\ref{ass:A5}, for each $\nu\in\{\text{ipcw},\text{ctw}\}$,
\begin{equation*}
  \sqrt{n}\bigl\{\widehat{\bpi}^{(\nu)}(\tau)-\bpi(\tau)\bigr\}=\frac{\sqrt{n}}{n_1}\sum_{i=1}^{n_1}\bm{\psi}_{1,\nu}(O_{1i})+\frac{\sqrt{n}}{n_0}\sum_{j=1}^{n_0}\bm{\psi}_{0,\nu}(O_{0j})+o_p(1),
\end{equation*}
where the individual-level influence function \(\bm{\psi}_{a,\nu}(O_{ai})=\bm{\xi}_{a,\nu}(O_{ai})+\bm{r}_{a,\nu}(O_{ai})\) decomposes the contribution of individual \(i\) in arm \(a\) into the first-order Hoeffding projection \(\bm{\xi}_{1,\nu}(o)=\E\{\bm{H}^{(\nu)}(o,O_0;\eta_{\nu0})\}-\bpi(\tau)\) and \(\bm{\xi}_{0,\nu}(o)=\E\{\bm{H}^{(\nu)}(O_1,o;\eta_{\nu0})\}-\bpi(\tau) \), and the nuisance-estimation correction \(\bm{r}_{a,\nu}(O_{ai})=\Gamma_{a,\nu}\{\kappa_{a,\nu}(O_{ai})\}\). Thus, 
\[
\sqrt{n}\{\widehat{\bpi}^{(\nu)}(\tau)-\bpi(\tau)\}
\xrightarrow{d}
\mathcal{N}(\bm{0},\bOmega_\nu),
\]
where \(\bOmega_\nu=\rho^{-1}\Var\{\bm{\psi}_{1,\nu}(O_1)\}+(1-\rho)^{-1}\Var\{\bm{\psi}_{0,\nu}(O_0)\}\).
\end{theorem}

The proof of Theorem \ref{thm:asymp} is provided in \ref{supp:asymp}. The Hoeffding decomposition \citep{lee2019u} reduces the oracle U-statistic with $\eta_{\nu0}(\cdot)$ in place of $\widehat{\eta}_\nu(\cdot)$ to the individual-level projection $\bm{\xi}_{a,\nu}(\cdot)$ plus a degenerate second-order remainder, and the differentiability condition in \ref{ass:A5} together with \ref{ass:A4} converts the nuisance plug-in error into the linear correction $\bm{r}_{a,\nu}(O_{ai})=\Gamma_{a,\nu}\{\kappa_{a,\nu}(O_{ai})\}$. Because each individual contributes to many treated-control pairs, the relevant object for variance estimation is this individual-level influence function rather than the variance of the $n_1n_0$ pairwise terms. For the IPCW estimator, $\eta_{\text{ipcw}}(\cdot)$ contains only the censoring survival functions, so the nuisance correction is driven entirely by estimation of $G_a(\cdot)$: \(\bm r_{a,\text{ipcw}}(\cdot)= \bm r^{G}_{a,\text{ipcw}}(\cdot)\). For the conditional tie-weighting estimator, $\eta_{\text{ctw}}(\cdot)$ contains both the censoring survival functions and the joint event-time quantities used to estimate the conditional tie probabilities. The nuisance correction therefore splits as \(\bm r_{a,\text{ctw}}(\cdot)=\bm r^{G}_{a,\text{ctw}}(\cdot)+\bm r^{E}_{a,\text{ctw}}(\cdot)\), where $\bm r^{G}_{a,\text{ctw}}(\cdot)$ comes from estimating $G_a(\cdot)$ and $\bm r^{E}_{a,\text{ctw}}(\cdot)$ comes from estimating $\{\mathcal{S}_{q,a}(\cdot),\mathcal{H}_{q,a}(\cdot)\}_{q\ge2}$. This decomposition is useful for understanding the cost of recovering partially observed lower priority information. The term $\bm{\xi}_{a,\nu}$ is the variability that would remain if all nuisance functions were known; the censoring correction is common to both estimators; and the event-model correction appears only for conditional tie weighting. Thus, the conditional tie weighting estimator trades additional nuisance-model variability for a larger set of contributing lower priority comparisons. This decomposition is useful for understanding the cost of recovering partially observed lower priority information. The term \(\bm{\xi}_{a,\nu}\) is the variability that would remain even if all nuisance functions were known; it reflects the fact that each enrolled individual is paired with many individuals in the opposite arm. The censoring correction \(\bm r^G_{a,\nu}\) is common to both estimators and reflects uncertainty in the inverse probability weights. The additional term \(\bm r^E_{a,I}\) appears only for the proposed estimator and represents the price paid for estimating the conditional tie probabilities. This term is driven by lower priority comparisons whose contribution depends on the fitted higher priority descent probability. It is absent for the first priority component and becomes small when most lower priority comparisons have verified ties from higher priority. Thus, the proposed estimator trades additional nuisance-model variability for a larger set of contributing pairwise comparisons.

For implementation of the variance estimator, it is useful to write the empirical influence-function components explicitly. Let $\widehat K_{ij,q}^{(\nu)}$ and $\widehat L_{ij,q}^{(\nu)}$ denote the fitted component-$q$ win and loss kernels for estimator $\nu\in\{\text{ipcw},\text{ctw}\}$, so that \(\widehat K_{ij}^{(\nu)}=\sum_{q=1}^{Q}\widehat K_{ij,q}^{(\nu)}\), \(\widehat L_{ij}^{(\nu)}= \sum_{q=1}^{Q}\widehat L_{ij,q}^{(\nu)}\). The empirical Hoeffding projections are computed by pairing each individual with all individuals in the opposite arm:
\[
\widehat{\bm\xi}_{1,\nu,i}=
\begin{pmatrix}
n_0^{-1}\sum_{j=1}^{n_0}\widehat K_{ij}^{(\nu)}-\widehat\pi_t^{(\nu)}(\tau)\\[2mm]
n_0^{-1}\sum_{j=1}^{n_0}\widehat L_{ij}^{(\nu)}-\widehat\pi_c^{(\nu)}(\tau)
\end{pmatrix},
\qquad
\widehat{\bm\xi}_{0,\nu,j}=
\begin{pmatrix}
n_1^{-1}\sum_{i=1}^{n_1}\widehat K_{ij}^{(\nu)}-\widehat\pi_t^{(\nu)}(\tau)\\[2mm]
n_1^{-1}\sum_{i=1}^{n_1}\widehat L_{ij}^{(\nu)}-\widehat\pi_c^{(\nu)}(\tau)
\end{pmatrix}.
\]
These terms capture the two-sample U-statistic variation that would remain if the nuisance functions were known. The censoring-model correction has the same form for the two estimators, although the kernels and censoring time arguments differ. Let $\widehat s_{1,ij,q}^{W}$ and $\widehat s_{0,ij,q}^{W}$ be the time arguments at which $\widehat G_1$ and $\widehat G_0$ are evaluated in the component-$q$ win kernel, and define $\widehat s_{1,ij,q}^{L}$ and $\widehat s_{0,ij,q}^{L}$ analogously for the loss kernel. For example, for a lower-priority treated win under CTW, the censoring weights are evaluated at the observed resolving time $\widetilde Y_{qj}(\tau)$, whereas under the IPCW estimator in \eqref{eq:cuiA}, lower-priority weights are evaluated at $\tau$. Let $\widehat\kappa_{ai,G}(s\mid\bz)$ denote the estimated influence function of $\widehat G_a(s\mid\bz)$ for subject $i$ in arm $a$. The treatment-arm censoring correction is
\[
\widehat{\bm r}^{G}_{1,\nu,i}=-\frac{1}{n_0}\sum_{j=1}^{n_0}\sum_{q=1}^{Q}
\begin{pmatrix}
\widehat K_{ij,q}^{(\nu)}
\dfrac{\widehat\kappa_{1i,G}(\widehat s_{1,ij,q}^{W}\mid\bZ_i)}
{\widehat G_1(\widehat s_{1,ij,q}^{W}\mid\bZ_i)}
\\[4mm]
\widehat L_{ij,q}^{(\nu)}
\dfrac{\widehat\kappa_{1i,G}(\widehat s_{1,ij,q}^{L}\mid\bZ_i)}
{\widehat G_1(\widehat s_{1,ij,q}^{L}\mid\bZ_i)}
\end{pmatrix},
\]
and the control-arm censoring correction is
\[
\widehat{\bm r}^{G}_{0,\nu,j}=-\frac{1}{n_1}\sum_{i=1}^{n_1}\sum_{q=1}^{Q}
\begin{pmatrix}
\widehat K_{ij,q}^{(\nu)}
\dfrac{\widehat\kappa_{0j,G}(\widehat s_{0,ij,q}^{W}\mid\bZ_j)}
{\widehat G_0(\widehat s_{0,ij,q}^{W}\mid\bZ_j)}
\\[4mm]
\widehat L_{ij,q}^{(\nu)}
\dfrac{\widehat\kappa_{0j,G}(\widehat s_{0,ij,q}^{L}\mid\bZ_j)}
{\widehat G_0(\widehat s_{0,ij,q}^{L}\mid\bZ_j)}
\end{pmatrix}.
\]
For the conditional tie weighting estimator, an additional event-model correction appears because the lower-priority kernels depend on the fitted conditional tie probabilities. For $q\ge2$, define
\[
\widehat\Delta^{>}_{ai,q}(u,t)=\frac{\widehat\kappa_{ai,S,q}(\tau,t,\bZ_i)}
{\widehat{\mathcal S}_{q,a}(\tau,t\mid\bZ_i)}-
\frac{\widehat\kappa_{ai,S,q}(u,t,\bZ_i)}{\widehat{\mathcal S}_{q,a}(u,t\mid\bZ_i)},
\]
and
\[
\widehat\Delta^{=}_{ai,q}(u,t)=\frac{\widehat\kappa_{ai,H,q}(\tau,t,\bZ_i)}{\widehat{\mathcal H}_{q,a}(\tau,t\mid\bZ_i)}-\frac{\widehat\kappa_{ai,H,q}(u,t,\bZ_i)}{\widehat{\mathcal H}_{q,a}(u,t\mid\bZ_i)},
\]
where $\widehat\kappa_{ai,S,q}$ and $\widehat\kappa_{ai,H,q}$ are the estimated influence functions of the fitted joint survival function $\widehat{\mathcal S}_{q,a}$ and subdensity $\widehat{\mathcal H}_{q,a}$. These two quantities are the empirical log-ratio perturbations of $\mathcal R_{q,a}^{>}$ and $\mathcal R_{q,a}^{=}$. The treatment-arm event model correction for conditional tie weighting is
\[
\widehat{\bm r}^{E}_{1,\text{ctw},i}=\frac{1}{n_0}\sum_{j=1}^{n_0}\sum_{q=2}^{Q}
\begin{pmatrix}
\widehat K_{ij,q}^{(\text{ctw})}
\widehat\Delta^{>}_{1i,q}\{u_{iq},\widetilde Y_{qj}(\tau)\}
\\[2mm]
\widehat L_{ij,q}^{(\text{ctw})}
\widehat\Delta^{=}_{1i,q}\{u_{iq},\widetilde Y_{qi}(\tau)\}
\end{pmatrix},
\]
and the control-arm event-model correction is
\[
\widehat{\bm r}^{E}_{0,\text{ctw},j}=\frac{1}{n_1}\sum_{i=1}^{n_1}\sum_{q=2}^{Q}
\begin{pmatrix}
\widehat K_{ij,q}^{(\text{ctw})}
\widehat\Delta^{=}_{0j,q}\{u_{jq},\widetilde Y_{qj}(\tau)\}
\\[2mm]
\widehat L_{ij,q}^{(\text{ctw})}
\widehat\Delta^{>}_{0j,q}\{u_{jq},\widetilde Y_{qi}(\tau)\}
\end{pmatrix}.
\]
The roles of $>$ and $=$ are reversed between the win and loss coordinates because, for a treated win, the control individual has the component-$q$ event at the resolving time, whereas for a treated loss, the treated individual has the component-$q$ event at the resolving time. Combining these components, the empirical influence functions used for sandwich variance estimation are \(\widehat{\bm\psi}_{a,\text{ipcw},i}=\widehat{\bm\xi}_{a,\text{ipcw},i}+\widehat{\bm r}^{G}_{a,\text{ipcw},i}\), \(\widehat{\bm\psi}_{a,\text{ctw},i}=\widehat{\bm\xi}_{a,\text{ctw},i}+
\widehat{\bm r}^{G}_{a,\text{ctw},i}+\widehat{\bm r}^{E}_{a,\text{ctw},i}\). Equivalently, in win-loss coordinate form,
\[
\widehat{\bm\psi}_{a,\text{ctw},i}=
\begin{pmatrix}
\widehat\xi^{W}_{a,\text{ctw},i}
+
\widehat r^{G,W}_{a,\text{ctw},i}
+
\widehat r^{E,W}_{a,\text{ctw},i}
\\[1mm]
\widehat\xi^{L}_{a,\text{ctw},i}
+
\widehat r^{G,L}_{a,\text{ctw},i}
+
\widehat r^{E,L}_{a,\text{ctw},i}
\end{pmatrix}.
\]
The detailed expression are provided in \ref{supp:asymp}.
 
Since the win ratio, net benefit, and win odds are smooth functions of $(\pi_t,\pi_c)$, we use the delta method to obtain their asymptotic distribution, as below.

\begin{corollary}\label{cor:delta}
Under the regularity conditions assumed in Theorem \ref{thm:asymp}, if $\pi_t(\tau)>0$, $\pi_c(\tau)>0$, and $|\pi_t(\tau)-\pi_c(\tau)|<1$, then for each $\nu\in\{\text{ipcw},\text{ctw}\}$, we have
\begin{align*}
  \sqrt{n}\bigl\{\log\widehat{\WR}^{(\nu)}(\tau)-\log\WR(\tau)\bigr\}
  &\xrightarrow{d}\mathcal{N}(0,\bm{g}_{\WR}^\top\bOmega_\nu\bm{g}_{\WR}),\\
  \sqrt{n}\bigl\{\widehat{\NB}^{(\nu)}(\tau)-\NB(\tau)\bigr\}
  &\xrightarrow{d}\mathcal{N}(0,\bm{g}_{\NB}^\top\bOmega_\nu\bm{g}_{\NB}),\\
  \sqrt{n}\bigl\{\log\widehat{\WO}^{(\nu)}(\tau)-\log\WO(\tau)\bigr\}
  &\xrightarrow{d}\mathcal{N}(0,\bm{g}_{\WO}^\top\bOmega_\nu\bm{g}_{\WO}),
\end{align*}
where $\bm{g}_{\WR}=\{1/\pi_t(\tau),-1/\pi_c(\tau)\}^\top$, $\bm{g}_{\NB}=(1,-1)^\top$, and $\bm{g}_{\WO}=2\{1-(\pi_t(\tau)-\pi_c(\tau))^2\}^{-1}(1,-1)^\top$.
\end{corollary}
The sandwich covariance estimator is obtained by taking the empirical second moment of the estimated individual-level influence functions:
\begin{equation} \label{eq:Omega-hat}
   \widehat{\bOmega}_\nu
  =\frac{n}{n_1^2}\sum_{i=1}^{n_1}\widehat{\bm{\psi}}_{1,\nu,i}\widehat{\bm{\psi}}_{1,\nu,i}^\top
   +\frac{n}{n_0^2}\sum_{j=1}^{n_0}\widehat{\bm{\psi}}_{0,\nu,j}\widehat{\bm{\psi}}_{0,\nu,j}^\top.   
\end{equation}
Under the regularity conditions of Theorem \ref{thm:asymp} and consistency of the plug-in nuisance influence functions, $\widehat{\bOmega}_\nu$ consistently estimates $\bOmega_\nu$. For each win summary estimator in Corollary \ref{cor:delta}, the estimated standard error is $\{\bm{g}_h^\top\widehat{\bOmega}_\nu\bm{g}_h/n\}^{1/2}$, and Wald intervals follow under normal approximation. Because $\WR(\tau)$ and $\WO(\tau)$ are positive and right skewed in finite samples, we construct their intervals on the log scale before exponentiation; $\NB(\tau)$ is reported directly on its natural scale.

\section{Simulation studies} \label{sec:simulation}

We conducted extensive simulation studies to assess the finite sample performance of the proposed estimator and to compare it with the IPCW estimator of \citet{cui2025ipcw} under prioritized composite endpoints. The simulation includes two parts. The first part considers different level of nuisance model mis-specification, where data are generated from a Weibull proportional hazards marginal model linked by a Gumbel copula, and the fitted nuisance models are varied to assess the impact of misspecifying the censoring model, the event time model, or both. The second part is a copula sensitivity study, where the marginal event time model is assumed to be correctly specified by Cox proportional hazards models but the data generating and working copula families are varied. 

For the first part, we focus on a two component prioritized endpoint with death as the first priority event and a serious nonfatal event as the second priority event. We also conducted parallel simulations under a three component prioritized endpoint and those results are reported in Web Appendix \ref{supp:sim3}. In each simulated data set, each treatment arm includes $n_a=400$ individuals. We generated three baseline covariates independently as $Z_{1i}\sim\mathrm{Bernoulli}(0.5)$, $Z_{2i}\sim\mathrm{Uniform}(0,1)$, and $Z_{3i}\sim\mathrm{Bernoulli}(0.4)$. Conditional on treatment assignment $A_i$ and covariates $\bZ_i=(Z_{1i},Z_{2i},Z_{3i})^\top$, the latent event times followed Weibull proportional hazards models; that is, for $q=1,2$, 
\begin{equation}\label{eq:sim-weibull}
\Pr(T_{qi}>t\mid A_i,\bZ_i)=\exp\!\left\{-\lambda_{q0}t^{\rho_q}\exp(\beta_{q1}Z_{1i}+\beta_{q2}Z_{2i}+\beta_{q3}Z_{3i}+\beta_{qa}A_i)\right\},
\end{equation}
where $\lambda_{q0}$ and $\rho_q$ are the Weibull baseline scale and shape for component $q$, and $\beta_{q1},\beta_{q2},\beta_{q3},\beta_{qa}$ are the regression coefficients for the three covariates and the treatment indicator. Parameters were calibrated so that the first priority event (death) was relatively rare with a modest treatment effect, and the second priority event was more common with a larger effect, so the lower priority comparison contribute nontrivially. The full parameter set is given in \ref{supp:sim_additional}. Dependence between the two latent event times was generated using the Gumbel--Hougaard copula (Table \ref{tab:copula-main}) applied to the Weibull marginal survival distributions, where $u_1$ and $u_2$ are the two generic copula arguments lying in $[0,1]$. We considered weak dependence with $\theta=1.25$ and strong dependence with $\theta=4$. Censoring times were generated from a proportional hazards model 
\begin{equation}\label{eq:sim-censor}
\Pr(C_i>t\mid\bZ_i)=\exp\left\{-\lambda_C t\exp(0.80Z_{1i}+1.00Z_{2i}+0.65Z_{3i})\right\},
\end{equation}
where $\lambda_C$ is the baseline censoring rate. We calibrated $\lambda_C$ so that the overall censoring proportion by $\tau=36$ months was approximately $20\%$, $40\%$, $60\%$, or $80\%$, and we considered restriction times $\tau\in\{12,24,36\}$ months. Because the IPCW estimator of \citet{cui2025ipcw} requires only a censoring model, it is unaffected by misspecification of the event model. To examine sensitivity to nuisance model misspecification, we considered four working model configurations. Configuration $\mathcal{M}_1$ considered a Cox model for censoring and the correctly specified Weibull Gumbel event time model. Configuration $\mathcal{M}_2$ considered the same Cox censoring model but replaced the event time model by exponential margins with independence between event types. Configuration $\mathcal{M}_3$ specified a Kaplan-Meier censoring model that ignored covariate dependence together with the correctly specified Weibull Gumbel event time model. Configuration $\mathcal{M}_4$ misspecified both nuisance components by combining Kaplan-Meier censoring with independent exponential event times.

We then designed a copula sensitivity study to assess whether the efficiency gain from conditional tie recovery depends strongly on the chosen copula family when the marginal event time model is correctly specified. In this second study, the marginal event time distributions were generated from the same Weibull proportional hazards models in \eqref{eq:sim-weibull}. The data generating copula was varied over four one parameter families, including Gumbel--Hougaard, Clayton, Frank, and Plackett (see Table \ref{tab:copula-main}). For each generated data set, we fit the proposed estimator using Cox margins combined with each of these four working copula families. We also fit the IPCW estimator of \citet{cui2025ipcw}, which does not use an event time copula.

For each setting in both simulation studies, the true values of the component-specific win probabilities and the resulting summary measures were computed under the uncensored data generating mechanism by deterministic numerical integration. Letting $\pi_{tq}(\tau)$ and $\pi_{cq}(\tau)$ denote the treatment and control win probabilities contributed by the $q$th priority level, we obtained $\pi_t(\tau)=\sum_{q=1}^2\pi_{tq}(\tau)$, $\pi_c(\tau)=\sum_{q=1}^2\pi_{cq}(\tau)$, $\NB(\tau)=\pi_t(\tau)-\pi_c(\tau)$, $\WR(\tau)=\pi_t(\tau)/\pi_c(\tau)$, and $\WO(\tau)=\{1+\NB(\tau)\}/\{1-\NB(\tau)\}$. For each combination of simulation design factors, we generated 1{,}000 Monte Carlo data sets. We summarize performance by the relative bias (RBias), Monte Carlo standard deviation (MCSD), average estimated standard error (ASE), and empirical coverage probability (COV) of the nominal 95\% confidence interval. We report the relative efficiency $\mathrm{RE}=\mathrm{MCSD}^2\{\widehat{\NB}^{(\text{ipcw})}(\tau)\}/\mathrm{MCSD}^2\{\widehat{\NB}^{(\text{ctw})}(\tau)\}$ to compare two estimators.

Table \ref{tab:sim-NB80} summarizes the results for $\NB(\tau)$ under the highest censoring setting, where the overall censoring proportion by $\tau=36$ months is approximately 80\%. Results under 20\%, 40\%, and 60\% censoring, together with corresponding results for $\WR(\tau)$, $\WO(\tau)$ are reported in \ref{supp:sim_additional}. Under the correctly specified working configuration $\mathcal{M}_1$, the proposed estimator is consistently more efficient than the IPCW estimator of \citet{cui2025ipcw}. The efficiency gain is visible at $\tau=12$ and increases as the restriction time lengthens. For example, under weak dependence $(\theta=1.25)$, the relative efficiency under $\mathcal{M}_1$ increases from $1.62$ at $\tau=12$ to 2.52 at $\tau=24$ and 2.65 at $\tau=36$. Under strong dependence $(\theta=4.00)$, the corresponding values are $1.53$, $2.58$, and $2.90$. This pattern is consistent with the mechanism of the proposed estimator. When censoring is heavy and the restriction time is long, more lower priority comparisons are blocked by unverifiable higher priority ties, and conditional tie weighting recovers part of this otherwise discarded information.

\begin{table}[htbp]
\centering
\caption{Simulation results for $\NB(\tau)$ under 80\% censoring by $\tau=36$ for the two component prioritized endpoint.}
\label{tab:sim-NB80}
\resizebox{0.75\textwidth}{!}{%
\begin{tabular}{ccccrrrrrrrrr}
\toprule
 & & & & \multicolumn{4}{c}{$\widehat{\NB}^{(\text{ipcw})}(\tau)$} & \multicolumn{4}{c}{$\widehat{\NB}^{(\text{ctw})}(\tau)$} & \\
\cmidrule(lr){5-8}\cmidrule(lr){9-12}
$\theta$ & $\tau$ & $\mathcal{M}_k$ & True & RB\% & MCSD & ASE & Cov & RB\% & MCSD & ASE & Cov & RE \\
\midrule
\multirow{12}{*}{1.25}
& \multirow{4}{*}{12} & $\mathcal{M}_1$ & 0.078 & $-$0.5 & 0.0523 & 0.0504 & 0.944 & 0.9 & 0.0411 & 0.0402 & 0.951 & 1.62 \\
&  & $\mathcal{M}_2$ & 0.078 & $-$0.5 & 0.0523 & 0.0504 & 0.944 & 1.9 & 0.0417 & 0.0407 & 0.950 & 1.58 \\
&  & $\mathcal{M}_3$ & 0.078 & $-$6.0 & 0.0447 & 0.0433 & 0.942 & $-$1.9 & 0.0390 & 0.0377 & 0.936 & 1.31 \\
&  & $\mathcal{M}_4$ & 0.078 & $-$6.0 & 0.0447 & 0.0433 & 0.942 & $-$0.9 & 0.0394 & 0.0381 & 0.938 & 1.28 \\
\cmidrule(lr){2-13}
& \multirow{4}{*}{24} & $\mathcal{M}_1$ & 0.106 & $-$0.5 & 0.0993 & 0.0935 & 0.963 & 3.3 & 0.0626 & 0.0599 & 0.953 & 2.52 \\
&  & $\mathcal{M}_2$ & 0.106 & $-$0.5 & 0.0993 & 0.0935 & 0.963 & 6.4 & 0.0641 & 0.0614 & 0.955 & 2.40 \\
&  & $\mathcal{M}_3$ & 0.106 & $-$2.7 & 0.0640 & 0.0636 & 0.948 & 1.6 & 0.0513 & 0.0499 & 0.940 & 1.56 \\
&  & $\mathcal{M}_4$ & 0.106 & $-$2.7 & 0.0640 & 0.0636 & 0.948 & 4.5 & 0.0524 & 0.0510 & 0.945 & 1.49 \\
\cmidrule(lr){2-13}
& \multirow{4}{*}{36} & $\mathcal{M}_1$ & 0.111 & $-$6.7 & 0.1492 & 0.1427 & 0.966 & 3.1 & 0.0917 & 0.0826 & 0.958 & 2.65 \\
&  & $\mathcal{M}_2$ & 0.111 & $-$6.7 & 0.1492 & 0.1427 & 0.966 & 9.1 & 0.0939 & 0.0853 & 0.960 & 2.52 \\
&  & $\mathcal{M}_3$ & 0.111 & 0.3 & 0.0802 & 0.0812 & 0.951 & 2.8 & 0.0597 & 0.0593 & 0.947 & 1.81 \\
&  & $\mathcal{M}_4$ & 0.111 & 0.3 & 0.0802 & 0.0812 & 0.951 & 8.2 & 0.0615 & 0.0613 & 0.948 & 1.70 \\
\midrule
\multirow{12}{*}{4.00}
& \multirow{4}{*}{12} & $\mathcal{M}_1$ & 0.082 & $-$0.9 & 0.0477 & 0.0496 & 0.968 & $-$1.3 & 0.0386 & 0.0401 & 0.962 & 1.53 \\
&  & $\mathcal{M}_2$ & 0.082 & $-$0.9 & 0.0477 & 0.0496 & 0.968 & $-$0.7 & 0.0397 & 0.0411 & 0.959 & 1.44 \\
&  & $\mathcal{M}_3$ & 0.082 & $-$5.8 & 0.0416 & 0.0425 & 0.953 & $-$3.5 & 0.0363 & 0.0374 & 0.955 & 1.31 \\
&  & $\mathcal{M}_4$ & 0.082 & $-$5.8 & 0.0416 & 0.0425 & 0.953 & $-$2.8 & 0.0373 & 0.0383 & 0.956 & 1.24 \\
\cmidrule(lr){2-13}
& \multirow{4}{*}{24} & $\mathcal{M}_1$ & 0.117 & $-$5.5 & 0.0957 & 0.0921 & 0.967 & $-$4.1 & 0.0596 & 0.0603 & 0.956 & 2.58 \\
&  & $\mathcal{M}_2$ & 0.117 & $-$5.5 & 0.0957 & 0.0921 & 0.967 & $-$2.0 & 0.0630 & 0.0634 & 0.961 & 2.31 \\
&  & $\mathcal{M}_3$ & 0.117 & $-$6.5 & 0.0607 & 0.0628 & 0.949 & $-$4.1 & 0.0489 & 0.0499 & 0.947 & 1.54 \\
&  & $\mathcal{M}_4$ & 0.117 & $-$6.5 & 0.0607 & 0.0628 & 0.949 & $-$1.5 & 0.0516 & 0.0525 & 0.952 & 1.38 \\
\cmidrule(lr){2-13}
& \multirow{4}{*}{36} & $\mathcal{M}_1$ & 0.129 & $-$11.0 & 0.1479 & 0.1306 & 0.957 & $-$6.6 & 0.0869 & 0.0824 & 0.973 & 2.90 \\
&  & $\mathcal{M}_2$ & 0.129 & $-$11.0 & 0.1479 & 0.1306 & 0.957 & $-$3.3 & 0.0920 & 0.0873 & 0.971 & 2.58 \\
&  & $\mathcal{M}_3$ & 0.129 & $-$4.4 & 0.0788 & 0.0805 & 0.961 & $-$3.7 & 0.0576 & 0.0595 & 0.956 & 1.87 \\
&  & $\mathcal{M}_4$ & 0.129 & $-$4.4 & 0.0788 & 0.0805 & 0.961 & 0.4 & 0.0620 & 0.0639 & 0.958 & 1.61 \\
\bottomrule
\end{tabular}}
\end{table}

The comparison between $\mathcal{M}_1$ and $\mathcal{M}_2$ shows the effect of event time model misspecification when the censoring model remains correctly specified. Such misspecification attenuates the efficiency gain and can increase finite-sample bias, especially at longer restriction times, but the proposed estimator still remains more efficient than the IPCW estimator of \citet{cui2025ipcw} across all settings. For instance, at $\tau=36$, the relative efficiency under $\mathcal{M}_2$ is $2.52$ for $\theta=1.25$ and $2.58$ for $\theta=4.00$. Thus, even a misspecified working event-time model may still recover useful partial information, although less effectively than a correctly specified model. The comparison between $\mathcal{M}_3$ and $\mathcal{M}_4$ shows the additional impact of censoring model misspecification. Because both estimators rely on inverse probability of censoring weighting, misspecifying the censoring distribution can affect bias and coverage for both methods. However, the proposed estimator continues to provide a variance reduction relative to the existing IPCW estimator of \citet{cui2025ipcw}. At $\tau=36$, the relative efficiency remains 1.81 under $\mathcal{M}_3$ and 1.70 under $\mathcal{M}_4$ for $\theta=1.25$. These results show that the proposed estimator's main benefit is efficiency improvement from partially recovering lower priority comparisons, while accurate censoring-model specification remains important for bias performance.

Figure \ref{fig:copula-NB80} reports the copula sensitivity results for $\NB(\tau)$ under 80\% censoring across the three restriction times $\tau=24$, and $36$ months. In this part of the simulation, the event time marginal models are correctly specified through Cox proportional hazards models, while the data generating copula and the working copula are varied across Gumbel--Hougaard, Clayton, Frank, and Plackett. 
Across the four data generating copulas and the two restriction times, the proposed estimator has substantially smaller Monte Carlo variability than \citet{cui2025ipcw}. The RBias panels show that the proposed estimator is generally centered close to the truth across both correctly specified and misspecified working copulas, whereas estimator of \citet{cui2025ipcw} exhibits larger negative relative bias in several data generating settings, especially at the longer restriction times. The COV panels show that empirical coverage remains stable and close to the nominal 95\% level across all copula matrix. Within each data generating copula row, the four proposed method columns are visually similar across RBias, COV, and RE, even when the working copula differs from the data generating copula. With correctly specified Cox margins, the main efficiency improvement is less sensitive to the selection of the copula family. Additional copula sensitivity results under 40\% censoring and the corresponding results for $\WR(\tau)$ and $\WO(\tau)$ are provided in Web Appendix \ref{supp:sim_copula}.

\begin{figure}[htbp!]
\centering
\includegraphics[width=0.9\textwidth]{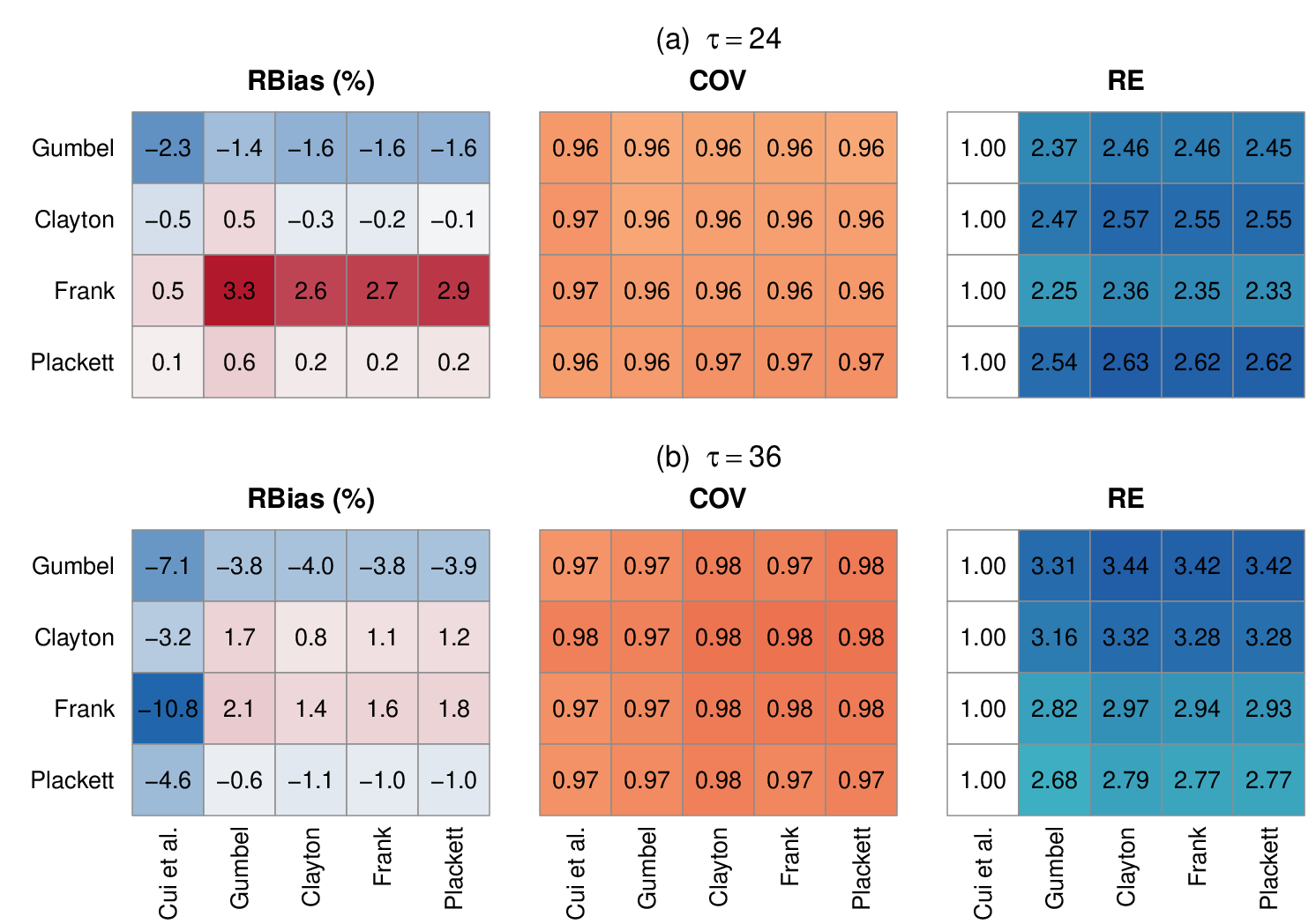}
\caption{Figure 2. Copula sensitivity of \(NB(\tau)\) under 80\% censoring. The two stacked blocks correspond to (a) \(\tau=24\) months and (b) \(\tau=36\) months. Within each block, rows indicate the data generating copula. Columns compare the IPCW estimator of \citet{cui2025ipcw} with the proposed estimator of \eqref{eq:ctw-win} fitted under each working copula. The panels report relative bias (RBias), empirical coverage of the nominal 95\% confidence interval (COV), and relative efficiency \(RE=\mathrm{MCSD}^2\{\widehat{NB}^{(\text{ipcw})}(\tau)\}/\mathrm{MCSD}^2\{\widehat{NB}^{(\text{ctw})}(\tau)\}\).}
\label{fig:copula-NB80}
\end{figure}

\section{Application to the HF-ACTION trial with a composite survival outcome} \label{sec:real}

We reanalyze the HF-ACTION trial using the estimand of the restricted-time win statistics defined in Section \ref{sec:method}. The design and primary clinical results of the trial have been reported previously in \citet{o2009efficacy}. The present analysis includes 2{,}130 patients after excluding patients with missing baseline covariates required for the censoring and event time models, with 1{,}060 assigned to aerobic exercise training added to usual care and 1{,}070 assigned to usual care alone. The hierarchical composite endpoint consists of all-cause death followed by time to first all-cause hospitalization. This ordering follows the clinical priority of mortality while allowing hospitalization to contribute when death does not distinguish a treated-control pair. The hospitalization burden in the analysis cohort was substantial, with a mean number of hospitalizations per patient of 2.01 in the usual-care arm and 2.00 in the exercise-training arm, supporting the relevance of the lower priority component. We vary the restriction time continuously over approximately 1 to 4 years of follow-up. At each restriction time, death is evaluated as the first priority component, and time to first all-cause hospitalization is evaluated as the second priority component only when the pair descends past death. Thus, a treated patient wins on hospitalization if the control patient's first hospitalization occurs earlier, provided the required genuine tie on death is established or, under the proposed estimator, partially recovered through conditional tie weighting. We estimate the restricted-time win and loss probabilities, together with the win ratio, net benefit, and win odds, using both the existing IPCW estimator and the proposed conditional tie weighting estimator. Within each arm, the censoring survival function and the marginal event time survival functions for death and hospitalization are estimated by Cox models adjusting for age, BMI, sex, systolic and diastolic blood pressure, history of angina, ischemic etiology, cardiopulmonary exercise test duration, atrial fibrillation or flutter, and diabetes. The fitted event time marginals are linked through a Gumbel copula with arm-specific dependence parameters estimated by maximum likelihood. Confidence intervals use the sandwich estimators from Section \ref{sec:asymptotics}, with win ratio and win odds intervals constructed on the log scale before exponentiation.

\begin{figure}[!ht]
\centering
\includegraphics[width=0.85\textwidth]{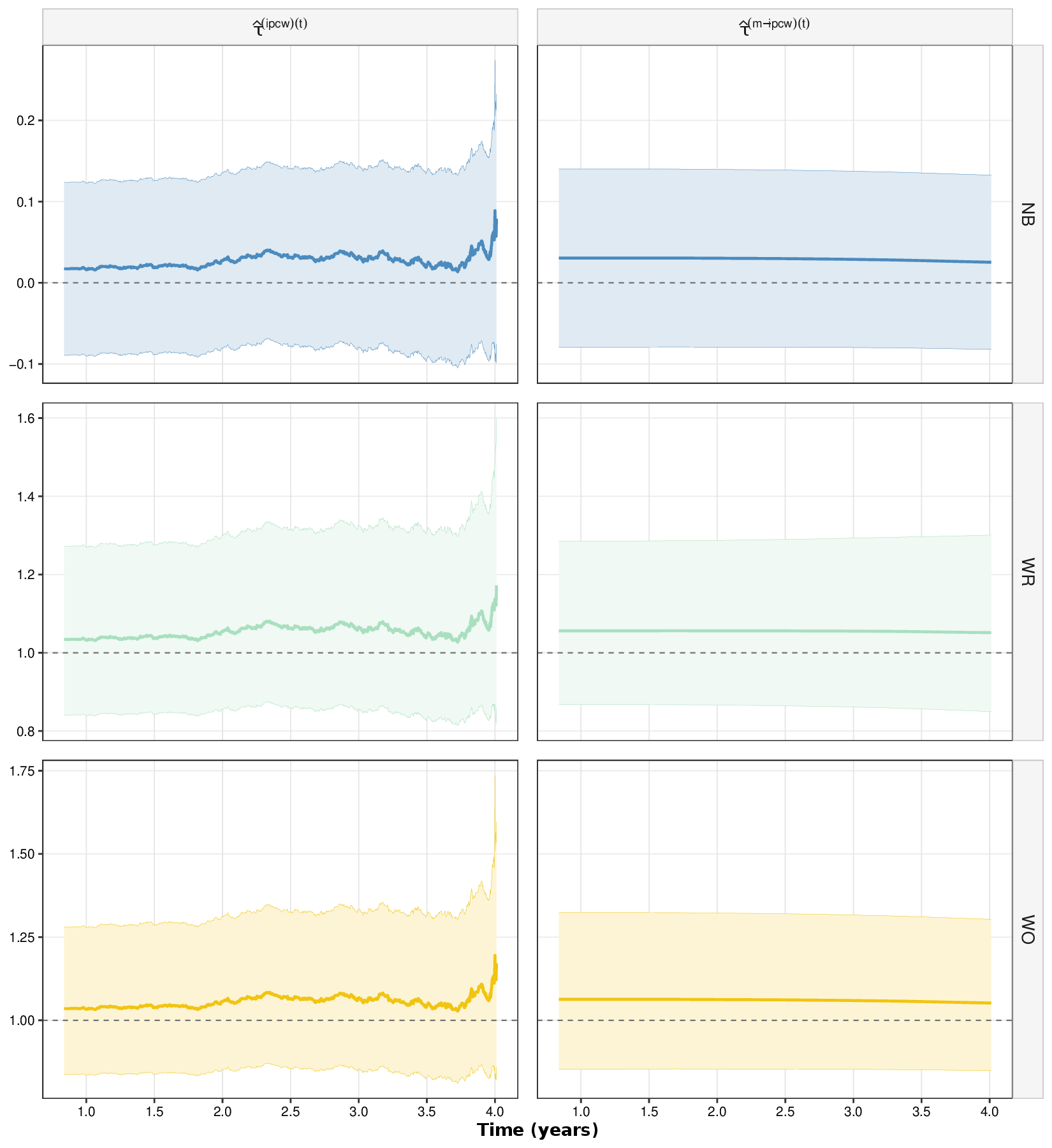}
\caption{Estimated net benefit (NB, top row), win ratio (WR, middle row), and win odds (WO, bottom row) as functions of the restriction time $\tau \in [1, 4]$ years for the HF-ACTION trial, comparing exercise training versus usual care on the
hierarchical composite of all-cause death and all-cause hospitalization. The left column shows the existing IPCW estimator $\widehat{\NB}^{(\text{ipcw})}(\tau)$ \citep{cui2025ipcw} and the right column shows the proposed estimator $\widehat{\NB}^{(\text{ctw})}(\tau)$. The shaded bands are pointwise 95\% confidence intervals.}
\label{fig:hfaction}
\end{figure}

Figure \ref{fig:hfaction} shows the estimated $\NB(\tau)$, $\WR(\tau)$, and $\WO(\tau)$ as functions of the restriction time $\tau \in [1, 4]$ years for both estimators. First, both estimators indicate a consistent, favorable treatment effect of exercise training across all restriction times, with the exercise-training arm accumulating more pairwise wins than the usual-care arm on death and on first hospitalization. Second, the proposed estimator $\widehat{\NB}^{(\text{ctw})}(\tau)$ produces notably narrower confidence bands than $\widehat{\NB}^{(\text{ipcw})}(\tau)$ throughout the entire range of $\tau$, reflecting the efficiency gain from incorporating partially observed higher priority comparisons through conditional tie weighting. This pattern is consistent
with the simulation findings in Section \ref{sec:simulation}, where efficiency gains were most pronounced under moderate-to-heavy censoring, as is the case here with an 83\% censoring rate. Third, the two estimators differ slightly in their estimated trajectories over $\tau$. The estimator $\widehat{\NB}^{(\text{ipcw})}(\tau)$ exhibits somewhat larger variability as $\tau$ increases, which is expected since censoring accumulates and an increasing proportion of lower priority hospitalization comparisons remain indeterminate. By contrast, $\widehat{\NB}^{(\text{ctw})}(\tau)$ is more stable across restriction times, recovering information from these partially resolved pairs through the conditional tie probability weights $\widehat{\mathcal{R}}_{2,a}^{>}$ and $\widehat{\mathcal{R}}_{2,a}^{=}$.

Table \ref{tab:hfaction} summarizes the estimated win statistics at selected restriction times $\tau \in \{1, 2, 3, 4\}$ years. At $\tau = 3$ years, the proposed estimator gives $\widehat{\NB}^{(\text{ctw})}(3) = 0.029$ (95\% CI: $-0.080$, $0.137$), $\widehat{\WR}^{(\text{ctw})}(3) = 1.055$ (95\% CI: $0.861$, $1.293$), and $\widehat{\WO}^{(\text{ctw})}(3) = 1.059$ (95\% CI: $0.853$, $1.316$). At the same time, the estimator of \citet{cui2025ipcw} gives $\widehat{\NB}^{(\text{ipcw})}(3) = 0.032$ (95\% CI: $-0.079$, $0.143$), $\widehat{\WR}^{(\text{ipcw})}(3) = 1.064$ (95\% CI: $0.856$, $1.324$), and $\widehat{\WO}^{(\text{ipcw})}(3) = 1.066$ (95\% CI: $0.853$, $1.332$). At the component level at $\tau = 3$ years, under $\widehat{\pi}^{(\text{ctw})}$, the estimated death-component treatment win probability is $\widehat{\pi}_{t1}(3) = 0.305$ versus control loss $\widehat{\pi}_{c1}(3) = 0.280$, and the hospitalization-component treatment win probability is $\widehat{\pi}_{t2}(3) = 0.246$ versus $\widehat{\pi}_{c2}(3) = 0.241$. Both components contribute to the overall win advantage of the exercise-training arm, with the death component contributing the larger share of the net benefit. While neither the net benefit nor the win ratio is statistically significant at the conventional 5\% level, the direction and magnitude of the estimates are consistent across all restriction times and both methods, suggesting a clinically meaningful but modest benefit of exercise training on the hierarchical composite of death and hospitalization in this patient population.

\begin{table}[htbp]
\centering
\caption{Estimated win statistics for the HF-ACTION data at selected restriction times $\tau$, comparing exercise training ($n_{1}=1{,}060$) versus usual care ($n_{0}=1{,}070$). For each $\tau$, results are shown for the IPCW estimator of \citet{cui2025ipcw} $\widehat{\NB}^{(\text{ipcw})}$ \citep{cui2025ipcw} (top row) and the proposed conditional tie-weighting estimator $\widehat{\NB}^{(\text{ctw})}$ (bottom row). 95\% confidence intervals are in parentheses. RE $=\widehat{\mathrm{Var}}(\widehat{\NB}^{(\text{ipcw})})/
\widehat{\mathrm{Var}}(\widehat{\NB}^{(\text{ctw})})$.}
\label{tab:hfaction}
\resizebox{0.75\textwidth}{!}{%
\begin{tabular}{clllll}
\toprule
$\tau$ (yr) & Estimator & NB (95\% CI) & WR (95\% CI) & WO (95\% CI) & RE \\
\midrule
\multirow{2}{*}{1}
  & $\widehat{\NB}^{(\text{ipcw})}(t)$ & $0.016$ ($-0.090$,\;$0.123$) & $1.032$ ($0.839$,\;$1.270$) & $1.033$ ($0.835$,\;$1.278$) & \multirow{2}{*}{0.94} \\
  & $\widehat{\NB}^{(\text{ctw})}(t)$ & $0.030$ ($-0.079$,\;$0.140$) & $1.056$ ($0.868$,\;$1.285$) & $1.063$ ($0.853$,\;$1.324$) & \\
\midrule
\multirow{2}{*}{2}
  & $\widehat{\NB}^{(\text{ipcw})}(t)$ & $0.028$ ($-0.080$,\;$0.136$) & $1.055$ ($0.856$,\;$1.301$) & $1.057$ ($0.852$,\;$1.312$) & \multirow{2}{*}{0.97} \\
  & $\widehat{\NB}^{(\text{ctw})}(t)$ & $0.030$ ($-0.079$,\;$0.140$) & $1.056$ ($0.866$,\;$1.287$) & $1.062$ ($0.853$,\;$1.322$) & \\
\midrule
\multirow{2}{*}{3}
  & $\widehat{\NB}^{(\text{ipcw})}(t)$ & $0.032$ ($-0.079$,\;$0.143$) & $1.064$ ($0.856$,\;$1.324$) & $1.066$ ($0.853$,\;$1.332$) & \multirow{2}{*}{1.05} \\
  & $\widehat{\NB}^{(\text{ctw})}(t)$ & $0.029$ ($-0.080$,\;$0.137$) & $1.055$ ($0.861$,\;$1.293$) & $1.059$ ($0.853$,\;$1.316$) & \\
\midrule
\multirow{2}{*}{4}
  & $\widehat{\NB}^{(\text{ipcw})}(t)$ & $0.065$ ($-0.074$,\;$0.204$) & $1.131$ ($0.868$,\;$1.474$) & $1.139$ ($0.861$,\;$1.507$) & \multirow{2}{*}{1.69} \\
  & $\widehat{\NB}^{(\text{ctw})}(t)$ & $0.025$ ($-0.082$,\;$0.133$) & $1.051$ ($0.850$,\;$1.301$) & $1.052$ ($0.849$,\;$1.304$) & \\
\bottomrule
\end{tabular}}
\end{table}

The relative efficiency pattern is informative and broadly consistent with the simulation results in Section \ref{sec:simulation}. At shorter restriction times $\tau = 1$ and $\tau = 2$ years, where censoring has had little time to accumulate and most higher priority death comparisons are directly ascertainable, the two estimators are nearly equivalent (RE $\approx 0.94$--$0.97$ for NB), confirming that the proposed method does not sacrifice efficiency when the information gain from conditional tie recovery is small. As $\tau$ increases toward 3--4 years and censoring renders more death comparisons indeterminate, the efficiency advantage of $\widehat{\NB}^{(\text{ctw})}(\tau)$ grows substantially, with $\mathrm{RE} = 1.05$ at $\tau = 3$ years and $\mathrm{RE} = 1.69$ at $\tau = 4$ years for NB, with similar gains for WR ($\mathrm{RE} = 1.55$) and WO ($\mathrm{RE} = 1.70$) at the longest horizon. This pattern arises because, at $\tau = 4$ years, a substantial fraction of treated-control pairs have at least one member censored before $\tau$ on the death component, leaving their first priority comparison indeterminate. Thus, the estimator of \citet{cui2025ipcw} discards these pairs entirely, while the proposed estimator recovers their
partial information through the conditional tie weights $\widehat{\mathcal{R}}_{2,a}^{>}$ and $\widehat{\mathcal{R}}_{2,a}^{=}$.

\section{Discussion}\label{sec:discussion}

This paper is motivated by a practical limitation of restricted win statistics in trials with hierarchical survival endpoints. Under right censoring, lower priority comparisons may be lost not because the lower priority event is unobserved, but because the higher priority genuine tie required for descent through the hierarchy is not confirmed. Existing IPCW estimators correct the bias induced by censoring, but they retain an all-or-nothing rule requiring higher priority genuine ties to be confirmed before lower priority information can be used. We proposed a conditional tie weighting IPCW estimator that replaces the unobserved genuine tie indicator by its conditional probability given the observed data, allowing partially observed pairs with censoring-induced ties to contribute fractionally to lower priority comparisons while preserving the original clinical hierarchy.
This softens the statistical rigidity of the standard genuine tie approach while preserving the original clinical hierarchy, and recovers information that standard IPCW discards. The efficiency gains are most pronounced when censoring is substantial, the restriction horizon is long, the higher priority event is relatively rare by $\tau$, and the lower priority component is clinically and statistically informative. This is precisely the settings (and in our HF-ACTION data example) where many pairs would, in principle, descend the hierarchy but censoring prevents that descent from being fully confirmed.

Our investigation also points out an important implication concerning how restricted win statistics should be interpreted and reported. We adopt the restricted horizon route to estimand construction, which separates the treatment effect from censoring by treating the comparison horizon as part of the clinical question rather than a byproduct of observed follow-up \citep{mao2024defining,wang2025restricted}. This approach is complementary to semiparametric models that impose temporal structure on win probabilities, such as proportional win-fractions regression \citep{mao2021class} and generalized win fraction regression model \citep{cao2026generalized}. In practice, the restriction time should therefore be chosen for clinical reasons and interpreted as part of the estimand itself, not merely as a tuning parameter for the analysis. In addition, it is also suggested that no single win summary should be reported, because the win ratio, net benefit, and win odds are all functions of the same pairwise win, loss, and tie probabilities \citep{dong2023stratified}, and they provide complementary views of the same underlying comparison. In particular, the net benefit has the advantage of an absolute scale, while the win ratio and win odds retain the relative effect interpretation often preferred in clinical reports. When possible, reporting component specific contributions to the overall win probability may further clarify whether an observed treatment advantage is driven by the highest priority outcome or by lower priority components.

The proposed efficiency gain also comes with a trade-off. Unlike the IPCW of \citet{cui2025ipcw}, our estimator requires an additional model for the joint distribution of the component event times in order to estimate the conditional tie probabilities. When this working model is well chosen, the resulting estimator can recover information that IPCW alone leaves unused. When it is misspecified, however, the gain may be attenuated, and the analysis may become more sensitive to modeling assumptions. Our simulations empirically show that even under event time model misspecification, the proposed estimator retains a meaningful efficiency advantage over the existing IPCW estimator, suggesting that a partially informative joint model still recovers useful lower priority information rather than none. Moreover, the efficiency gain and bias performance are found to be largely robust to misspecification of the working copula assumption. Nevertheless, the additional joint event time model remains an assumption, and we recommend reporting sensitivity analyses across alternative dependence structures as a matter of good practice, particularly when the fitted model is weakly supported by subject-matter knowledge or available data.

\section*{Acknowledgments}

This work is supported by the United States National Institutes of Health (NIH), National Heart, Lung, and Blood Institute (NHLBI, grant number 1R01HL178513). All statements in this report, including its findings and conclusions, are solely those of the authors and do not necessarily represent the views of the NIH. The data example in Section \ref{sec:real} was prepared using the HF-ACTION Research Materials obtained from the NHLBI Biologic Specimen and Data Repository Information Coordinating Center (BioLINCC) and does not necessarily reflect the opinions or views of HF-ACTION or NHLBI.


\section*{Supplementary Materials}

The Web Appendix contains proofs of consistency and asymptotic normality, derivations of the influence function corrections and sandwich variance estimators, details on the copula implementation of the conditional tie probabilities, additional simulation results, and a tutorial for the \texttt{winIPCW} R package at \url{https://github.com/fancy575/winIPCW}.

\clearpage
\newpage
\printbibliography

\newpage
\appendix
\setcounter{section}{0}

\renewcommand{\thesection}{Web Appendix A\arabic{section}}
\renewcommand{\thesubsection}{\thesection.\arabic{subsection}}
\renewcommand{\thesubsubsection}{\thesubsection.\arabic{subsubsection}}

\setcounter{table}{0}
\setcounter{figure}{0}
\renewcommand{\thetable}{\arabic{table}}
\renewcommand{\thefigure}{\arabic{figure}}
\renewcommand{\tablename}{Web Table}
\renewcommand{\figurename}{Web Figure}

\section{Proof of the consistency of the estimator of \citet{cui2025ipcw} in \eqref{eq:cuiA} of the main text} \label{supp:consistent_cui}

We first establish that the estimator $\widehat{\pi}_{tq}^{(\text{ipcw})}(\tau)$ in \eqref{eq:cuiA} is unbiased for $\pi_{tq}(\tau)$ at the oracle level, and then invoke the law of large numbers for two-sample U-statistics together with uniform consistency of the estimated nuisance functions to complete the consistency argument.

\begin{proof}[Proof of consistency of $\widehat{\pi}_{t}^{(\text{ipcw})}(\tau)$]
We work with the oracle version in which $G_a(\cdot\mid\bZ)$ is known. Consider first the component $q=1$. The oracle kernel is
\[
K_{ij,1}^{(\text{ipcw})} = \frac{\Ind{\widetilde{Y}_{1i}(\tau)>\widetilde{Y}_{1j}(\tau)}\,\delta_{1j}(\tau)}{G_1\!\bigl(\widetilde{Y}_{1j}(\tau)\mid\bZ_i\bigr)\,G_0\!\bigl(\widetilde{Y}_{1j}(\tau)\mid\bZ_j\bigr)}.
\]
We compute $\E\{K_{ij,1}^{(\text{ipcw})}\mid\bZ_i,\bZ_j\}$ by the law of iterated expectation, first over $(C_i,\text{ipcw}_j)$ given $(\bT_i,\bT_j,\bZ_i,\bZ_j)$, and then over $(\bT_i,\bT_j)$ given $(\bZ_i,\bZ_j)$. On the event $\{T_{1j}\le\tau\}$, we have $\Ind{\widetilde{Y}_{1i}(\tau)>\widetilde{Y}_{1j}(\tau)}\,\delta_{1j}(\tau) = \Ind{T_{1i}>T_{1j}}\,\Ind{T_{1j}\le\tau}\,\Ind{C_i>T_{1j}}\,\Ind{C_j>T_{1j}}$, so the iterated expectation gives

{\small
\begin{align*}
\E\{K_{ij,1}^{(\text{ipcw})}\mid\bZ_i,\bZ_j\} &= \E_{\bT}\!\left[\E_{C}\!\left[\frac{\Ind{T_{1i}>T_{1j}}\,\Ind{T_{1j}\le\tau}\,\Ind{C_i>T_{1j}}\,\Ind{C_j>T_{1j}}}{G_1(T_{1j}\mid\bZ_i)\,G_0(T_{1j}\mid\bZ_j)}\;\biggm|\;\bT_i,\bT_j,\bZ_i,\bZ_j\right]\;\biggm|\;\bZ_i,\bZ_j\right] \\[4pt]
&= \E_{\bT}\!\left[\Ind{T_{1i}>T_{1j}}\,\Ind{T_{1j}\le\tau} \cdot\frac{\E_{C}\!\left[\Ind{C_i>T_{1j}}\mid\bZ_i,T_{1j}\right]}{G_1(T_{1j}\mid\bZ_i)} \cdot\frac{\E_{C}\!\left[\Ind{C_j>T_{1j}}\mid\bZ_j,T_{1j}\right]}{G_0(T_{1j}\mid\bZ_j)}\;\biggm|\;\bZ_i,\bZ_j\right] \\[4pt]
&= \E_{\bT}\bigl[\Ind{T_{1i}>T_{1j}}\,\Ind{T_{1j}\le\tau}\mid\bZ_i,\bZ_j\bigr] \\[4pt]
&= \pi_{t1}(\tau).
\end{align*}
}The second equality follows by iterated conditioning. Conditional on the latent event times and baseline covariates, the indicator and weight terms are fixed, and the two censoring indicators will be separated because the data from different arms are independent (recall Regularity Condition \eqref{ass:A1}). The covariate-dependent censoring assumption $C_a\perp\bT_a\mid(A_a,\bZ_a)$ (\ref{ass:A2}) then makes each censoring factor depend only on its own covariates, and hence it cancels the weight in the denominator. After cancellation, the remaining term is the first-priority win probability, which equals $\pi_{t1}(\tau)$ after averaging over the covariates by \eqref{eq:pi-decomp}. This establishes unbiasedness of the oracle kernel at $q=1$.

Now fix $q\ge 2$. The oracle kernel under \eqref{eq:cuiA} is
\[
K_{ij,q}^{(\text{ipcw})} = \frac{\Bigl[\prod_{k<q}\Ind{\widetilde{Y}_{ki}(\tau)=\tau,\;\widetilde{Y}_{kj}(\tau)=\tau}\Bigr]\Ind{\widetilde{Y}_{qi}(\tau)>\widetilde{Y}_{qj}(\tau)}\,\delta_{qj}(\tau)}{G_1(\tau\mid\bZ_i)\,G_0(\tau\mid\bZ_j)}.
\]
Under the common censoring convention, $\Ind{\widetilde{Y}_{ki}(\tau)=\tau} = \Ind{T_{ki}>\tau}\,\Ind{C_i>\tau}$ for each $k<q$, and on the event $\{C_j>\tau\}$ the condition $T_{qj}\le\tau$ implies $C_j>T_{qj}$, so $\Ind{\widetilde{Y}_{qi}(\tau)>\widetilde{Y}_{qj}(\tau)}\,\delta_{qj}(\tau) = \Ind{T_{qi}>T_{qj}}\,\Ind{T_{qj}\le\tau}$. The law of iterated expectation gives

{\small
\begin{align*}
&\E\{K_{ij,q}^{(\text{ipcw})}\mid\bZ_i,\bZ_j\} \\[4pt]
&\quad= \E_{\bT}\!\left[\E_{C}\!\left[\frac{\prod_{k<q}\Ind{T_{ki}>\tau}\cdot\Ind{C_i>\tau}\cdot\prod_{k<q}\Ind{T_{kj}>\tau}\cdot\Ind{C_j>\tau}\cdot\Ind{T_{qi}>T_{qj}}\,\Ind{T_{qj}\le\tau}}{G_1(\tau\mid\bZ_i)\,G_0(\tau\mid\bZ_j)}\;\biggm|\;\bT_i,\bT_j,\bZ_i,\bZ_j\right]\;\biggm|\;\bZ_i,\bZ_j\right] \\[4pt]
&\quad= \E_{\bT}\!\left[\prod_{k<q}\Ind{T_{ki}>\tau} \cdot\frac{\E_{C}\!\left[\Ind{C_i>\tau}\mid\bZ_i\right]}{G_1(\tau\mid\bZ_i)} \cdot\prod_{k<q}\Ind{T_{kj}>\tau} \cdot\frac{\E_{C}\!\left[\Ind{C_j>\tau}\mid\bZ_j\right]}{G_0(\tau\mid\bZ_j)} \cdot\Ind{T_{qi}>T_{qj}}\,\Ind{T_{qj}\le\tau}\;\biggm|\;\bZ_i,\bZ_j\right] \\[4pt]
&\quad= \E_{\bT}\!\left[\prod_{k<q}\Ind{T_{ki}>\tau}\cdot\prod_{k<q}\Ind{T_{kj}>\tau}\cdot\Ind{T_{qi}>T_{qj}}\,\Ind{T_{qj}\le\tau}\;\biggm|\;\bZ_i,\bZ_j\right] \\[4pt]
&\quad= \pi_{tq}(\tau).
\end{align*}
}Here, the second and third equalities follow exactly as in the $q=1$ case, where the censoring indicators separate by arm independence (recall Regularity Condition \eqref{ass:A1}), reduce to their own arm's covariates under $C_a\perp\bT_a\mid(A_a,\bZ_a)$ (recall Regularity Condition \eqref{ass:A2}), and cancel the inverse weights, now evaluated at the horizon $\tau$. In the final line, continuity identifies a higher-priority tie with joint survival up to $\tau$, i.e. $\{Y_{ki}=Y_{kj}\}=\{T_{ki}>\tau,\,T_{kj}>\tau\}$ for each $k<q$, so the integrand is the component $q$ win kernel and its expectation over $(\bZ_i,\bZ_j)$ equals $\pi_{tq}(\tau)$. This establishes unbiasedness of the oracle kernel at every priority level $q=1,\dots,Q$; summing over $q$ then gives $\E\{\widehat{\pi}_{t}^{(\text{ipcw})}(\tau)\}=\pi_t(\tau)$ at the oracle level.

To establish consistency, note that under \eqref{ass:A2} the kernel $K_{ij,q}^{(\text{ipcw})}$ is bounded by $\epsilon_G^{-2}$ uniformly in $(i,j)$. Together with the i.i.d.\ structure in \eqref{ass:A1} and the finite second moment in \eqref{ass:A3}, the law of large numbers for two-sample $U$-statistics gives
\[
\frac{1}{n_1 n_0}\sum_{i:\,A_i=1}\sum_{j:\,A_j=0} K_{ij,q}^{(\text{ipcw})} \xrightarrow{p} \pi_{tq}(\tau) \quad\text{as }n\to\infty,
\]
for every $q=1,\dots,Q$. Replacing the oracle $G_a(\cdot\mid\bZ)$ by $\widehat{G}_a(\cdot\mid\bZ)$ introduces the kernel difference
\begin{align*}
\widehat{K}_{ij,q}^{(\text{ipcw})} - K_{ij,q}^{(\text{ipcw})} &= N_{ij,q}^{(\text{ipcw})} \left(\frac{1}{\widehat{G}_1(s_{ij,q}\mid\bZ_i)\,\widehat{G}_0(s_{ij,q}\mid\bZ_j)} -\frac{1}{G_1(s_{ij,q}\mid\bZ_i)\,G_0(s_{ij,q}\mid\bZ_j)}\right),
\end{align*}
where $N_{ij,q}^{(\text{ipcw})}$ collects the indicator terms in the numerator of $K_{ij,q}^{(\text{ipcw})}$, and $s_{ij,q}=\widetilde{Y}_{1j}(\tau)$ for $q=1$ and $s_{ij,q}=\tau$ for $q\ge 2$. Under Regularity Condition \eqref{ass:A4}, $\widehat{G}_a$ satisfies $\sup_{0\le t\le\tau}|\widehat{G}_a(t\mid\bZ)-G_a(t\mid\bZ)|=o_p(1)$ uniformly, so by the continuous mapping theorem and $G_a(\tau\mid\bZ)\ge\epsilon_G>0$ from \eqref{ass:A2}, the difference is $o_p(1)$ uniformly in $(i,j)$, which yields
\[
\frac{1}{n_1 n_0}\sum_{i:\,A_i=1}\sum_{j:\,A_j=0}\bigl(\widehat{K}_{ij,q}^{(\text{ipcw})} - K_{ij,q}^{(\text{ipcw})}\bigr) = o_p(1),
\]
and hence $\widehat{\pi}_{tq}^{(\text{ipcw})}(\tau)\xrightarrow{p}\pi_{tq}(\tau)$ for every $q=1,\dots,Q$. Summing over $q$ gives $\widehat{\pi}_{t}^{(\text{ipcw})}(\tau)\xrightarrow{p}\pi_t(\tau)$. The consistency of $\widehat{\pi}_{c}^{(\text{ipcw})}(\tau)$ follows by the identical argument applied to the loss kernels, and whenever $\pi_c(\tau)>0$ the plug-in estimators $\widehat{\WR}^{(\text{ipcw})}(\tau)$, $\widehat{\NB}^{(\text{ipcw})}(\tau)$, and $\widehat{\WO}^{(\text{ipcw})}(\tau)$ are consistent by applying the continuous mapping theorem.
\end{proof}

\section{Proof of the consistency of the proposed estimator in \eqref{eq:ctw-win} of the main text}
\label{supp:consistent_proposed}

The estimator \eqref{eq:ctw-win} in the main paper can also be motivated directly from the conditional expectation argument. We first write the structural component-$q$ treatment-win kernel as $W_{ij,q}(\tau)=\{\prod_{k<q}\Ind{T_{ki}>\tau,\;T_{kj}>\tau}\}\Ind{Y_{qi}>Y_{qj}}$. Given the observed state at the time the component $q$ comparison is resolved, the treated individual's contribution to the unresolved gating event is 
$$\E\left[\prod_{k<q}\Ind{T_{ki}>\tau}\mid T_{1i}>u_{iq},\dots,T_{(q-1)i}>u_{iq},\,T_{qi}>t,\,A_i=1,\,\bZ_i\right],$$ 
which is exactly $\mathcal{R}_{q,1}^{>}(u_{iq},t\mid\bZ_i)$. The analogous conditional expectation for the control individual is $\mathcal{R}_{q,0}^{=}(u_{jq},t\mid\bZ_j)$. Thus, our estimator replaces an unobserved binary gate by its conditional expectation, and any consistent estimator of those conditional expectations would lead to the same first-order target. We now establish consistency of the proposed conditional tie weighting estimator. The proof follows the same strategy as in \ref{supp:consistent_cui}, where we first show that the oracle kernel is unbiased for the component specific restricted win probability, and then invoke the law of large numbers for two-sample $U$-statistics together with uniform consistency of the estimated nuisance functions.

\begin{proof}[Proof of consistency of $\widehat{\pi}_{t}^{(\text{ctw})}(\tau)$]
We first work with the oracle version in which $G_a(\cdot\mid\bZ)$, $\mathcal{S}_{q,a}(\cdot,\cdot\mid\bZ)$, and $\mathcal{H}_{q,a}(\cdot,\cdot\mid\bZ)$ are known. For the first-priority component, the proposed estimator coincides exactly with the estimator in \eqref{eq:cuiA}, that is, \(\widehat{\pi}_{t1}^{(\text{ctw})}(\tau)=\widehat{\pi}_{t1}^{(\text{ipcw})}(\tau)\).
Hence, by \ref{supp:consistent_cui}, we already have \(\widehat{\pi}_{t1}^{(\text{ctw})}(\tau)\xrightarrow{p}\pi_{t1}(\tau)\). It therefore remains to consider a fixed component $q\ge 2$. The oracle component $q$ treated win kernel corresponding to \eqref{eq:ctw-win} is
\[
K_{ij,q}^{(\text{ctw})}=\frac{\overline{\delta}_{qi}(\tau)\,\overline{\delta}_{qj}(\tau)\,\mathcal{R}_{q,1}^{>}\!\bigl(u_{iq},\widetilde{Y}_{qj}(\tau)\mid\bZ_i\bigr)\,\mathcal{R}_{q,0}^{=}\!\bigl(u_{jq},\widetilde{Y}_{qj}(\tau)\mid\bZ_j\bigr)\,
\Ind{\widetilde{Y}_{qi}(\tau)>\widetilde{Y}_{qj}(\tau)}\,\delta_{qj}(\tau)}{G_1\!\bigl(\widetilde{Y}_{qj}(\tau)\mid\bZ_i\bigr)\,G_0\!\bigl(\widetilde{Y}_{qj}(\tau)\mid\bZ_j\bigr)}.
\]
We now show that $\E\{K_{ij,q}^{(\text{ctw})}\}=\pi_{tq}(\tau)$. For $0\le t\le\tau$, define
\[
A_{iq}(t)=\frac{\overline{\delta}_{qi}(\tau)\,\mathcal{R}_{q,1}^{>}(u_{iq},t\mid\bZ_i)\,\Ind{\widetilde{Y}_{qi}(\tau)>t}}{G_1(t\mid\bZ_i)}.
\]
On the event $\{\overline{\delta}_{qi}(\tau)=1,\ \widetilde{Y}_{qi}(\tau)>t\}$, the common censoring convention implies
\[
u_{iq}=C_i\wedge\tau,\qquad T_{1i}>u_{iq},\dots,T_{(q-1)i}>u_{iq},\qquad T_{qi}>t,\qquad C_i>t.
\]
Hence, using the definition of $\mathcal{S}_{q,1}$ in \eqref{eq:jointSq} and under Assumption \ref{ass:A2}, we have

{\footnotesize
\begin{align}
\E\!\left\{A_{iq}(t)\mid \bZ_i,A_i=1\right\}&=\E_C\!\left[\frac{\Ind{C_i>t}}{G_1(t\mid\bZ_i)}\,\mathcal{R}_{q,1}^{>}(C_i\wedge\tau,t\mid\bZ_i)\,\Pr\!\left(
T_{1i}>C_i\wedge\tau,\dots,T_{(q-1)i}>C_i\wedge\tau,\,
T_{qi}>t\mid C_i,\bZ_i,A_i=1\right)\Biggm| \bZ_i,A_i=1\right]\notag\\&=\E_C\!\left[\frac{\Ind{C_i>t}}{G_1(t\mid\bZ_i)}\,
\frac{\mathcal{S}_{q,1}(\tau,t\mid\bZ_i)}{\mathcal{S}_{q,1}(C_i\wedge\tau,t\mid\bZ_i)}\,\mathcal{S}_{q,1}(C_i\wedge\tau,t\mid\bZ_i)\Biggm| \bZ_i,A_i=1\right]\notag\\
&=\mathcal{S}_{q,1}(\tau,t\mid\bZ_i)\,\E_C\!\left[\frac{\Ind{C_i>t}}{G_1(t\mid\bZ_i)}\Biggm| \bZ_i,A_i=1\right]\notag\\
&=\mathcal{S}_{q,1}(\tau,t\mid\bZ_i),
\label{eq:ctw-proof-treated-final}
\end{align}
}
Thus $A_{iq}(t)$ is an unbiased estimator of the conditional survival term $\mathcal{S}_{q,1}(\tau,t\mid\bZ_i)$.

Next, for any bounded measurable function $\phi:[0,\tau]\to\mathbb{R}$, define
\[
B_{jq}(\phi)=\frac{\overline{\delta}_{qj}(\tau)\,\mathcal{R}_{q,0}^{=}\!\bigl(u_{jq},\widetilde{Y}_{qj}(\tau)\mid\bZ_j\bigr)\,\delta_{qj}(\tau)\,\phi\!\bigl(\widetilde{Y}_{qj}(\tau)\bigr)}{G_0\!\bigl(\widetilde{Y}_{qj}(\tau)\mid\bZ_j\bigr)}.
\]
Conditional on $(C_j,\bZ_j,A_j=0)$, the event
\[\left\{\overline{\delta}_{qj}(\tau)=1,\;\delta_{qj}(\tau)=1,\;\widetilde{Y}_{qj}(\tau)\in[t,t+\dee t)\right\}\]
is equivalent to
\[
\left\{T_{1j}>C_j\wedge\tau,\dots,T_{(q-1)j}>C_j\wedge\tau,\;T_{qj}\in[t,t+\dee t),\;C_j>t\right\}.
\]
Therefore, by the definition of $\mathcal{H}_{q,0}$ and of $\mathcal{R}_{q,0}^{=}$ in \eqref{eq:Rq},
{\small
\begin{align}
\E\!\left\{B_{jq}(\phi)\mid \bZ_j,A_j=0\right\}
&=\E_C\!\left[\int_0^\tau\frac{\Ind{C_j>t}}{G_0(t\mid\bZ_j)}\,\mathcal{R}_{q,0}^{=}(C_j\wedge\tau,t\mid\bZ_j)\,\phi(t)\,\mathcal{H}_{q,0}(C_j\wedge\tau,t\mid\bZ_j)\,\dee t\Biggm| \bZ_j,A_j=0\right]\notag\\
&=\E_C\!\left[\int_0^\tau\frac{\Ind{C_j>t}}{G_0(t\mid\bZ_j)}\,\phi(t)\,\mathcal{H}_{q,0}(\tau,t\mid\bZ_j)\,\dee t\Biggm|\bZ_j,A_j=0\right]\notag\\
&=\int_0^\tau\phi(t)\,\mathcal{H}_{q,0}(\tau,t\mid\bZ_j)\,\E_C\!\left[\frac{\Ind{C_j>t}}{G_0(t\mid\bZ_j)}\Biggm| \bZ_j,A_j=0\right]\dee t\notag\\
&=\int_0^\tau \phi(t)\,\mathcal{H}_{q,0}(\tau,t\mid\bZ_j)\,\dee t.
\label{eq:ctw-proof-control-final}
\end{align}
}Now condition on $(\bZ_i,\bZ_j)$. By cross arm independence from Regularity Condition (\ref{ass:A1}), together with \eqref{eq:ctw-proof-treated-final} and \eqref{eq:ctw-proof-control-final}, we can show that
\begin{align}
\E\!\left\{K_{ij,q}^{(\text{ctw})}\mid \bZ_i,\bZ_j\right\}
&=\E\!\left[A_{iq}\!\left(\widetilde{Y}_{qj}(\tau)\right)\,\frac{\overline{\delta}_{qj}(\tau)\,\mathcal{R}_{q,0}^{=}\!\bigl(u_{jq},\widetilde{Y}_{qj}(\tau)\mid\bZ_j\bigr)\,\delta_{qj}(\tau)}{
G_0\!\bigl(\widetilde{Y}_{qj}(\tau)\mid\bZ_j\bigr)}
\Biggm| \bZ_i,\bZ_j
\right]\notag\\
&=\int_0^\tau\mathcal{S}_{q,1}(\tau,t\mid\bZ_i)\,\mathcal{H}_{q,0}(\tau,t\mid\bZ_j)\,\dee t.
\label{eq:ctw-proof-integral}
\end{align}
The right-hand side is exactly
\[
\Pr\!\bigl(T_{1i}>\tau,\dots,T_{(q-1)i}>\tau,\;T_{1j}>\tau,\dots,T_{(q-1)j}>\tau,\;T_{qi}>T_{qj},\;T_{qj}\le\tau\mid \bZ_i,\bZ_j,\;A_i=1,\;A_j=0\bigr),
\]
because $\mathcal{S}_{q,1}(\tau,t\mid\bZ_i)$ is the conditional survival function for the treated subject beyond $t$ on component $q$ while remaining tied through $\tau$ on all higher priority components, and $\mathcal{H}_{q,0}(\tau,t\mid\bZ_j)\,\dee t$ is the corresponding conditional subdensity that the control subject remains tied through $\tau$ on all higher priority components and experiences the component $q$ event in $[t,t+\dee t)$. Therefore, we obtain
\begin{align*}
&\E\!\left\{K_{ij,q}^{(\text{ctw})}\mid \bZ_i,\bZ_j\right\}\\
=&\E\!\left[
\left\{\prod_{k<q}\Ind{T_{ki}>\tau,T_{kj}>\tau}\right\}\Ind{T_{qi}>T_{qj},\;T_{qj}\le\tau}\Biggm| \bZ_i,\bZ_j,\;A_i=1,\;A_j=0\right],
\end{align*}
and taking expectation over $(\bZ_i,\bZ_j)$ yields \(\E\!\left\{K_{ij,q}^{(\text{ctw})}\right\}=\pi_{tq}(\tau)\).
Thus the oracle component-$q$ kernel is unbiased for every $q\ge2$.

Under Regularity Condition \eqref{ass:A2}, the IPCW factor is bounded by $\epsilon_G^{-2}$. In addition, the conditional tie ratios lie in $[0,1]$ by construction, so \(|K_{ij,q}^{(\text{ctw})}|\le \epsilon_G^{-2}\) uniformly in $(i,j)$ for each fixed $q\ge2$. Together with Regularity Conditions \eqref{ass:A1} and \eqref{ass:A3}, the law of large numbers for two-sample $U$-statistics implies
\[
\frac{1}{n_1n_0}\sum_{i:\,A_i=1}\sum_{j:\,A_j=0}K_{ij,q}^{(\text{ctw})}\xrightarrow{p}\pi_{tq}(\tau),\qquad q=2,\dots,Q.
\]

Finally, it remains to replace the oracle nuisance functions by their estimators. Let $\widehat{K}_{ij,q}^{(\text{ctw})}(\cdot)$ denote the feasible kernel obtained by substituting $\widehat{G}_a(\cdot)$, $\widehat{\mathcal{S}}_{q,a}(\cdot)$, and $\widehat{\mathcal{H}}_{q,a}(\cdot)$ into $K_{ij,q}^{(\text{ctw})}(\cdot)$. Under Regularity Condition \eqref{ass:A4}, these nuisance estimators are consistent, and by Regularity Condition \eqref{ass:A2} the denominators are bounded away from zero on $0\le t\le u\le\tau$. Therefore, by continuity of the maps \((x,y)\mapsto \frac{x}{y}\), \((x,y)\mapsto \frac{1}{xy}\), we have
\[
\sup_{i,j}\left|\widehat{K}_{ij,q}^{(\text{ctw})}-K_{ij,q}^{(\text{ctw})}\right|=o_p(1),
\qquad q=2,\dots,Q.
\]
Hence,
\[
\frac{1}{n_1n_0}\sum_{i:\,A_i=1}\sum_{j:\,A_j=0}\left(\widehat{K}_{ij,q}^{(\text{ctw})}-K_{ij,q}^{(\text{ctw})}\right)=o_p(1),
\qquad q=2,\dots,Q,
\]
which gives \(\widehat{\pi}_{tq}^{(\text{ctw})}(\tau)\xrightarrow{p}\pi_{tq}(\tau)\) for \(q=2,\dots,Q\). Combining this with the already established consistency of the first priority term, \(\widehat{\pi}_{t1}^{(\text{ctw})}(\tau)\xrightarrow{p}\pi_{t1}(\tau)\),
we obtain
\[
\widehat{\pi}_{t}^{(\text{ctw})}(\tau)=\sum_{q=1}^Q\widehat{\pi}_{tq}^{(\text{ctw})}(\tau)\xrightarrow{p}\sum_{q=1}^Q\pi_{tq}(\tau)=\pi_t(\tau).
\]
The consistency of $\widehat{\pi}_{c}^{(\text{ctw})}(\tau)$ follows by the same argument with the roles of treatment and control interchanged. Finally, whenever $\pi_c(\tau)>0$, the continuous mapping theorem implies that
\[
\widehat{\WR}^{(\text{ctw})}(\tau)\xrightarrow{p}\WR(\tau),\qquad
\widehat{\NB}^{(\text{ctw})}(\tau)\xrightarrow{p}\NB(\tau),\qquad
\widehat{\WO}^{(\text{ctw})}(\tau)\xrightarrow{p}\WO(\tau).
\]
This completes the proof.
\end{proof}


\section{Copula implementation of the conditional tie probabilities} 
\label{supp:copula}

This section gives implementation details for the copula-based nuisance model used to estimate the conditional tie probabilities $\mathcal{R}_{q,a}^{>}(u,t\mid\bz)$ and $\mathcal{R}_{q,a}^{=}(u,t\mid\bz)$ in \eqref{eq:Rq} of the main text. As discussed in Section \ref{sec:ctw}, the proposed estimator requires a model for the joint distribution of the higher-priority event times and the current component event time, conditional on treatment arm and baseline covariates. A convenient approach is to estimate the component-specific marginal survival functions separately and then link them through a survival copula. This separates marginal event-risk modeling from within-individual dependence modeling among endpoint components \citep{nelsen2006introduction,joe1997multivariate}.

For each arm $a\in\{0,1\}$ and component $r\in\{1,\ldots,Q\}$, let
\[
S_{r,a}(s\mid\bz)=\Pr(T_r>s\mid A=a,\bZ=\bz)
\]
denote the marginal survival function and let
\[
f_{r,a}(s\mid\bz)=-\partial S_{r,a}(s\mid\bz)/\partial s
\]
denote the corresponding marginal density. For a current priority level $q\ge2$, let $C_{q,a}(\cdot;\theta_{q,a})$ be a $q$-dimensional survival copula, possibly depending on both treatment arm $a$ and priority level $q$. We model the joint survival function in \eqref{eq:jointSq} by
\[
\mathcal{S}_{q,a}(u,t\mid\bz)
=
C_{q,a}
\{S_{1,a}(u\mid\bz),\ldots,S_{q-1,a}(u\mid\bz),S_{q,a}(t\mid\bz);\theta_{q,a}\}.
\]
Let $\partial_q C_{q,a}$ denote the partial derivative of the copula with respect to its $q$th argument. Differentiating the joint survival function with respect to the current component time $t$ gives
\[
  \mathcal{H}_{q,a}(u,t\mid\bz)
  =
  \partial_q C_{q,a}
  \{S_{1,a}(u\mid\bz),\ldots,S_{q-1,a}(u\mid\bz),S_{q,a}(t\mid\bz);\theta_{q,a}\}
  f_{q,a}(t\mid\bz).
\]
Therefore, the two conditional tie probabilities in \eqref{eq:Rq} can be written as
\begin{align}
  \mathcal{R}_{q,a}^{>}(u,t\mid\bz)
  &=
  \frac{
  C_{q,a}\{S_{1,a}(\tau\mid\bz),\ldots,S_{q-1,a}(\tau\mid\bz),S_{q,a}(t\mid\bz);\theta_{q,a}\}
  }{
  C_{q,a}\{S_{1,a}(u\mid\bz),\ldots,S_{q-1,a}(u\mid\bz),S_{q,a}(t\mid\bz);\theta_{q,a}\}
  },
  \label{eq:supp-Rgt-copula}
  \\
  \mathcal{R}_{q,a}^{=}(u,t\mid\bz)
  &=
  \frac{
  \partial_q C_{q,a}\{S_{1,a}(\tau\mid\bz),\ldots,S_{q-1,a}(\tau\mid\bz),S_{q,a}(t\mid\bz);\theta_{q,a}\}
  }{
  \partial_q C_{q,a}\{S_{1,a}(u\mid\bz),\ldots,S_{q-1,a}(u\mid\bz),S_{q,a}(t\mid\bz);\theta_{q,a}\}
  }.
  \label{eq:supp-Req-copula}
\end{align}
The marginal density $f_{q,a}(t\mid\bz)$ cancels from $\mathcal{R}_{q,a}^{=}$, which is why the proposed estimator can be implemented with semiparametric Cox margins without estimating a smooth event-time density. Thus, the event-time nuisance model enters the conditional tie weights only through fitted marginal survival functions and copula derivatives.

For the common two-component setting, with death as the first-priority component and a nonfatal event as the second-priority component, the joint survival function reduces to
\[
\mathcal{S}_{2,a}(u,t\mid\bz)
=
C_a\{S_{1,a}(u\mid\bz),S_{2,a}(t\mid\bz);\theta_a\},
\]
where we suppress the priority subscript on $C$ and $\theta$. The corresponding conditional tie probabilities are
\begin{equation}
  \mathcal{R}_{2,a}^{>}(u,t\mid\bz)=
  \frac{
  C_a\{S_{1,a}(\tau\mid\bz),S_{2,a}(t\mid\bz);\theta_a\}
  }{
  C_a\{S_{1,a}(u\mid\bz),S_{2,a}(t\mid\bz);\theta_a\}
  },
  \qquad
  \mathcal{R}_{2,a}^{=}(u,t\mid\bz)=
  \frac{
  C_{a,2}\{S_{1,a}(\tau\mid\bz),S_{2,a}(t\mid\bz);\theta_a\}
  }{
  C_{a,2}\{S_{1,a}(u\mid\bz),S_{2,a}(t\mid\bz);\theta_a\}
  },
  \label{eq:supp-R2-copula}
\end{equation}
where $C_{a,2}(u,v;\theta_a)=\partial C_a(u,v;\theta_a)/\partial v$. Once the fitted marginal survival functions $\widehat S_{1,a}$ and $\widehat S_{2,a}$ and the fitted copula parameter $\widehat\theta_a$ are available, all pairwise conditional tie probabilities are obtained by direct substitution. As an analytic check, the independence copula $C(u,v)=uv$ gives
\[
\mathcal{R}_{2,a}^{>}(u,t\mid\bz)
=
\mathcal{R}_{2,a}^{=}(u,t\mid\bz)
=
\frac{S_{1,a}(\tau\mid\bz)}{S_{1,a}(u\mid\bz)},
\]
which is the marginal conditional probability of remaining event-free on the first-priority component from $u$ to $\tau$.

The main text reports four one-parameter copula families, Gumbel--Hougaard, Clayton, Frank, and Plackett, because these are the families used in the copula-sensitivity simulation study and they provide closed-form expressions for both $C(u,v)$ and $C_2(u,v)$. The proposed estimator is not restricted to those four choices. For completeness, Table \ref{tab:supp-copula} lists a broader bivariate copula library available under the same implementation framework, including the independence copula, four Archimedean families, two elliptical families, and the Plackett family \citep{gumbel1960bivariate,clayton1978model,frank1979simultaneous,joe1997multivariate,demarta2005tcopula,plackett1965class}. Thus, Table \ref{tab:copula-main} in the main text can be viewed as the simulation copula table, whereas Table \ref{tab:supp-copula} below introduces additional copula choices that can be substituted into \eqref{eq:supp-Rgt-copula}--\eqref{eq:supp-Req-copula}. The Archimedean and Plackett families have analytic derivatives, while the Gaussian and Student $t$ copulas are written in terms of the corresponding bivariate distribution functions. We summarize each family by its copula function $C(u,v)$ and the partial derivative $C_2(u,v)=\partial C(u,v)/\partial v$ needed to evaluate $\mathcal{R}_{2,a}^{=}(\cdot)$.

\begin{table}[htbp]
\centering
\caption{Supplementary bivariate copula families available in the proposed estimator \eqref{eq:ctw-win}. The four families used in the main copula-sensitivity simulation study are summarized in Table \ref{tab:copula-main}; this table additionally lists the independence, Joe, Gaussian, and Student $t$ copulas. The partial derivative $C_2(u,v)=\partial C(u,v)/\partial v$ is used to evaluate the conditional tie ratio $\mathcal{R}_{2,a}^{=}$ in \eqref{eq:supp-R2-copula}.}
\label{tab:supp-copula}
\resizebox{\textwidth}{!}{%
\begin{tabular}{ccc}
\toprule
Family
&
Copula $C(u,v)$
&
Partial derivative $C_2(u,v)$
\\
\midrule
Independence
&
$uv$
&
$u$
\\[2mm]
\midrule
Gumbel--Hougaard
&
\shortstack[c]{
$\exp\left\{-A^{1/\theta}\right\}$, $\theta\ge 1$\\
$A=(-\log u)^\theta+(-\log v)^\theta$
}
&
$C(u,v)A^{1/\theta-1}(-\log v)^{\theta-1}/v$
\\[3mm]
\midrule
Clayton
&
\shortstack[c]{
$B^{-1/\theta}$, $\theta>0$\\
$B=u^{-\theta}+v^{-\theta}-1$
}
&
$v^{-\theta-1}B^{-1/\theta-1}$
\\[3mm]
\midrule
Frank
&
\shortstack[c]{
$-\theta^{-1}\log(B)$, $\theta\ne 0$\\
$B=1+\dfrac{(e^{-\theta u}-1)(e^{-\theta v}-1)}{e^{-\theta}-1}$
}
&
$\dfrac{e^{-\theta v}(e^{-\theta u}-1)}{(e^{-\theta}-1)B}$
\\[3mm]
\midrule
Joe
&
\shortstack[c]{
$1-D^{1/\theta}$, $\theta\ge 1$\\
$D=(1-u)^\theta+(1-v)^\theta-(1-u)^\theta(1-v)^\theta$
}
&
$D^{1/\theta-1}(1-v)^{\theta-1}\{1-(1-u)^\theta\}$
\\[3mm]
\midrule
Gaussian
&
\shortstack[c]{
$\Phi_\rho(x,y)$, $-1<\rho<1$\\
$x=\Phi^{-1}(u)$, $y=\Phi^{-1}(v)$
}
&
$\Phi\left\{\dfrac{x-\rho y}{(1-\rho^2)^{1/2}}\right\}$
\\[3mm]
\midrule
Student $t$
&
\shortstack[c]{
$t_{\nu,\rho}(x,y)$, $-1<\rho<1$, $\nu>0$\\
$x=t_\nu^{-1}(u)$, $y=t_\nu^{-1}(v)$
}
&
$t_{\nu+1}\left[\left\{\dfrac{\nu+1}{\nu+y^2}\right\}^{1/2}\dfrac{x-\rho y}{(1-\rho^2)^{1/2}}\right]$
\\[3mm]
\midrule
Plackett
&
\shortstack[c]{
$\dfrac{B-D^{1/2}}{2(\psi-1)}$, $\psi>0$, $\psi\ne 1$\\
$B=1+(\psi-1)(u+v)$\\
$D=B^2-4\psi(\psi-1)uv$
}
&
$\dfrac{1}{2}\left\{1-\dfrac{B-2\psi u}{D^{1/2}}\right\}$
\\
\bottomrule
\end{tabular}%
}
\end{table}

Within each treatment arm and for each component $r\in\{1,\ldots,Q\}$, we fit a marginal survival model using the observed time and event indicator pair $\{\widetilde Y_{ri}(\tau),\delta_{ri}(\tau)\}$ together with the baseline covariates $\bZ_i$. In the semiparametric implementation used in the main analyses, the marginal model is the Cox proportional hazards model
\[
  \lambda_{r,a}(s\mid\bz)
  =
  \lambda_{0r,a}(s)\exp(\bm\beta_{r,a}^{\top}\bz),
  \qquad
  S_{r,a}(s\mid\bz)
  =
  \exp\{-\Lambda_{0r,a}(s)\exp(\bm\beta_{r,a}^{\top}\bz)\},
\]
where $\lambda_{0r,a}(s)$ is the unspecified baseline hazard, $\Lambda_{0r,a}(s)=\int_0^s\lambda_{0r,a}(v)\dee v$ is the baseline cumulative hazard, and $\bm\beta_{r,a}$ is the regression coefficient vector. The fitted survival function is
\[
\widehat S_{r,a}(s\mid\bz)
=
\exp\{-\widehat\Lambda_{0r,a}(s)\exp(\widehat{\bm\beta}_{r,a}^{\top}\bz)\},
\]
where $\widehat{\bm\beta}_{r,a}$ is the partial likelihood estimator and $\widehat\Lambda_{0r,a}$ is the Breslow estimator. Parametric alternatives such as Weibull, log-normal, or generalized gamma margins can be substituted without changing the form of the proposed estimator.

The copula parameter is estimated after fitting the marginal models. This two-stage procedure is often called inference functions for margins or pseudo-likelihood estimation \citep{shih1995inferences,genest1993statistical}. We describe the bivariate case, which is the default implementation for $Q=2$. For individual $i$ in arm $a$, define $x_{ri}=\widetilde Y_{ri}(\tau)$, $d_{ri}=\delta_{ri}(\tau)$, and the fitted copula uniform $\widehat U_{ri}=\widehat S_{r,a}(x_{ri}\mid\bZ_i)$ for $r=1,2$. Let $c_\theta(u,v)=\partial^2 C_\theta(u,v)/\partial u\,\partial v$ denote the copula density and write $C_{\theta,1}=\partial C_\theta/\partial u$ and $C_{\theta,2}=\partial C_\theta/\partial v$ for the first-order partial derivatives. With the fitted margins treated as fixed, the copula pseudo log-likelihood for arm $a$ is
\begin{align*}
  \ell_a(\theta)=
  \sum_{i:A_i=a}
  \bigl\{
  &d_{1i}d_{2i}\log c_\theta(\widehat U_{1i},\widehat U_{2i})
  +d_{1i}(1-d_{2i})\log C_{\theta,1}(\widehat U_{1i},\widehat U_{2i})\\
  &+(1-d_{1i})d_{2i}\log C_{\theta,2}(\widehat U_{1i},\widehat U_{2i})
  +(1-d_{1i})(1-d_{2i})\log C_\theta(\widehat U_{1i},\widehat U_{2i})
  \bigr\}.
\end{align*}
The marginal density terms drop out because they do not depend on $\theta$ in the second-stage fit. We then estimate $\theta_a$ by
\[
\widehat\theta_a
=
\operatorname*{arg\,max}_{\theta\in\Theta}\ell_a(\theta),
\]
where $\Theta$ is the parameter space of the chosen family. For one-parameter Archimedean families, an estimate based on Kendall's $\tau_K$ is useful as an initial value when it lies in the admissible range. The Gumbel and Clayton starting values can be obtained from $\theta=1/(1-\widehat\tau_K)$ and $\theta=2\widehat\tau_K/(1-\widehat\tau_K)$, respectively. The Frank parameter is obtained by solving
\[
\widehat\tau_K
=
1-\frac{4}{\theta}+\frac{4D_1(\theta)}{\theta},
\qquad
D_1(\theta)
=
\theta^{-1}\int_0^\theta \frac{x}{\exp(x)-1}\dee x,
\]
where $D_1(\theta)$ is the first Debye function.

For $q>2$, the same construction applies with a $q$-dimensional copula. Let $x_{ri}=\widetilde Y_{ri}(\tau)$, $d_{ri}=\delta_{ri}(\tau)$, and $\widehat U_{ri}=\widehat S_{r,a}(x_{ri}\mid\bZ_i)$ for $r=1,\ldots,q$. If $\mathcal{D}_i=\{r:d_{ri}=1\}$ is the set of observed event components for individual $i$, then the copula contribution to the likelihood is the mixed partial derivative
\[
\partial_{\mathcal{D}_i}
C_{q,a}(\widehat U_{1i},\ldots,\widehat U_{qi};\theta_{q,a}),
\]
taken with respect to the arguments indexed by $\mathcal{D}_i$ and with no derivative taken for censored components. When the hierarchy is short, a full $q$-dimensional Archimedean or elliptical copula is straightforward. When the hierarchy is longer and the dependence structure is not plausibly exchangeable, pairwise composite likelihoods, nested Archimedean copulas, or vine copulas can be used. The proposed estimator only requires evaluating $C_{q,a}$ and $\partial_q C_{q,a}$ at the fitted model.

A simple extension beyond two components is the exchangeable Archimedean copula
\[
  C_\theta(u_1,\ldots,u_q)
  =
  \varphi_\theta^{-1}
  \{\varphi_\theta(u_1)+\cdots+\varphi_\theta(u_q)\},
\]
where $\varphi_\theta:[0,1]\to[0,\infty)$ is a strictly decreasing convex generator with $\varphi_\theta(1)=0$ and $\varphi_\theta^{-1}$ is its inverse. For such copulas,
\[
  \partial_q C_\theta(u_1,\ldots,u_q)
  =
  \{\varphi_\theta^{-1}\}'
  \{\varphi_\theta(u_1)+\cdots+\varphi_\theta(u_q)\}
  \varphi_\theta'(u_q),
\]
where $\varphi_\theta'$ and $\{\varphi_\theta^{-1}\}'$ are the derivatives of the generator and its inverse. Hence $\mathcal{R}_{q,a}^{>}$ is obtained from \eqref{eq:supp-Rgt-copula} by replacing the first $q-1$ copula arguments with the higher-priority marginal survivals at $\tau$ or $u$, and $\mathcal{R}_{q,a}^{=}$ is obtained from \eqref{eq:supp-Req-copula} by taking the ratio of the corresponding partial derivatives. The Gumbel generator is $\varphi_\theta(u)=(-\log u)^\theta$ with $\theta\ge1$, the Clayton generator is $\varphi_\theta(u)=(u^{-\theta}-1)/\theta$ with $\theta>0$, and the Frank generator is
\[
\varphi_\theta(u)
=
-\log\left\{
\frac{e^{-\theta u}-1}{e^{-\theta}-1}
\right\},
\qquad \theta\ne0.
\]
If the dependence between adjacent components is expected to differ from the dependence between clinically distant components, a nested Archimedean copula or a vine construction provides additional flexibility. Such choices affect only the nuisance functions $\widehat{\mathcal{R}}_{q,a}^{>}$ and $\widehat{\mathcal{R}}_{q,a}^{=}$, and the IPCW structure of \eqref{eq:ctw-win} is unchanged.

In numerical implementation, the fitted marginal survival probabilities are truncated to a small interval $[\epsilon,1-\epsilon]$ with $\epsilon=10^{-6}$ before evaluating logarithms and copula derivatives. The ratios $\widehat{\mathcal{R}}_{q,a}^{>}$ and $\widehat{\mathcal{R}}_{q,a}^{=}$ are conditional probabilities and therefore lie in $[0,1]$ under a valid fitted joint distribution, and small numerical deviations are truncated back to $[0,1]$. When $u=\tau$, both ratios equal one by construction, which provides a useful diagnostic. When a higher-priority event is observed before $\tau$, the factor $\overline{\delta}_{qi}(\tau)$ in \eqref{eq:ctw-win} equals zero and the value assigned to $u_{iq}$ does not affect the estimator.

The copula model is modular. Any method that consistently estimates $\mathcal{S}_{q,a}$ and $\mathcal{H}_{q,a}$, or equivalently the two ratios $\mathcal{R}_{q,a}^{>}$ and $\mathcal{R}_{q,a}^{=}$, can replace the copula construction. Multistate models, frailty models, and semi-competing risks models are therefore compatible with the same proposed estimator. The copula implementation is attractive because it is transparent, computationally simple, and based on interpretable dependence families.

\section{Proof of Theorem \ref{thm:asymp} and Corollary \ref{cor:delta} in Section \ref{sec:asymptotics} of the main paper}
\label{supp:asymp}

We prove Theorem \ref{thm:asymp} in two steps. First, we derive the first-order expansion of the oracle two-sample $U$-statistic obtained by fixing the nuisance functions at their true values. Second, we add the first-order effect of estimating the nuisance functions.

\begin{proof}[Proof of Theorem \ref{thm:asymp}]
Fix $\nu\in\{\text{ipcw},\text{ctw}\}$. Let
\[
  \widehat{\bpi}^{(\nu),\mathrm{or}}(\tau)
  =
  \frac{1}{n_1n_0}
  \sum_{i=1}^{n_1}\sum_{j=1}^{n_0}
  \bm{H}^{(\nu)}(O_{1i},O_{0j};\eta_{\nu0})
\]
denote the oracle estimator. Here and below, $O_1$ and $O_0$ denote independent generic observations from the treatment and control arms, respectively. By the identification and consistency arguments in Sections \ref{supp:consistent_cui} and \ref{supp:consistent_proposed},
\[
  \E\{\bm{H}^{(\nu)}(O_1,O_0;\eta_{\nu0})\}=\bpi(\tau),
\]
so the oracle statistic is centered at the target win and loss probability vector.

Define the first-order Hoeffding projections
\[
  \bm{\xi}_{1,\nu}(o_1)=\E\{\bm{H}^{(\nu)}(o_1,O_0;\eta_{\nu0})\}-\bpi(\tau),
  \qquad
  \bm{\xi}_{0,\nu}(o_0)=
  \E\{\bm{H}^{(\nu)}(O_1,o_0;\eta_{\nu0})\}-\bpi(\tau),
\]
and define the degenerate remainder kernel
\[
  \widetilde{\bm{H}}^{(\nu)}(o_1,o_0)=\bm{H}^{(\nu)}(o_1,o_0;\eta_{\nu0})-\bpi(\tau)-\bm{\xi}_{1,\nu}(o_1)-\bm{\xi}_{0,\nu}(o_0).
\]
By construction, \(\E\{\widetilde{\bm{H}}^{(\nu)}(O_1,O_0)\mid O_1\}=\bm{0}\),  \(\E\{\widetilde{\bm{H}}^{(\nu)}(O_1,O_0)\mid O_0\}=\bm{0}\). The two-sample Hoeffding decomposition therefore, gives
\begin{equation*}
  \widehat{\bpi}^{(\nu),\mathrm{or}}(\tau)-\bpi(\tau)=
  \frac{1}{n_1}\sum_{i=1}^{n_1}\bm{\xi}_{1,\nu}(O_{1i})+
  \frac{1}{n_0}\sum_{j=1}^{n_0}\bm{\xi}_{0,\nu}(O_{0j})+\bm{R}_{n,\nu},
\end{equation*}
where
\[
  \bm{R}_{n,\nu}=\frac{1}{n_1n_0}\sum_{i=1}^{n_1}\sum_{j=1}^{n_0}\widetilde{\bm{H}}^{(\nu)}(O_{1i},O_{0j}).
\]
Under Assumptions \ref{ass:A1} and \ref{ass:A3}, the oracle kernel is square-integrable. Because the degenerate kernel has mean zero after conditioning on either arm, all covariance terms in $\E\|\bm{R}_{n,\nu}\|^2$ vanish unless both the treatment-arm and control-arm indices coincide. Hence, \(\E\|\bm{R}_{n,\nu}\|^2=O\{(n_1n_0)^{-1}\}=O(n^{-2})\), and therefore $\bm{R}_{n,\nu}=O_p\{(n_1n_0)^{-1/2}\}=o_p(n^{-1/2})$. Thus,
\begin{equation}\label{eq:oracle-linear}
  \widehat{\bpi}^{(\nu),\mathrm{or}}(\tau)-\bpi(\tau)
  =\frac{1}{n_1}\sum_{i=1}^{n_1}\bm{\xi}_{1,\nu}(O_{1i})+\frac{1}{n_0}\sum_{j=1}^{n_0}\bm{\xi}_{0,\nu}(O_{0j})+
  o_p(n^{-1/2}).
\end{equation}

We next account for nuisance estimation. Write
\[
  \Delta_{n,\nu}=\frac{1}{n_1n_0}\sum_{i=1}^{n_1}\sum_{j=1}^{n_0}\left[\bm{H}^{(\nu)}(O_{1i},O_{0j};\widehat{\eta}_{\nu})-\bm{H}^{(\nu)}(O_{1i},O_{0j};\eta_{\nu0})\right].
\]
The stochastic equicontinuity and plug-in remainder condition in Assumption \ref{ass:A5} imply
\begin{equation}\label{eq:plugin-step}
  \Delta_{n,\nu}=\Psi_\nu(\widehat{\eta}_{\nu})-\Psi_\nu(\eta_{\nu0})+o_p(n^{-1/2}),
\end{equation}
where $\Psi_\nu(\eta_\nu)=\E\{\bm{H}^{(\nu)}(O_1,O_0;\eta_\nu)\}$. Since $\Psi_\nu$ is Frechet differentiable at $\eta_{\nu0}$,
\[
  \Psi_\nu(\widehat{\eta}_{\nu})-\Psi_\nu(\eta_{\nu0})=\Gamma_{1,\nu}[\widehat{\eta}_{1,\nu}-\eta_{1,\nu0}]+\Gamma_{0,\nu}[\widehat{\eta}_{0,\nu}-\eta_{0,\nu0}]+o_p(n^{-1/2}),
\]
where $\Gamma_{a,\nu}$ is the derivative of $\Psi_\nu$ with respect to the arm-$a$ nuisance component. By Assumption \ref{ass:A4},
\[
  \widehat{\eta}_{a,\nu}-\eta_{a,\nu0}=\frac{1}{n_a}\sum_{\ell=1}^{n_a}\kappa_{a,\nu}(O_{a\ell})+o_p(n_a^{-1/2}),\qquad a\in\{0,1\}.
\]
Because $\Gamma_{a,\nu}$ is linear and continuous,
\[
  \Gamma_{a,\nu}[\widehat{\eta}_{a,\nu}-\eta_{a,\nu0}]=\frac{1}{n_a}\sum_{\ell=1}^{n_a}\Gamma_{a,\nu}[\kappa_{a,\nu}(O_{a\ell})]+o_p(n_a^{-1/2}), \qquad a\in\{0,1\}.
\]
Substituting this expansion into \eqref{eq:plugin-step} gives
\begin{equation}\label{eq:nuisance-linear}
  \widehat{\bpi}^{(\nu)}(\tau)-\widehat{\bpi}^{(\nu),\mathrm{or}}(\tau)=
  \frac{1}{n_1}\sum_{i=1}^{n_1}\Gamma_{1,\nu}[\kappa_{1,\nu}(O_{1i})]+\frac{1}{n_0}\sum_{j=1}^{n_0}\Gamma_{0,\nu}[\kappa_{0,\nu}(O_{0j})]+o_p(n^{-1/2}).
\end{equation}
Combining \eqref{eq:oracle-linear} and \eqref{eq:nuisance-linear},
\[
  \widehat{\bpi}^{(\nu)}(\tau)-\bpi(\tau)=\frac{1}{n_1}\sum_{i=1}^{n_1}\left\{\bm{\xi}_{1,\nu}(O_{1i})+\Gamma_{1,\nu}[\kappa_{1,\nu}(O_{1i})]\right\}+\frac{1}{n_0}\sum_{j=1}^{n_0}\left\{\bm{\xi}_{0,\nu}(O_{0j})+\Gamma_{0,\nu}[\kappa_{0,\nu}(O_{0j})]\right\}+o_p(n^{-1/2}).
\]
With \(\bm{\psi}_{a,\nu}(O_{ai})=\bm{\xi}_{a,\nu}(O_{ai})+\Gamma_{a,\nu}[\kappa_{a,\nu}(O_{ai})]\), \(a\in\{0,1\}\), this becomes
\[
  \widehat{\bpi}^{(\nu)}(\tau)-\bpi(\tau)=\frac{1}{n_1}\sum_{i=1}^{n_1}\bm{\psi}_{1,\nu}(O_{1i})
  +\frac{1}{n_0}\sum_{j=1}^{n_0}\bm{\psi}_{0,\nu}(O_{0j})+o_p(n^{-1/2}).
\]
Multiplying by $\sqrt{n}$ yields the asserted linear expansion,
\[
  \sqrt{n}\{\widehat{\bpi}^{(\nu)}(\tau)-\bpi(\tau)\}=\frac{\sqrt{n}}{n_1}\sum_{i=1}^{n_1}\bm{\psi}_{1,\nu}(O_{1i})+\frac{\sqrt{n}}{n_0}\sum_{j=1}^{n_0}\bm{\psi}_{0,\nu}(O_{0j})+o_p(1).
\]

Finally, $\E\{\bm{\xi}_{a,\nu}(O_a)\}=\bm{0}$ by definition and $\E\{\kappa_{a,\nu}(O_a)\}=\bm{0}$ by Assumption \ref{ass:A4}, so $\E\{\bm{\psi}_{a,\nu}(O_a)\}=\bm{0}$. Under Assumptions \ref{ass:A1}, \ref{ass:A3}, and \ref{ass:A4}, $\bm{\psi}_{a,\nu}(O_a)$ is square-integrable. The multivariate central limit theorem gives
\[
  \frac{1}{\sqrt{n_1}}\sum_{i=1}^{n_1}\bm{\psi}_{1,\nu}(O_{1i})
  \xrightarrow{d} \mathcal{N}\!\left(\bm{0},\Var\{\bm{\psi}_{1,\nu}(O_1)\}\right),
\]
and
\[
  \frac{1}{\sqrt{n_0}}\sum_{j=1}^{n_0}\bm{\psi}_{0,\nu}(O_{0j})\xrightarrow{d}\mathcal{N}\!\left(\bm{0},\Var\{\bm{\psi}_{0,\nu}(O_0)\}\right).
\]
The two limiting normal variables are independent because the two treatment arms are independent. Since $n_1/n\to\rho$ and $n_0/n\to1-\rho$, Slutsky's theorem gives
\[
  \sqrt{n}\{\widehat{\bpi}^{(\nu)}(\tau)-\bpi(\tau)\}\xrightarrow{d}\mathcal{N}(\bm{0},\bOmega_\nu),
\]
where \(\bOmega_\nu=\rho^{-1}\Var\{\bm{\psi}_{1,\nu}(O_1)\}+(1-\rho)^{-1}\Var\{\bm{\psi}_{0,\nu}(O_0)\}\).  This proves Theorem \ref{thm:asymp}.
\end{proof}

The corollary follows from the multivariate delta method applied to smooth functions of $(\pi_t,\pi_c)$.

\begin{proof}[Proof of Corollary \ref{cor:delta}]
Define \( g_{\WR}(x,y)=\log(x/y)\), \(g_{\NB}(x,y)=x-y\) and
\[
g_{\WO}(x,y)=\log\!\left[\frac{x+\frac12(1-x-y)}{y+\frac12(1-x-y)}\right]=\log\!\left(\frac{1+x-y}{1-x+y}\right).
\]
Under the stated conditions, these three maps are continuously differentiable at $(\pi_t(\tau),\pi_c(\tau))$. Their gradients are
\[
  \nabla g_{\WR}(x,y)=
  \begin{pmatrix}
    1/x\\[1mm]
    -1/y
  \end{pmatrix},
  \qquad
  \nabla g_{\NB}(x,y)=
  \begin{pmatrix}
    1\\[1mm]
    -1
  \end{pmatrix},
\]
and
\[
  \nabla g_{\WO}(x,y)
  =
  \frac{2}{1-(x-y)^2}
  \begin{pmatrix}
    1\\[1mm]
    -1
  \end{pmatrix}.
\]
Evaluating these gradients at $(\pi_t(\tau),\pi_c(\tau))$ gives
\[
  \bm{g}_{\WR}
  =
  \begin{pmatrix}
    1/\pi_t(\tau)\\[1mm]
    -1/\pi_c(\tau)
  \end{pmatrix},
  \qquad
  \bm{g}_{\NB}
  =
  \begin{pmatrix}
    1\\[1mm]
    -1
  \end{pmatrix},
\]
and
\[
  \bm{g}_{\WO}
  =
  \frac{2}{1-\{\pi_t(\tau)-\pi_c(\tau)\}^2}
  \begin{pmatrix}
    1\\[1mm]
    -1
  \end{pmatrix}.
\]
Applying the multivariate delta method to Theorem \ref{thm:asymp} yields the three stated limiting distributions.
\end{proof}

\section{Derivation of the influence-function corrections and sandwich variance estimators}
\label{supp:ifdetails}

This section gives the details behind the influence functions used for variance estimation. The main paper shows the asymptotic expansion in terms of the individual-level influence function
\[
  \bm{\psi}_{a,\nu}(O_{ai})=
  \bm{\xi}_{a,\nu}(O_{ai})+\bm{r}_{a,\nu}(O_{ai}),\qquad
  \bm{r}_{a,\nu}(O_{ai})=\Gamma_{a,\nu}[\kappa_{a,\nu}(O_{ai})],
\]
where $\bm{\xi}_{a,\nu}$ is the first-order Hoeffding projection and $\bm{r}_{a,\nu}$ is the nuisance estimation correction. We now make $\bm{r}_{a,\nu}$ explicit for the censoring model and, for the CTW estimator, for the event-time model used to estimate the conditional tie probabilities.

\subsection{General form of the nuisance corrections}

Recall that
\[
  \Psi_\nu(\eta_\nu)=\E\{\bm{H}^{(\nu)}(O_1,O_0;\eta_\nu)\},
  \qquad\eta_\nu=(\eta_{1,\nu},\eta_{0,\nu}).
\]
Let \(D_{a,\nu}\bm{H}^{(\nu)}(o_1,o_0;\eta_{\nu0})[h_a]\) denote the Frechet derivative of the pairwise kernel with respect to the arm-$a$ nuisance component in the direction \(h_a\). By linearity of expectation,
\[
  \Gamma_{a,\nu}[h_a] = \E\!\left\{D_{a,\nu}\bm{H}^{(\nu)}(O_1,O_0;\eta_{\nu0})[h_a]
  \right\},
  \qquad a\in\{0,1\}.
\]
When the arm-$a$ nuisance estimator has influence function $\kappa_{a,\nu}(O_{ai})$, the corresponding nuisance correction for a treatment-arm subject is obtained by holding $O_{1i}$ fixed and averaging over a generic control arm observation:
\[
  \bm{r}_{1i,\nu}=\E\!\left[D_{1,\nu}\bm{H}^{(\nu)}(O_{1i},O_0;\eta_{\nu0})[\kappa_{1,\nu}(O_{1i})]\Bigm| O_{1i}\right].
\]
Similarly, for a control-arm subject,
\[
  \bm{r}_{0j,\nu} =\E\!\left[D_{0,\nu}\bm{H}^{(\nu)}(O_1,O_{0j};\eta_{\nu0})[\kappa_{0,\nu}(O_{0j})]\Bigm| O_{0j}\right].
\]
These two forms are simply the arm-specific version of the derivative term $\Gamma_{a,\nu}[\kappa_{a,\nu}(O_{ai})]$ in the main theorem. The remainder of this section evaluates them for the nuisance functions used by the two estimators.

\subsection{Censoring model correction}

The correction from estimating the censoring survival functions \(G_1(\cdot)\) and \(G_0(\cdot)\) is common to both estimators. Write \(K_q^{(\nu)}\) and \(L_q^{(\nu)}\) for the component-$q$ contributions to the win and loss kernels, so that
\[
  K^{(\nu)}=\sum_{q=1}^Q K_q^{(\nu)},
  \qquad
  L^{(\nu)}=\sum_{q=1}^Q L_q^{(\nu)}.
\]
Let \(t_{ij,q}^W\) and \(t_{ij,q}^L\) denote the observed times at which the censoring survival functions are evaluated in the component $q$ win and loss kernels. The exact value of these times is determined by the corresponding kernel. For example, for the first-priority win kernel, \(t_{ij,1}^W=\widetilde{Y}_{1j}(\tau)\), for the IPCW estimator of \eqref{eq:cuiA} at a lower priority level, the censoring time argument is $\tau$, and for the proposed estimator at a lower priority level, the censoring time argument is the observed component $q$ comparison time. This notation lets us write the correction in a single form.

For the treatment-arm censoring survival \(G_1(\cdot)\), the dependence of \(K_q^{(\nu)}\) and \(L_q^{(\nu)}\) on \(G_1(\cdot)\) enters through the factors \(1/G_1(t_{ij,q}^W\mid\bZ_i)\) and \(1/G_1(t_{ij,q}^L\mid\bZ_i)\), respectively. Thus, for a perturbation \(G_1^\varepsilon=G_1+\varepsilon h_1\),
\begin{align}
  D_{1,G}K_q^{(\nu)}[h_1]&=-K_q^{(\nu)}(O_{1i},O_{0j};\eta_{\nu0})\frac{h_1(t_{ij,q}^W\mid\bZ_i)}
       {G_1(t_{ij,q}^W\mid\bZ_i)},
  \label{eq:supp-DKG}
  \\
  D_{1,G}L_q^{(\nu)}[h_1]
  &=-L_q^{(\nu)}(O_{1i},O_{0j};\eta_{\nu0})\frac{h_1(t_{ij,q}^L\mid\bZ_i)}
       {G_1(t_{ij,q}^L\mid\bZ_i)}.
  \label{eq:supp-DLG}
\end{align}
Substituting the influence function $\kappa_{1i,G}$ of $\widehat G_1$ into \eqref{eq:supp-DKG}--\eqref{eq:supp-DLG} gives the treatment-arm censoring correction
\begin{equation}\label{eq:supp-rG-trt}
  \bm{r}_{1i,\nu}^{G}
  =-\E\!\left[\sum_{q=1}^Q
  \begin{pmatrix}
    K_q^{(\nu)}(O_{1i},O_0;\eta_{\nu0})\,
    \kappa_{1i,G}(t_{ij,q}^W\mid\bZ_i)\big/
    G_1(t_{ij,q}^W\mid\bZ_i)
    \\[4pt]
    L_q^{(\nu)}(O_{1i},O_0;\eta_{\nu0})\,
    \kappa_{1i,G}(t_{ij,q}^L\mid\bZ_i)\big/
    G_1(t_{ij,q}^L\mid\bZ_i)
  \end{pmatrix}
  \Biggm| O_{1i}
  \right].
\end{equation}
The control-arm correction is analogous, but the conditional expectation is taken over a generic treatment arm observation:
\begin{equation}\label{eq:supp-rG-ctl}
  \bm{r}_{0j,\nu}^{G}
  =-\E\!\left[\sum_{q=1}^Q
  \begin{pmatrix}
    K_q^{(\nu)}(O_1,O_{0j};\eta_{\nu0})\,
    \kappa_{0j,G}(t_{ij,q}^W\mid\bZ_j)\big/
    G_0(t_{ij,q}^W\mid\bZ_j)
    \\[4pt]
    L_q^{(\nu)}(O_1,O_{0j};\eta_{\nu0})\,
    \kappa_{0j,G}(t_{ij,q}^L\mid\bZ_j)\big/
    G_0(t_{ij,q}^L\mid\bZ_j)
  \end{pmatrix}
  \Biggm| O_{0j}
  \right].
\end{equation}
Equations \eqref{eq:supp-rG-trt} and \eqref{eq:supp-rG-ctl} show that the censoring correction has the same form for both estimators, with only the kernels and the time arguments \(t_{ij,q}^W,t_{ij,q}^L\) changing.

\subsubsection{Kaplan--Meier censoring model}

Under completely independent censoring within arm \(a\), \(G_a(t\mid\bZ)=G_a(t)\). Let \(X_i^{(a)}=D_i^{(a)}\wedge C_i^{(a)}\) be the observed follow-up time, where \(D_i^{(a)}\) denotes the time at which the clinical observation process would end in the absence of censoring and \(C_i^{(a)}\) is the censoring time. Define the censoring counting process \(N_{ai}^C(t)=\Ind{X_i^{(a)}\le t,\;D_i^{(a)}>X_i^{(a)}}\), the censoring at-risk process \(Y_{ai}^C(t)=\Ind{X_i^{(a)}\ge t}\), the censoring cumulative hazard \(\Lambda_a^C(t)\), and the censoring martingale \(M_{ai}^C(t)=N_{ai}^C(t) -\int_0^t Y_{ai}^C(u)\,\dee\Lambda_a^C(u)\). Let \(y_a(t)=\E\{Y_{ai}^C(t)\}\). The Kaplan--Meier estimator of the censoring survival function has the standard martingale expansion
\[
  \sqrt{n_a}\{\widehat{G}_a(t)-G_a(t)\}=\frac{1}{\sqrt{n_a}}\sum_{i=1}^{n_a}\kappa_{ai,G}^{\mathrm{KM}}(t) + o_p(1),
\]
where
\begin{equation}\label{eq:supp-kappaKM}
  \kappa_{ai,G}^{\mathrm{KM}}(t)=-G_a(t)\int_0^t\frac{\dee M_{ai}^C(u)}{y_a(u)}.
\end{equation}
Substituting \(\kappa_{ai,G}^{\mathrm{KM}}\) into \eqref{eq:supp-rG-trt} or \eqref{eq:supp-rG-ctl} gives the censoring correction under completely independent censoring.

\subsubsection{Cox censoring model}

Under covariate-dependent censoring, suppose the censoring hazard in arm \(a\) follows the Cox model \(\lambda_a^C(t\mid\bZ)=\lambda_{0a}^C(t)\exp(\bm{\alpha}_a^\top\bZ)\).
Let \(\bm{s}_a^{(k)}(t)=\E\!\left[Y_{ai}^C(t)\exp(\bm{\alpha}_{a0}^\top\bZ_i)\bZ_i^{\otimes k}\right]\), for \(k=0,1,2,\)
and \(\overline{\bm z}_a(t)=\bm{s}_a^{(1)}(t)/s_a^{(0)}(t)\). With \(M_{ai}^C(t)\) now denoting the censoring martingale under the Cox model, define the individual censoring score
\[
  U_{ai}^C=\int_0^\tau\{\bZ_i-\overline{\bm z}_a(u)\}\,\dee M_{ai}^C(u),
\]
and the information matrix
\[
  A_a^C = \int_0^\tau
  \left\{
  \frac{s_a^{(2)}(u)}{s_a^{(0)}(u)}-\overline{\bm z}_a(u)^{\otimes 2}\right\}s_a^{(0)}(u)\,\dee\Lambda_{0a}^C(u).
\]
Then
\[\sqrt{n_a}(\widehat{\bm\alpha}_a-\bm{\alpha}_{a0}) = (A_a^C)^{-1}\frac{1}{\sqrt{n_a}}\sum_{i=1}^{n_a}U_{ai}^C+ o_p(1).
\]
Let \(B_a^C(t)=\int_0^t\overline{\bm z}_a(u)\,\dee\Lambda_{0a}^C(u)\). The Breslow estimator of the baseline cumulative censoring hazard satisfies
\[\sqrt{n_a}\{\widehat{\Lambda}_{0a}^C(t)-\Lambda_{0a}^C(t)\}=\frac{1}{\sqrt{n_a}}\sum_{i=1}^{n_a}\phi_{ai}^C(t)+o_p(1),
\]
where \(  \phi_{ai}^C(t)=\int_0^t\{s_a^{(0)}(u)\}^{-1}\dee M_{ai}^C(u)-B_a^C(t)^\top (A_a^C)^{-1}U_{ai}^C
\). Since \(G_a(t\mid\bz)=\exp\{-\Lambda_{0a}^C(t)\exp(\bm{\alpha}_{a0}^\top\bz)\}\),
the delta method gives the Cox-model influence function for the censoring survival:
\begin{equation}\label{eq:supp-kappaCox}
  \kappa_{ai,G}^{\mathrm{Cox}}(t,\bz)=-G_a(t\mid\bz)\exp(\bm{\alpha}_{a0}^\top\bz)
  \left[\phi_{ai}^C(t)+\Lambda_{0a}^C(t)\,\bz^\top (A_a^C)^{-1}U_{ai}^C\right].
\end{equation}
Substituting \(\kappa_{ai,G}^{\mathrm{Cox}}\) into \eqref{eq:supp-rG-trt} or \eqref{eq:supp-rG-ctl} gives the censoring correction under the Cox censoring model.

\subsection{Event-model correction for the proposed estimator}

Only the proposed estimator \((\nu=I)\) has an additional nuisance correction from estimating the joint event-time quantities \(\mathcal{S}_{q,a}\) and \(\mathcal{H}_{q,a}\). We first give the general form for any \(q\ge2\), and then specialize to the bivariate Gumbel copula implementation.

Fix \(q\ge2\). In the component-$q$ treated-win kernel, the treatment-arm event model enters through
\(\mathcal{R}_{q,1}^{>}(u_{iq},\widetilde{Y}_{qj}(\tau)\mid\bZ_i)\), because the treated subject must remain event-free on component \(q\) beyond the control subject's observed component-$q$ event time. Therefore, for a perturbation \(h_{1,S,q}\) of \(\mathcal{S}_{q,1}\),
\begin{equation}\label{eq:supp-DKevent}
  D_{1,E}K_q^{(\text{ctw})}[h_{1,S,q}]=K_q^{(\text{ctw})}(O_{1i},O_{0j};\eta_{I0})\,D\log \mathcal{R}_{q,1}^{>}(u_{iq},\widetilde{Y}_{qj}(\tau)\mid\bZ_i)[h_{1,S,q}].
\end{equation}
In the component $q$ treated-loss kernel, the treatment-arm event model enters through
\(\mathcal{R}_{q,1}^{=}(u_{iq},\widetilde{Y}_{qi}(\tau)\mid\bZ_i)\), because the treated subject is the one whose component-$q$ event time resolves the loss. Thus, for a perturbation \(h_{1,H,q}\) of \(\mathcal{H}_{q,1}\),
\begin{equation*}
  D_{1,E}L_q^{(\text{ctw})}[h_{1,H,q}]=L_q^{(\text{ctw})}(O_{1i},O_{0j};\eta_{I0})\,D\log \mathcal{R}_{q,1}^{=}(u_{iq},\widetilde{Y}_{qi}(\tau)\mid\bZ_i)[h_{1,H,q}].
\end{equation*}
Using \(\mathcal{R}_{q,a}^{>}(u,t\mid\bz)=\mathcal{S}_{q,a}(\tau,t\mid\bz)/\mathcal{S}_{q,a}(u,t\mid\bz)\) and
\(\mathcal{R}_{q,a}^{=}(u,t\mid\bz)=\mathcal{H}_{q,a}(\tau,t\mid\bz)/\mathcal{H}_{q,a}(u,t\mid\bz)\), the quotient rule gives
\begin{align}
  D\log \mathcal{R}_{q,a}^{>}(u,t\mid\bz)[h_{S,q}]
  &=\frac{h_{S,q}(\tau,t,\bz)}
       {\mathcal{S}_{q,a}(\tau,t\mid\bz)}-\frac{h_{S,q}(u,t,\bz)}
       {\mathcal{S}_{q,a}(u,t\mid\bz)},
  \label{eq:supp-DlogRgt}
  \\
  D\log \mathcal{R}_{q,a}^{=}(u,t\mid\bz)[h_{H,q}]
  &=\frac{h_{H,q}(\tau,t,\bz)}{\mathcal{H}_{q,a}(\tau,t\mid\bz)}-\frac{h_{H,q}(u,t,\bz)} {\mathcal{H}_{q,a}(u,t\mid\bz)}.
  \label{eq:supp-DlogReq}
\end{align}

Let \(\kappa_{ai,S,q}\) and \(\kappa_{ai,H,q}\) denote the influence functions of the estimators of \(\mathcal{S}_{q,a}\) and \(\mathcal{H}_{q,a}\), respectively. Substituting \eqref{eq:supp-DKevent}--\eqref{eq:supp-DlogReq} into the general derivative formula yields the treatment-arm event-model correction
\begin{equation}\label{eq:supp-rE-general-trt}
  \bm{r}_{1i,\text{ctw}}^{E}
  =\E\!\left[\sum_{q=2}^Q\begin{pmatrix}K_q^{(\text{ctw})}(O_{1i},O_0;\eta_{I0})\,
    D\log \mathcal{R}_{q,1}^{>}(u_{iq},\widetilde{Y}_{qj}(\tau)\mid\bZ_i)[\kappa_{1i,S,q}]
    \\[4pt]
    L_q^{(\text{ctw})}(O_{1i},O_0;\eta_{I0})\,D\log \mathcal{R}_{q,1}^{=}(u_{iq},\widetilde{Y}_{qi}(\tau)\mid\bZ_i)[\kappa_{1i,H,q}]
  \end{pmatrix}
  \Biggm| O_{1i}
  \right].
\end{equation}
The control-arm correction has the same logic, but the roles of \(>\) and \(=\) are interchanged between the win and loss kernels:
\begin{equation}\label{eq:supp-rE-general-ctl}
  \bm{r}_{0j,\text{ctw}}^{E}=\E\!\left[\sum_{q=2}^Q
  \begin{pmatrix}
    K_q^{(\text{ctw})}(O_1,O_{0j};\eta_{I0})\,
    D\log \mathcal{R}_{q,0}^{=}(u_{jq},\widetilde{Y}_{qj}(\tau)\mid\bZ_j)
    [\kappa_{0j,H,q}]
    \\[4pt]
    L_q^{(\text{ctw})}(O_1,O_{0j};\eta_{I0})\,
    D\log \mathcal{R}_{q,0}^{>}(u_{jq},\widetilde{Y}_{qi}(\tau)\mid\bZ_j)
    [\kappa_{0j,S,q}]
  \end{pmatrix}
  \Biggm| O_{0j}
  \right].
\end{equation}
Thus, for the proposed estimator, \(\bm{r}_{a,\text{ctw}}(O_{ai})=\bm{r}_{a,\text{ctw}}^{G}(O_{ai})+\bm{r}_{a,\text{ctw}}^{E}(O_{ai})\), for \(a\in\{0,1\}\) whereas for the estimator of \eqref{eq:cuiA}, \(\bm{r}_{a,\text{ipcw}}(O_{ai})=\bm{r}_{a,\text{ipcw}}^{G}(O_{ai})\).

\subsection{Finite-dimensional Gumbel copula model}
We now specialize the event-model correction to the common bivariate setting \(Q=2\) under the Gumbel copula. In arm \(a\), write
\[
  \mathcal{S}_{2,a}(u,t\mid\bz;\bm\vartheta_a)=C_{\theta_a}
  \{S_{1,a}(u\mid\bz;\bm\vartheta_a),S_{2,a}(t\mid\bz;\bm\vartheta_a)\},
\]
where \(\bm\vartheta_a\) collects all marginal event-time parameters and the copula parameter \(\theta_a\). The Gumbel copula is
\[
  C_\theta(u,v)=
  \exp\!\left\{-\left[(-\log u)^\theta+(-\log v)^\theta\right]^{1/\theta}\right\},\qquad \theta\ge1.
\]
Let \(A_\theta(u,v)=(-\log u)^\theta+(-\log v)^\theta\). To verify the copula derivatives used below, write \(B_\theta(u,v)=A_\theta(u,v)^{1/\theta}\) and  \(C_\theta(u,v)=\exp\{-B_\theta(u,v)\). Since \(\frac{\partial B_\theta(u,v)}{\partial v}=A_\theta(u,v)^{1/\theta-1}(-\log v)^{\theta-1}\left(-\frac{1}{v}\right)\).
we have
\[
\frac{\partial C_\theta(u,v)}{\partial v}=
-C_\theta(u,v)\frac{\partial B_\theta(u,v)}{\partial v}=C_\theta(u,v)A_\theta(u,v)^{1/\theta-1}\frac{(-\log v)^{\theta-1}}{v}.\]
The derivative with respect to \(u\) is obtained by the same calculation with \(u\) and \(v\) interchanged. Therefore, the first-order partial derivatives are
\begin{align*}
  \partial_1 C_\theta(u,v)
  &=C_\theta(u,v)A_\theta(u,v)^{1/\theta-1}
  \frac{(-\log u)^{\theta-1}}{u},
  \\
  \partial_2 C_\theta(u,v)
  &=C_\theta(u,v)A_\theta(u,v)^{1/\theta-1}
  \frac{(-\log v)^{\theta-1}}{v}.
\end{align*}
The chain rule gives
\begin{equation*}
  \nabla_{\bm\vartheta_a}\mathcal{S}_{2,a}(u,t\mid\bz) =\partial_1 C_{\theta_a}\,\nabla_{\bm\vartheta_a}S_{1,a}(u\mid\bz)
  +\partial_2 C_{\theta_a}\,\nabla_{\bm\vartheta_a}S_{2,a}(t\mid\bz)
  +\partial_\theta C_{\theta_a}\,\bm e_\theta,
\end{equation*}
where all copula derivatives are evaluated at \(\{S_{1,a}(u\mid\bz),S_{2,a}(t\mid\bz),\theta_a\}\), and \(\bm e_\theta\) selects the copula parameter from \(\bm\vartheta_a\). If the fitted event model is finite-dimensional and regular, then
\[
  \sqrt{n_a}(\widehat{\bm\vartheta}_a-\bm\vartheta_{a0})=I_a^{-1}\frac{1}{\sqrt{n_a}}\sum_{i=1}^{n_a}U_{ai}^{E}+
  o_p(1),
\]
where \(U_{ai}^{E}\) is the individual event model score and \(I_a\) is the Fisher information. Therefore, the influence function of the fitted joint survival is
\begin{equation}\label{eq:supp-kappaS-param}
  \kappa_{ai,S}(u,t,\bz)=\nabla_{\bm\vartheta_a}
  \mathcal{S}_{2,a}(u,t\mid\bz)^\top I_a^{-1}U_{ai}^{E}.
\end{equation}

The associated subdensity is
\[
  \mathcal{H}_{2,a}(u,t\mid\bz)=\partial_2 C_{\theta_a}\{S_{1,a}(u\mid\bz),S_{2,a}(t\mid\bz)\}f_{2,a}(t\mid\bz),
\]
where \(f_{2,a}(t\mid\bz)=-\partial S_{2,a}(t\mid\bz)/\partial t\). A further application of the chain rule gives
\begin{align*}
  \nabla_{\bm\vartheta_a}\mathcal{H}_{2,a}(u,t\mid\bz)
  &=
  \partial_{12}C_{\theta_a}\,
  f_{2,a}(t\mid\bz)\,
  \nabla_{\bm\vartheta_a}S_{1,a}(u\mid\bz)\\
  &\quad+
  \partial_{22}C_{\theta_a}\,
  f_{2,a}(t\mid\bz)\,
  \nabla_{\bm\vartheta_a}S_{2,a}(t\mid\bz)
  +
  \partial_{2\theta}C_{\theta_a}\,
  f_{2,a}(t\mid\bz)\,\bm e_\theta\\
  &\quad+
  \partial_2 C_{\theta_a}\,
  \nabla_{\bm\vartheta_a}f_{2,a}(t\mid\bz),
\end{align*}
with all copula derivatives evaluated at the same arguments. Hence,
\begin{equation}\label{eq:supp-kappaH-param}
  \kappa_{ai,H}(u,t,\bz)=\nabla_{\bm\vartheta_a}\mathcal{H}_{2,a}(u,t\mid\bz)^\top I_a^{-1}U_{ai}^{E}.
\end{equation}
Substituting \eqref{eq:supp-kappaS-param} and \eqref{eq:supp-kappaH-param} into the log-ratio derivatives \eqref{eq:supp-DlogRgt}--\eqref{eq:supp-DlogReq}, and then into \eqref{eq:supp-rE-general-trt}--\eqref{eq:supp-rE-general-ctl}, gives the event model correction under a finite-dimensional Gumbel copula event-time model.

\subsection{Semiparametric Cox margins under the Gumbel copula}
\label{supp:semi-copula}

We now derive the same event model correction when the two copula margins are modeled semi-parametrically by Cox proportional hazards models. This form is useful in practice because the conditional tie ratio \(\mathcal{R}_{2,a}^{=}\) can be evaluated without estimating a smooth marginal density for the second component. For component \(q\in\{1,2\}\) in arm \(a\), suppose
\[
  \lambda_{q,a}(t\mid\bz)=\lambda_{0,q,a}(t)\exp(\bm\beta_{q,a}^\top\bz),
  \qquad
  S_{q,a}(t\mid\bz)= \exp\{-\Lambda_{0,q,a}(t)\exp(\bm\beta_{q,a}^\top\bz)\}.
\]
Let \(N_{q,ai}(t)\) and \(Y_{q,ai}(t)\) denote the component-$m$ counting and at-risk processes, and define the event martingale
\(M_{q,ai}(t)=N_{q,ai}(t)-\int_0^t Y_{q,ai}(u)\exp(\bm\beta_{q,a0}^\top\bZ_i) \dee\Lambda_{0,q,a}(u)\).
Write \(s_{q,a}^{(k)}(t)=\E\!\left[Y_{q,ai}(t)\exp(\bm\beta_{q,a0}^\top\bZ_i)\bZ_i^{\otimes k}\right]\), and  \(\overline{\bm z}_{q,a}(t)=\frac{s_{q,a}^{(1)}(t)}{s_{q,a}^{(0)}(t)}\). The individual Cox score for component \(q\) is \(U_{m,ai}^{E}=\int_0^\tau\{\bZ_i-\overline{\bm z}_{m,a}(u)\}\,\dee M_{m,ai}(u)\),
and the information matrix is
\[
  A_{m,a}^{E}=\int_0^\tau\left[\frac{\bm s_{q,a}^{(2)}(u)}{s_{q,a}^{(0)}(u)}-\overline{\bm z}_{q,a}(u)^{\otimes 2}
  \right]s_{q,a}^{(0)}(u)\,\dee\Lambda_{0,q,a}(u).\]
Then
\[
  \sqrt{n_a}(\widehat{\bm\beta}_{q,a}-\bm\beta_{q,a0})=(A_{q,a}^{E})^{-1}
  \frac{1}{\sqrt{n_a}}\sum_{i=1}^{n_a}U_{q,ai}^{E}+o_p(1).
\]
Let \(B_{q,a}(t)=\int_0^t\overline{\bm z}_{q,a}(u)\,\dee\Lambda_{0,q,a}(u)\). The Breslow estimator of \(\Lambda_{0,q,a}\) has the expansion \(\sqrt{n_a}\{\widehat{\Lambda}_{0,q,a}(t)-\Lambda_{0,q,a}(t)\}=\frac{1}{\sqrt{n_a}}\sum_{i=1}^{n_a}\phi_{q,ai}(t)+o_p(1)\), where \( \phi_{q,ai}(t)=\int_0^t\{s_{q,a}^{(0)}(u)\}^{-1}\dee M_{q,ai}(u)-B_{q,a}(t)^\top (A_{q,a}^{E})^{-1}U_{q,ai}^{E}\). Applying the delta method to \(S_{q,a}(t\mid\bz)\) gives the marginal survival influence function
\begin{equation}\label{eq:supp-kappaSm-Cox}
  \kappa_{q,ai,S}^{\mathrm{Cox}}(t,\bz)=-S_{q,a}(t\mid\bz)\exp(\bm\beta_{q,a0}^\top\bz)
  \left[\phi_{q,ai}(t)+\Lambda_{0,q,a}(t)\,\bz^\top (A_{q,a}^{E})^{-1}U_{q,ai}^{E}\right].
\end{equation}
For the bivariate copula, write \(\mathcal{S}_{2,a}(u,t\mid\bz)=C_{\theta_a}\{S_{1,a}(u\mid\bz),S_{2,a}(t\mid\bz)\}\),
and define \(\mathcal{A}_{2,a}(u,t\mid\bz)=\partial_2 C_{\theta_a}\{S_{1,a}(u\mid\bz),S_{2,a}(t\mid\bz)\}\).
Then \(\mathcal{H}_{2,a}(u,t\mid\bz)=\mathcal{A}_{2,a}(u,t\mid\bz)f_{2,a}(t\mid\bz)\), and the marginal density cancels in the conditional tie ratio:
\[\mathcal{R}_{2,a}^{=}(u,t\mid\bz)=\frac{\mathcal{H}_{2,a}(\tau,t\mid\bz)}{\mathcal{H}_{2,a}(u,t\mid\bz)}= \frac{\mathcal{A}_{2,a}(\tau,t\mid\bz)}{\mathcal{A}_{2,a}(u,t\mid\bz)}.\]
Thus, under Cox margins, the event model correction can be computed using fitted marginal survival functions and copula derivatives only. Let \(\widehat{\theta}_a\) be a regular estimator of the copula parameter satisfying
\[
  \sqrt{n_a}(\widehat{\theta}_a-\theta_{a0})=\frac{1}{\sqrt{n_a}}\sum_{i=1}^{n_a}\kappa_{ai,\theta}+o_p(1),
\]
where \(\kappa_{ai,\theta}\) is mean zero and square integrable. Combining this expansion with \eqref{eq:supp-kappaSm-Cox} gives
\begin{align*}
  \kappa_{ai,\mathcal{S}}^{\mathrm{semi}}(u,t,\bz)&=\partial_1 C_{\theta_a}\,\kappa_{1,ai,S}^{\mathrm{Cox}}(u,\bz)+
  \partial_2 C_{\theta_a}\,\kappa_{2,ai,S}^{\mathrm{Cox}}(t,\bz)+\partial_\theta C_{\theta_a}\,\kappa_{ai,\theta},
  \\
  \kappa_{ai,\mathcal{A}}^{\mathrm{semi}}(u,t,\bz)&=\partial_{12} C_{\theta_a}\,\kappa_{1,ai,S}^{\mathrm{Cox}}(u,\bz)+\partial_{22} C_{\theta_a}\,\kappa_{2,ai,S}^{\mathrm{Cox}}(t,\bz)+\partial_{2\theta} C_{\theta_a}\,\kappa_{ai,\theta},
\end{align*}
where all copula derivatives are evaluated at
\(\{S_{1,a}(u\mid\bz),S_{2,a}(t\mid\bz),\theta_{a0}\}\). The corresponding log-ratio derivatives are
\begin{align}
  D\log\mathcal{R}_{2,a}^{>}(u,t\mid\bz)
  [\kappa_{ai,\mathcal{S}}^{\mathrm{semi}}]
  &=\frac{\kappa_{ai,\mathcal{S}}^{\mathrm{semi}}(\tau,t,\bz)}{\mathcal{S}_{2,a}(\tau,t\mid\bz)}-\frac{\kappa_{ai,\mathcal{S}}^{\mathrm{semi}}(u,t,\bz)}{\mathcal{S}_{2,a}(u,t\mid\bz)},
  \label{eq:supp-DlogRgt-semi}
  \\
  D\log\mathcal{R}_{2,a}^{=}(u,t\mid\bz)
  [\kappa_{ai,\mathcal{A}}^{\mathrm{semi}}]&=\frac{\kappa_{ai,\mathcal{A}}^{\mathrm{semi}}(\tau,t,\bz)}{\mathcal{A}_{2,a}(\tau,t\mid\bz)}-\frac{\kappa_{ai,\mathcal{A}}^{\mathrm{semi}}(u,t,\bz)}{\mathcal{A}_{2,a}(u,t\mid\bz)}.
  \label{eq:supp-DlogReq-semi}
\end{align}
Substituting \eqref{eq:supp-DlogRgt-semi}--\eqref{eq:supp-DlogReq-semi} into the general event-model correction gives, for a treatment-arm individual,
\begin{equation}\label{eq:supp-rE-semi}
  \bm{r}_{1i,\text{ctw}}^{E,\mathrm{semi}}
  =\E\!\left[
  \begin{pmatrix}
    K_2^{(\text{ctw})}(O_{1i},O_0;\eta_{I0})\,
    D\log\mathcal{R}_{2,1}^{>}(u_{i2},\widetilde{Y}_{2j}(\tau)\mid\bZ_i)
    [\kappa_{1i,\mathcal{S}}^{\mathrm{semi}}]
    \\[4pt]
    L_2^{(\text{ctw})}(O_{1i},O_0;\eta_{I0})\,
    D\log\mathcal{R}_{2,1}^{=}(u_{i2},\widetilde{Y}_{2i}(\tau)\mid\bZ_i)
    [\kappa_{1i,\mathcal{A}}^{\mathrm{semi}}]
  \end{pmatrix}
  \Biggm| O_{1i}
  \right],
\end{equation}
with the control arm correction defined analogously. This representation is algebraically the same as the finite-dimensional correction above, but it avoids the need to estimate \(f_{2,a}(t\mid\bz)\) and its influence function.

\subsection{Sandwich variance estimators}

The preceding displays lead directly to the sandwich variance estimators. Define the empirical Hoeffding projections
\begin{align*}
  \widehat{\bm\xi}_{1,\nu,i}&=\frac{1}{n_0}\sum_{j=1}^{n_0}
  \bm{H}^{(\nu)}(O_{1i},O_{0j};\widehat{\eta}_\nu)-\widehat{\bpi}^{(\nu)}(\tau),
  \\
  \widehat{\bm\xi}_{0,\nu,j}&=\frac{1}{n_1}\sum_{i=1}^{n_1}\bm{H}^{(\nu)}(O_{1i},O_{0j};\widehat{\eta}_\nu)-\widehat{\bpi}^{(\nu)}(\tau).
\end{align*}
Let \(\widehat{\bm r}_{ai,\nu}^{G}\) be the plug-in version of the censoring correction, using either the Kaplan--Meier influence function in \eqref{eq:supp-kappaKM} or the Cox influence function in \eqref{eq:supp-kappaCox}. Let \(\widehat{\bm r}_{ai,\text{ctw}}^{E}\) be the corresponding plug-in version of the event-model correction for the proposed estimator, using either the finite-dimensional form in \eqref{eq:supp-kappaS-param}--\eqref{eq:supp-kappaH-param} or the semiparametric Cox-margin form in \eqref{eq:supp-rE-semi}.

For the IPCW estimator \((\nu=C)\), the estimated influence function is \(\widehat{\bm\psi}_{a,\text{ipcw},i}=\widehat{\bm\xi}_{a,\text{ipcw},i}+\widehat{\bm r}_{ai,\text{ipcw}}^{G}\) for \(a\in\{0,1\}\), 
and the sandwich estimator is
\begin{equation}\label{eq:supp-OmegaC}
  \widehat{\bOmega}_C=\frac{n}{n_1^2}
  \sum_{i=1}^{n_1}
  \widehat{\bm\psi}_{1,\text{ipcw},i}\widehat{\bm\psi}_{1,\text{ipcw},i}^\top+\frac{n}{n_0^2}
  \sum_{j=1}^{n_0}\widehat{\bm\psi}_{0,\text{ipcw},j}\widehat{\bm\psi}_{0,\text{ipcw},j}^\top.
\end{equation}
For the proposed estimator \((\nu=I)\), the estimated influence function is \(\widehat{\bm\psi}_{a,\text{ctw},i}=\widehat{\bm\xi}_{a,\text{ctw},i}+\widehat{\bm r}_{ai,\text{ctw}}^{G}+\widehat{\bm r}_{ai,\text{ctw}}^{E}\) for \(a\in\{0,1\}\),  and the sandwich estimator is
\begin{equation}\label{eq:supp-OmegaI}
  \widehat{\bOmega}_I=\frac{n}{n_1^2}\sum_{i=1}^{n_1}\widehat{\bm\psi}_{1,\text{ctw},i}\widehat{\bm\psi}_{1,\text{ctw},i}^\top
  +\frac{n}{n_0^2}\sum_{j=1}^{n_0}\widehat{\bm\psi}_{0,\text{ctw},j}\widehat{\bm\psi}_{0,\text{ctw},j}^\top.
\end{equation}

To justify \eqref{eq:supp-OmegaC} and \eqref{eq:supp-OmegaI}, note that the proof of Theorem \ref{thm:asymp} implies \(  \widehat{\bm\psi}_{a,\nu,i} =\bm\psi_{a,\nu}(O_{ai})+o_p(1)\) in \(L_2\), uniformly within arm \(a\), provided the nuisance estimators and the sample analogues of the derivative terms are consistent. Therefore,
\[
  \frac{1}{n_a}\sum_{i=1}^{n_a}\widehat{\bm\psi}_{a,\nu,i}\widehat{\bm\psi}_{a,\nu,i}^\top\xrightarrow{p}\Var\{\bm\psi_{a,\nu}(O_a)\},\qquad a\in\{0,1\}.
\]
Since \(n_1/n\to\rho\) and \(n_0/n\to1-\rho\),
\[\widehat{\bOmega}_\nu\xrightarrow{p}\rho^{-1}\Var\{\bm\psi_{1,\nu}(O_1)\}+(1-\rho)^{-1}\Var\{\bm\psi_{0,\nu}(O_0)\}=\bOmega_\nu.
\]
Thus, the sandwich estimators consistently estimate the covariance matrix appearing in Theorem \ref{thm:asymp}.

For a smooth scalar summary \(h(\bpi)\) with gradient \(\bm g_h\), the estimated large-sample variance of \(h\{\widehat{\bpi}^{(\nu)}(\tau)\}\) is
\[
  \widehat{\Var}\!\left[
  h\{\widehat{\bpi}^{(\nu)}(\tau)\}
  \right]
  =
  \frac{1}{n}
  \widehat{\bm g}_h^\top
  \widehat{\bOmega}_\nu
  \widehat{\bm g}_h,
\]
where \(\widehat{\bm g}_h\) is the gradient evaluated at
\(\widehat{\bpi}^{(\nu)}(\tau)\). In particular,
\(\widehat{\Var}\{\widehat{\NB}^{(\nu)}(\tau)\}=\frac{1}{n}\widehat{\bm g}_{\NB}^\top\widehat{\bOmega}_\nu\widehat{\bm g}_{\NB}\), and
\(\widehat{\Var}\{\log\widehat{\WR}^{(\nu)}(\tau)\}=\frac{1}{n}\widehat{\bm g}_{\WR}^\top\widehat{\bOmega}_\nu\widehat{\bm g}_{\WR}\), \(\widehat{\Var}\{\log\widehat{\WO}^{(\nu)}(\tau)\}=\frac{1}{n}\widehat{\bm g}_{\WO}^\top\widehat{\bOmega}_\nu\widehat{\bm g}_{\WO}\). These are the variance estimators used for the simulation studies and the HF-ACTION data illustration.

\section{Additional simulation results}
\label{supp:sim_additional}

For the first priority event in the main paper, corresponding to death, we set $\rho_1=1.35$, $\lambda_{10}=0.0008$, $\beta_{11}=0.35$, $\beta_{12}=0.60$, $\beta_{13}=0.25$, and $\beta_{1a}=-0.05$. For the second priority event, corresponding to the serious nonfatal event, we set $\rho_2=0.95$, $\lambda_{20}=0.0200$, $\beta_{21}=0.30$, $\beta_{22}=0.70$, $\beta_{23}=0.20$, and $\beta_{2a}=-0.35$. Under this parameterisation, the first priority event is relatively rare and has a modest treatment effect, whereas the second priority event is more common and has a larger treatment effect, and thus lower priority comparisons contribute nontrivially to the overall win probabilities. 

This section provides additional simulation results that complement with Section \ref{sec:simulation} of the main text. The main text reports the primary results from two simulation studies. The first is the original nuisance-model misspecification study, in which data were generated from Weibull proportional hazards margins linked by a Gumbel--Hougaard copula and the working nuisance configurations $\mathcal{M}_1$--$\mathcal{M}_4$ were used to assess misspecification of the event-time model, the censoring model, and both nuisance components. The second is the copula-sensitivity study, in which the marginal event-time model was correctly specified through Cox proportional hazards models and the data-generating and working copula families were varied over Gumbel--Hougaard, Clayton, Frank, and Plackett.

In this section, we reports the remaining results including additional censoring levels for $\NB(\tau)$ in the original nuisance model misspecification study, the corresponding $\WR(\tau)$ and $\WO(\tau)$ results, the three-component prioritized endpoint simulations, and additional copula-sensitivity results for other estimands and censoring levels. For $\WR(\tau)$ and $\WO(\tau)$, MCSD, ASE, COV, and RE are reported on the log scale.

\subsection{Nuisance model misspecification study: additional net-benefit results}
\label{supp:sim_original_NB}

Tables \ref{tab:sim-NB20}--\ref{tab:sim-NB60} provide the corresponding results under 20\%, 40\%, and 60\% censoring. These tables use the same data-generating mechanism, restriction times, dependence strengths, and nuisance configurations as described in Section \ref{sec:simulation}. The lower-censoring results show the same qualitative pattern as the 80\% setting in the main text, but with smaller efficiency gains because fewer higher-priority ties are unresolved before the restriction time.

Under the correctly specified nuisance configuration $\mathcal{M}_1$, both estimators have small relative bias and empirical coverage close to the nominal 95\% level across censoring levels. The relative efficiency of the proposed estimator increases as censoring becomes heavier and as the restriction time lengthens, consistent with the fact that conditional tie weighting differs from IPCW estimator of \citet{cui2025ipcw} for lower-priority comparisons whose higher-priority ties are censoring-induced ties. Event-time model misspecification under $\mathcal{M}_2$ reduces the gain but generally does not eliminate it. Censoring-model misspecification under $\mathcal{M}_3$ and $\mathcal{M}_4$ has a more visible effect on finite-sample bias and coverage, because both estimators rely on the censoring model for inverse probability weighting.

\begin{table}[htbp]
\centering
\caption{Simulation results for $\NB(\tau)$ under 20\% censoring by $\tau=36$ for the two component prioritized endpoint. RB: relative bias (\%); MCSD: Monte Carlo standard deviation; ASE: average estimated standard error; Cov: empirical coverage of the 95\% Wald interval; RE $=\mathrm{MCSD}^2(C)/\mathrm{MCSD}^2(I)$. Based on 1{,}000 replications.}
\label{tab:sim-NB20}
\resizebox{\textwidth}{!}{%
\begin{tabular}{ccccrrrrrrrrr}
\toprule
 & & & & \multicolumn{4}{c}{$\widehat{\NB}^{(\text{ipcw})}(\tau)$} & \multicolumn{4}{c}{$\widehat{\NB}^{(\text{ctw})}(\tau)$} & \\
\cmidrule(lr){5-8}\cmidrule(lr){9-12}
$\theta$ & $\tau$ & $\mathcal{M}_k$ & True & RB\% & MCSD & ASE & Cov & RB\% & MCSD & ASE & Cov & RE \\
\midrule
\multirow{12}{*}{1.25}
& \multirow{4}{*}{12} & $\mathcal{M}_1$ & 0.078 & 0.7 & 0.0334 & 0.0340 & 0.957 & 1.1 & 0.0328 & 0.0333 & 0.955 & 1.04 \\
&  & $\mathcal{M}_2$ & 0.078 & 0.7 & 0.0334 & 0.0340 & 0.957 & 1.2 & 0.0329 & 0.0333 & 0.953 & 1.03 \\
&  & $\mathcal{M}_3$ & 0.078 & 0.1 & 0.0332 & 0.0338 & 0.953 & 0.9 & 0.0327 & 0.0332 & 0.959 & 1.03 \\
&  & $\mathcal{M}_4$ & 0.078 & 0.1 & 0.0332 & 0.0338 & 0.953 & 1.0 & 0.0328 & 0.0332 & 0.956 & 1.03 \\
\cmidrule(lr){2-13}
& \multirow{4}{*}{24} & $\mathcal{M}_1$ & 0.106 & $-$0.4 & 0.0414 & 0.0407 & 0.946 & $-$0.4 & 0.0397 & 0.0391 & 0.949 & 1.09 \\
&  & $\mathcal{M}_2$ & 0.106 & $-$0.4 & 0.0414 & 0.0407 & 0.946 & 0.0 & 0.0398 & 0.0393 & 0.947 & 1.09 \\
&  & $\mathcal{M}_3$ & 0.106 & $-$0.9 & 0.0412 & 0.0404 & 0.943 & $-$0.3 & 0.0396 & 0.0390 & 0.946 & 1.08 \\
&  & $\mathcal{M}_4$ & 0.106 & $-$0.9 & 0.0412 & 0.0404 & 0.943 & 0.1 & 0.0396 & 0.0391 & 0.946 & 1.08 \\
\cmidrule(lr){2-13}
& \multirow{4}{*}{36} & $\mathcal{M}_1$ & 0.111 & $-$0.7 & 0.0438 & 0.0438 & 0.944 & $-$0.8 & 0.0416 & 0.0415 & 0.947 & 1.11 \\
&  & $\mathcal{M}_2$ & 0.111 & $-$0.7 & 0.0438 & 0.0438 & 0.944 & $-$0.1 & 0.0418 & 0.0418 & 0.948 & 1.10 \\
&  & $\mathcal{M}_3$ & 0.111 & $-$0.4 & 0.0437 & 0.0435 & 0.946 & $-$0.2 & 0.0416 & 0.0414 & 0.949 & 1.10 \\
&  & $\mathcal{M}_4$ & 0.111 & $-$0.4 & 0.0437 & 0.0435 & 0.946 & 0.5 & 0.0418 & 0.0416 & 0.947 & 1.09 \\
\midrule
\multirow{12}{*}{4.00}
& \multirow{4}{*}{12} & $\mathcal{M}_1$ & 0.082 & 0.6 & 0.0335 & 0.0335 & 0.948 & $-$0.1 & 0.0329 & 0.0328 & 0.946 & 1.04 \\
&  & $\mathcal{M}_2$ & 0.082 & 0.6 & 0.0335 & 0.0335 & 0.948 & 0.2 & 0.0330 & 0.0329 & 0.946 & 1.03 \\
&  & $\mathcal{M}_3$ & 0.082 & $-$0.2 & 0.0333 & 0.0332 & 0.944 & $-$0.3 & 0.0328 & 0.0327 & 0.945 & 1.03 \\
&  & $\mathcal{M}_4$ & 0.082 & $-$0.2 & 0.0333 & 0.0332 & 0.944 & $-$0.1 & 0.0329 & 0.0328 & 0.948 & 1.02 \\
\cmidrule(lr){2-13}
& \multirow{4}{*}{24} & $\mathcal{M}_1$ & 0.117 & 0.3 & 0.0401 & 0.0402 & 0.954 & 0.2 & 0.0394 & 0.0387 & 0.944 & 1.04 \\
&  & $\mathcal{M}_2$ & 0.117 & 0.3 & 0.0401 & 0.0402 & 0.954 & 0.5 & 0.0398 & 0.0391 & 0.945 & 1.02 \\
&  & $\mathcal{M}_3$ & 0.117 & $-$0.6 & 0.0398 & 0.0399 & 0.944 & $-$0.1 & 0.0393 & 0.0386 & 0.944 & 1.03 \\
&  & $\mathcal{M}_4$ & 0.117 & $-$0.6 & 0.0398 & 0.0399 & 0.944 & 0.2 & 0.0396 & 0.0389 & 0.944 & 1.01 \\
\cmidrule(lr){2-13}
& \multirow{4}{*}{36} & $\mathcal{M}_1$ & 0.129 & 0.3 & 0.0438 & 0.0435 & 0.953 & 0.2 & 0.0421 & 0.0413 & 0.946 & 1.08 \\
&  & $\mathcal{M}_2$ & 0.129 & 0.3 & 0.0438 & 0.0435 & 0.953 & 0.6 & 0.0426 & 0.0419 & 0.945 & 1.05 \\
&  & $\mathcal{M}_3$ & 0.129 & 0.0 & 0.0434 & 0.0431 & 0.948 & 0.4 & 0.0421 & 0.0411 & 0.944 & 1.06 \\
&  & $\mathcal{M}_4$ & 0.129 & 0.0 & 0.0434 & 0.0431 & 0.948 & 0.7 & 0.0425 & 0.0417 & 0.945 & 1.04 \\
\bottomrule
\end{tabular}}
\end{table}

\begin{table}[htbp]
\centering
\caption{Simulation results for $\NB(\tau)$ under 40\% censoring by $\tau=36$ for the two component prioritized endpoint.}
\label{tab:sim-NB40}
\resizebox{\textwidth}{!}{%
\begin{tabular}{ccccrrrrrrrrr}
\toprule
 & & & & \multicolumn{4}{c}{$\widehat{\NB}^{(\text{ipcw})}(\tau)$} & \multicolumn{4}{c}{$\widehat{\NB}^{(\text{ctw})}(\tau)$} & \\
\cmidrule(lr){5-8}\cmidrule(lr){9-12}
$\theta$ & $\tau$ & $\mathcal{M}_k$ & True & RB\% & MCSD & ASE & Cov & RB\% & MCSD & ASE & Cov & RE \\
\midrule
\multirow{12}{*}{1.25}
& \multirow{4}{*}{12} & $\mathcal{M}_1$ & 0.078 & 0.4 & 0.0362 & 0.0362 & 0.956 & 0.7 & 0.0345 & 0.0343 & 0.950 & 1.11 \\
&  & $\mathcal{M}_2$ & 0.078 & 0.4 & 0.0362 & 0.0362 & 0.956 & 1.0 & 0.0346 & 0.0344 & 0.949 & 1.09 \\
&  & $\mathcal{M}_3$ & 0.078 & $-$1.2 & 0.0358 & 0.0354 & 0.955 & 0.1 & 0.0343 & 0.0340 & 0.955 & 1.09 \\
&  & $\mathcal{M}_4$ & 0.078 & $-$1.2 & 0.0358 & 0.0354 & 0.955 & 0.4 & 0.0344 & 0.0341 & 0.954 & 1.08 \\
\cmidrule(lr){2-13}
& \multirow{4}{*}{24} & $\mathcal{M}_1$ & 0.106 & 1.2 & 0.0456 & 0.0458 & 0.942 & 0.7 & 0.0408 & 0.0413 & 0.954 & 1.25 \\
&  & $\mathcal{M}_2$ & 0.106 & 1.2 & 0.0456 & 0.0458 & 0.942 & 1.6 & 0.0411 & 0.0417 & 0.951 & 1.23 \\
&  & $\mathcal{M}_3$ & 0.106 & 0.3 & 0.0443 & 0.0442 & 0.949 & 0.8 & 0.0405 & 0.0407 & 0.954 & 1.20 \\
&  & $\mathcal{M}_4$ & 0.106 & 0.3 & 0.0443 & 0.0442 & 0.949 & 1.7 & 0.0408 & 0.0411 & 0.953 & 1.18 \\
\cmidrule(lr){2-13}
& \multirow{4}{*}{36} & $\mathcal{M}_1$ & 0.111 & 1.2 & 0.0515 & 0.0519 & 0.958 & 1.2 & 0.0445 & 0.0450 & 0.954 & 1.34 \\
&  & $\mathcal{M}_2$ & 0.111 & 1.2 & 0.0515 & 0.0519 & 0.958 & 3.0 & 0.0453 & 0.0457 & 0.952 & 1.30 \\
&  & $\mathcal{M}_3$ & 0.111 & 2.2 & 0.0495 & 0.0494 & 0.958 & 2.4 & 0.0438 & 0.0441 & 0.952 & 1.28 \\
&  & $\mathcal{M}_4$ & 0.111 & 2.2 & 0.0495 & 0.0494 & 0.958 & 4.2 & 0.0445 & 0.0447 & 0.952 & 1.24 \\
\midrule
\multirow{12}{*}{4.00}
& \multirow{4}{*}{12} & $\mathcal{M}_1$ & 0.082 & 0.9 & 0.0357 & 0.0356 & 0.955 & 0.6 & 0.0335 & 0.0339 & 0.946 & 1.14 \\
&  & $\mathcal{M}_2$ & 0.082 & 0.9 & 0.0357 & 0.0356 & 0.955 & 0.9 & 0.0339 & 0.0341 & 0.940 & 1.11 \\
&  & $\mathcal{M}_3$ & 0.082 & $-$0.4 & 0.0349 & 0.0348 & 0.950 & 0.1 & 0.0331 & 0.0335 & 0.946 & 1.11 \\
&  & $\mathcal{M}_4$ & 0.082 & $-$0.4 & 0.0349 & 0.0348 & 0.950 & 0.3 & 0.0335 & 0.0338 & 0.945 & 1.08 \\
\cmidrule(lr){2-13}
& \multirow{4}{*}{24} & $\mathcal{M}_1$ & 0.117 & 0.4 & 0.0445 & 0.0453 & 0.951 & 0.0 & 0.0412 & 0.0410 & 0.948 & 1.17 \\
&  & $\mathcal{M}_2$ & 0.117 & 0.4 & 0.0445 & 0.0453 & 0.951 & 0.4 & 0.0423 & 0.0419 & 0.944 & 1.10 \\
&  & $\mathcal{M}_3$ & 0.117 & $-$1.4 & 0.0429 & 0.0436 & 0.951 & $-$0.4 & 0.0404 & 0.0404 & 0.945 & 1.13 \\
&  & $\mathcal{M}_4$ & 0.117 & $-$1.4 & 0.0429 & 0.0436 & 0.951 & 0.0 & 0.0414 & 0.0412 & 0.939 & 1.07 \\
\cmidrule(lr){2-13}
& \multirow{4}{*}{36} & $\mathcal{M}_1$ & 0.129 & 0.3 & 0.0501 & 0.0516 & 0.965 & 0.3 & 0.0445 & 0.0448 & 0.950 & 1.27 \\
&  & $\mathcal{M}_2$ & 0.129 & 0.3 & 0.0501 & 0.0516 & 0.965 & 0.9 & 0.0468 & 0.0465 & 0.944 & 1.15 \\
&  & $\mathcal{M}_3$ & 0.129 & $-$0.6 & 0.0484 & 0.0489 & 0.954 & 0.5 & 0.0439 & 0.0439 & 0.951 & 1.22 \\
&  & $\mathcal{M}_4$ & 0.129 & $-$0.6 & 0.0484 & 0.0489 & 0.954 & 1.3 & 0.0458 & 0.0454 & 0.949 & 1.12 \\
\bottomrule
\end{tabular}}
\end{table}

\begin{table}[htbp]
\centering
\caption{Simulation results for $\NB(\tau)$ under 60\% censoring by $\tau=36$ for the two component prioritized endpoint.}
\label{tab:sim-NB60}
\resizebox{\textwidth}{!}{%
\begin{tabular}{ccccrrrrrrrrr}
\toprule
 & & & & \multicolumn{4}{c}{$\widehat{\NB}^{(\text{ipcw})}(\tau)$} & \multicolumn{4}{c}{$\widehat{\NB}^{(\text{ctw})}(\tau)$} & \\
\cmidrule(lr){5-8}\cmidrule(lr){9-12}
$\theta$ & $\tau$ & $\mathcal{M}_k$ & True & RB\% & MCSD & ASE & Cov & RB\% & MCSD & ASE & Cov & RE \\
\midrule
\multirow{12}{*}{1.25}
& \multirow{4}{*}{12} & $\mathcal{M}_1$ & 0.078 & $-$0.8 & 0.0401 & 0.0401 & 0.939 & 0.6 & 0.0364 & 0.0360 & 0.937 & 1.22 \\
&  & $\mathcal{M}_2$ & 0.078 & $-$0.8 & 0.0401 & 0.0401 & 0.939 & 1.2 & 0.0367 & 0.0362 & 0.939 & 1.20 \\
&  & $\mathcal{M}_3$ & 0.078 & $-$3.2 & 0.0383 & 0.0380 & 0.943 & $-$0.7 & 0.0357 & 0.0352 & 0.938 & 1.15 \\
&  & $\mathcal{M}_4$ & 0.078 & $-$3.2 & 0.0383 & 0.0380 & 0.943 & $-$0.2 & 0.0359 & 0.0354 & 0.939 & 1.14 \\
\cmidrule(lr){2-13}
& \multirow{4}{*}{24} & $\mathcal{M}_1$ & 0.106 & 2.2 & 0.0567 & 0.0564 & 0.946 & 0.8 & 0.0472 & 0.0456 & 0.947 & 1.44 \\
&  & $\mathcal{M}_2$ & 0.106 & 2.2 & 0.0567 & 0.0564 & 0.946 & 2.3 & 0.0478 & 0.0463 & 0.947 & 1.40 \\
&  & $\mathcal{M}_3$ & 0.106 & $-$0.1 & 0.0517 & 0.0504 & 0.941 & 0.5 & 0.0453 & 0.0436 & 0.939 & 1.31 \\
&  & $\mathcal{M}_4$ & 0.106 & $-$0.1 & 0.0517 & 0.0504 & 0.941 & 2.0 & 0.0458 & 0.0442 & 0.937 & 1.28 \\
\cmidrule(lr){2-13}
& \multirow{4}{*}{36} & $\mathcal{M}_1$ & 0.111 & 1.9 & 0.0715 & 0.0714 & 0.969 & 0.5 & 0.0542 & 0.0530 & 0.950 & 1.74 \\
&  & $\mathcal{M}_2$ & 0.111 & 1.9 & 0.0715 & 0.0714 & 0.969 & 3.8 & 0.0554 & 0.0544 & 0.950 & 1.66 \\
&  & $\mathcal{M}_3$ & 0.111 & 1.6 & 0.0582 & 0.0591 & 0.959 & 2.2 & 0.0494 & 0.0487 & 0.948 & 1.39 \\
&  & $\mathcal{M}_4$ & 0.111 & 1.6 & 0.0582 & 0.0591 & 0.959 & 5.4 & 0.0504 & 0.0498 & 0.946 & 1.33 \\
\midrule
\multirow{12}{*}{4.00}
& \multirow{4}{*}{12} & $\mathcal{M}_1$ & 0.082 & 0.4 & 0.0387 & 0.0393 & 0.957 & 0.5 & 0.0344 & 0.0356 & 0.955 & 1.26 \\
&  & $\mathcal{M}_2$ & 0.082 & 0.4 & 0.0387 & 0.0393 & 0.957 & 0.9 & 0.0351 & 0.0361 & 0.957 & 1.22 \\
&  & $\mathcal{M}_3$ & 0.082 & $-$3.1 & 0.0369 & 0.0373 & 0.948 & $-$0.9 & 0.0339 & 0.0348 & 0.956 & 1.19 \\
&  & $\mathcal{M}_4$ & 0.082 & $-$3.1 & 0.0369 & 0.0373 & 0.948 & $-$0.4 & 0.0344 & 0.0352 & 0.953 & 1.15 \\
\cmidrule(lr){2-13}
& \multirow{4}{*}{24} & $\mathcal{M}_1$ & 0.117 & $-$0.3 & 0.0540 & 0.0558 & 0.955 & 0.7 & 0.0434 & 0.0454 & 0.961 & 1.55 \\
&  & $\mathcal{M}_2$ & 0.117 & $-$0.3 & 0.0540 & 0.0558 & 0.955 & 1.7 & 0.0454 & 0.0471 & 0.952 & 1.41 \\
&  & $\mathcal{M}_3$ & 0.117 & $-$3.6 & 0.0483 & 0.0496 & 0.958 & $-$0.4 & 0.0414 & 0.0433 & 0.960 & 1.36 \\
&  & $\mathcal{M}_4$ & 0.117 & $-$3.6 & 0.0483 & 0.0496 & 0.958 & 0.8 & 0.0430 & 0.0448 & 0.960 & 1.26 \\
\cmidrule(lr){2-13}
& \multirow{4}{*}{36} & $\mathcal{M}_1$ & 0.129 & $-$1.2 & 0.0686 & 0.0712 & 0.959 & 0.3 & 0.0499 & 0.0529 & 0.971 & 1.89 \\
&  & $\mathcal{M}_2$ & 0.129 & $-$1.2 & 0.0686 & 0.0712 & 0.959 & 1.9 & 0.0534 & 0.0562 & 0.958 & 1.65 \\
&  & $\mathcal{M}_3$ & 0.129 & $-$2.6 & 0.0557 & 0.0586 & 0.966 & 0.3 & 0.0457 & 0.0486 & 0.970 & 1.49 \\
&  & $\mathcal{M}_4$ & 0.129 & $-$2.6 & 0.0557 & 0.0586 & 0.966 & 2.4 & 0.0485 & 0.0513 & 0.968 & 1.32 \\
\bottomrule
\end{tabular}}
\end{table}

\subsection{Nuisance model misspecification study: win ratio and win odds}
\label{supp:sim_original_WRWO}

Tables \ref{tab:sim-WR20}--\ref{tab:sim-WR80} and Tables \ref{tab:sim-WO20}--\ref{tab:sim-WO80} report the corresponding results for $\WR(\tau)$ and $\WO(\tau)$ in the original two-component nuisance-model misspecification study. These analyses use the same simulated data sets as the net-benefit analysis, with death as the first-priority component and the serious nonfatal event as the second-priority component. Because inference for $\WR(\tau)$ and $\WO(\tau)$ is performed on the log scale, MCSD, ASE, COV, and RE are summarized on the log scale.

The results for $\WR(\tau)$ and $\WO(\tau)$ agree with the net-benefit findings. Under light censoring, both estimators have negligible bias, ASE tracks MCSD well, and empirical coverage is close to the nominal level. As censoring increases, the proposed estimator becomes increasingly more efficient than the IPCW estimator of \citet{cui2025ipcw}, especially at longer restriction horizons. This pattern is expected because heavier censoring and longer follow-up create more partially observed lower-priority comparisons, which are discarded by IPCW of \citet{cui2025ipcw} but used fractionally by conditional tie weighting. The same robustness pattern is also observed. The event-time model misspecification attenuates the efficiency gain, whereas censoring-model misspecification has a larger effect on bias and coverage because the inverse probability weights are affected for both estimators.

\begin{table}[htbp]
\centering
\caption{Simulation results for $\WR(\tau)$ under 20\% censoring by $\tau=36$ for the two-component prioritized endpoint. RB: relative bias (\%) for $\widehat{\WR}(\tau)$ on the natural scale; MCSD: Monte Carlo standard deviation of $\log\widehat{\WR}(\tau)$; ASE: average estimated standard error for $\log\widehat{\WR}(\tau)$; Cov: empirical coverage of the 95\% Wald interval based on log-scale inference; RE $=\mathrm{MCSD}^2(C)/\mathrm{MCSD}^2(I)$ computed on the log scale.}
\label{tab:sim-WR20}
\resizebox{\textwidth}{!}{%
\begin{tabular}{ccccrrrrrrrrr}
\toprule
 & & & & \multicolumn{4}{c}{$\widehat{\WR}^{(\text{ipcw})}(\tau)$} & \multicolumn{4}{c}{$\widehat{\WR}^{(\text{ctw})}(\tau)$} & \\
\cmidrule(lr){5-8}\cmidrule(lr){9-12}
$\theta$ & $\tau$ & $\mathcal{M}_k$ & True & RB\% & MCSD & ASE & Cov & RB\% & MCSD & ASE & Cov & RE \\
\midrule
\multirow{12}{*}{1.25}& \multirow{4}{*}{12} & $\mathcal{M}_1$ & 1.370 & 1.3 & 0.1359 & 0.1384 & 0.961 & 1.4 & 0.1335 & 0.1355 & 0.956 & 1.04 \\
&  & $\mathcal{M}_2$ & 1.370 & 1.3 & 0.1359 & 0.1384 & 0.961 & 1.4 & 0.1334 & 0.1354 & 0.956 & 1.04 \\
&  & $\mathcal{M}_3$ & 1.370 & 1.4 & 0.1360 & 0.1383 & 0.963 & 1.5 & 0.1334 & 0.1355 & 0.959 & 1.04 \\
&  & $\mathcal{M}_4$ & 1.370 & 1.4 & 0.1360 & 0.1383 & 0.963 & 1.5 & 0.1334 & 0.1354 & 0.958 & 1.04 \\
\cmidrule(lr){2-13}
& \multirow{4}{*}{24} & $\mathcal{M}_1$ & 1.331 & 0.7 & 0.1136 & 0.1114 & 0.945 & 0.6 & 0.1088 & 0.1072 & 0.946 & 1.09 \\
&  & $\mathcal{M}_2$ & 1.331 & 0.7 & 0.1136 & 0.1114 & 0.945 & 0.6 & 0.1085 & 0.1070 & 0.946 & 1.10 \\
&  & $\mathcal{M}_3$ & 1.331 & 0.7 & 0.1137 & 0.1113 & 0.944 & 0.7 & 0.1089 & 0.1072 & 0.946 & 1.09 \\
&  & $\mathcal{M}_4$ & 1.331 & 0.7 & 0.1137 & 0.1113 & 0.944 & 0.7 & 0.1086 & 0.1070 & 0.946 & 1.10 \\
\cmidrule(lr){2-13}
& \multirow{4}{*}{36} & $\mathcal{M}_1$ & 1.293 & 0.4 & 0.1030 & 0.1024 & 0.947 & 0.3 & 0.0978 & 0.0974 & 0.949 & 1.11 \\
&  & $\mathcal{M}_2$ & 1.293 & 0.4 & 0.1030 & 0.1024 & 0.947 & 0.3 & 0.0975 & 0.0971 & 0.949 & 1.12 \\
&  & $\mathcal{M}_3$ & 1.293 & 0.6 & 0.1031 & 0.1021 & 0.949 & 0.6 & 0.0980 & 0.0973 & 0.946 & 1.11 \\
&  & $\mathcal{M}_4$ & 1.293 & 0.6 & 0.1031 & 0.1021 & 0.949 & 0.5 & 0.0976 & 0.0970 & 0.945 & 1.12 \\
\midrule
\multirow{12}{*}{4.00}& \multirow{4}{*}{12} & $\mathcal{M}_1$ & 1.408 & 1.4 & 0.1422 & 0.1417 & 0.948 & 1.1 & 0.1393 & 0.1389 & 0.945 & 1.04 \\
&  & $\mathcal{M}_2$ & 1.408 & 1.4 & 0.1422 & 0.1417 & 0.948 & 1.0 & 0.1391 & 0.1387 & 0.947 & 1.04 \\
&  & $\mathcal{M}_3$ & 1.408 & 1.4 & 0.1425 & 0.1417 & 0.947 & 1.1 & 0.1395 & 0.1390 & 0.946 & 1.04 \\
&  & $\mathcal{M}_4$ & 1.408 & 1.4 & 0.1425 & 0.1417 & 0.947 & 1.0 & 0.1393 & 0.1387 & 0.949 & 1.05 \\
\cmidrule(lr){2-13}
& \multirow{4}{*}{24} & $\mathcal{M}_1$ & 1.389 & 0.8 & 0.1143 & 0.1144 & 0.956 & 0.7 & 0.1122 & 0.1103 & 0.945 & 1.04 \\
&  & $\mathcal{M}_2$ & 1.389 & 0.8 & 0.1143 & 0.1144 & 0.956 & 0.4 & 0.1117 & 0.1099 & 0.943 & 1.05 \\
&  & $\mathcal{M}_3$ & 1.389 & 0.8 & 0.1144 & 0.1143 & 0.951 & 0.8 & 0.1124 & 0.1104 & 0.946 & 1.04 \\
&  & $\mathcal{M}_4$ & 1.389 & 0.8 & 0.1144 & 0.1143 & 0.951 & 0.5 & 0.1117 & 0.1099 & 0.942 & 1.05 \\
\cmidrule(lr){2-13}
& \multirow{4}{*}{36} & $\mathcal{M}_1$ & 1.363 & 0.7 & 0.1063 & 0.1051 & 0.951 & 0.6 & 0.1023 & 0.1001 & 0.945 & 1.08 \\
&  & $\mathcal{M}_2$ & 1.363 & 0.7 & 0.1063 & 0.1051 & 0.951 & 0.1 & 0.1014 & 0.0995 & 0.945 & 1.10 \\
&  & $\mathcal{M}_3$ & 1.363 & 0.8 & 0.1062 & 0.1047 & 0.947 & 0.8 & 0.1026 & 0.1001 & 0.943 & 1.07 \\
&  & $\mathcal{M}_4$ & 1.363 & 0.8 & 0.1062 & 0.1047 & 0.947 & 0.2 & 0.1015 & 0.0994 & 0.945 & 1.09 \\
\bottomrule
\end{tabular}}
\end{table}

\begin{table}[htbp]
\centering
\caption{Simulation results for $\WR(\tau)$ under 40\% censoring by $\tau=36$ for the two-component prioritized endpoint. RB: relative bias (\%) for $\widehat{\WR}(\tau)$ on the natural scale; MCSD: Monte Carlo standard deviation of $\log\widehat{\WR}(\tau)$; ASE: average estimated standard error for $\log\widehat{\WR}(\tau)$; Cov: empirical coverage of the 95\% Wald interval based on log-scale inference; RE $=\mathrm{MCSD}^2(C)/\mathrm{MCSD}^2(I)$ computed on the log scale.}
\label{tab:sim-WR40}
\resizebox{\textwidth}{!}{%
\begin{tabular}{ccccrrrrrrrrr}
\toprule
 & & & & \multicolumn{4}{c}{$\widehat{\WR}^{(\text{ipcw})}(\tau)$} & \multicolumn{4}{c}{$\widehat{\WR}^{(\text{ctw})}(\tau)$} & \\
\cmidrule(lr){5-8}\cmidrule(lr){9-12}
$\theta$ & $\tau$ & $\mathcal{M}_k$ & True & RB\% & MCSD & ASE & Cov & RB\% & MCSD & ASE & Cov & RE \\
\midrule
\multirow{12}{*}{1.25}& \multirow{4}{*}{12} & $\mathcal{M}_1$ & 1.370 & 1.4 & 0.1481 & 0.1467 & 0.956 & 1.4 & 0.1410 & 0.1393 & 0.958 & 1.10 \\
&  & $\mathcal{M}_2$ & 1.370 & 1.4 & 0.1481 & 0.1467 & 0.956 & 1.4 & 0.1410 & 0.1392 & 0.957 & 1.10 \\
&  & $\mathcal{M}_3$ & 1.370 & 1.5 & 0.1486 & 0.1460 & 0.955 & 1.5 & 0.1413 & 0.1391 & 0.961 & 1.11 \\
&  & $\mathcal{M}_4$ & 1.370 & 1.5 & 0.1486 & 0.1460 & 0.955 & 1.5 & 0.1413 & 0.1390 & 0.961 & 1.11 \\
\cmidrule(lr){2-13}
& \multirow{4}{*}{24} & $\mathcal{M}_1$ & 1.331 & 1.3 & 0.1253 & 0.1248 & 0.939 & 0.9 & 0.1122 & 0.1130 & 0.954 & 1.25 \\
&  & $\mathcal{M}_2$ & 1.331 & 1.3 & 0.1253 & 0.1248 & 0.939 & 0.9 & 0.1118 & 0.1127 & 0.954 & 1.26 \\
&  & $\mathcal{M}_3$ & 1.331 & 1.5 & 0.1241 & 0.1226 & 0.949 & 1.2 & 0.1121 & 0.1124 & 0.950 & 1.23 \\
&  & $\mathcal{M}_4$ & 1.331 & 1.5 & 0.1241 & 0.1226 & 0.949 & 1.1 & 0.1117 & 0.1120 & 0.946 & 1.24 \\
\cmidrule(lr){2-13}
& \multirow{4}{*}{36} & $\mathcal{M}_1$ & 1.293 & 1.2 & 0.1213 & 0.1206 & 0.955 & 1.0 & 0.1048 & 0.1055 & 0.952 & 1.34 \\
&  & $\mathcal{M}_2$ & 1.293 & 1.2 & 0.1213 & 0.1206 & 0.955 & 0.9 & 0.1045 & 0.1050 & 0.952 & 1.35 \\
&  & $\mathcal{M}_3$ & 1.293 & 1.7 & 0.1180 & 0.1163 & 0.957 & 1.4 & 0.1038 & 0.1040 & 0.952 & 1.29 \\
&  & $\mathcal{M}_4$ & 1.293 & 1.7 & 0.1180 & 0.1163 & 0.957 & 1.4 & 0.1034 & 0.1033 & 0.954 & 1.30 \\
\midrule
\multirow{12}{*}{4.00}& \multirow{4}{*}{12} & $\mathcal{M}_1$ & 1.408 & 1.6 & 0.1511 & 0.1503 & 0.958 & 1.3 & 0.1414 & 0.1430 & 0.947 & 1.14 \\
&  & $\mathcal{M}_2$ & 1.408 & 1.6 & 0.1511 & 0.1503 & 0.958 & 1.0 & 0.1418 & 0.1425 & 0.948 & 1.14 \\
&  & $\mathcal{M}_3$ & 1.408 & 1.7 & 0.1501 & 0.1497 & 0.954 & 1.4 & 0.1410 & 0.1429 & 0.949 & 1.13 \\
&  & $\mathcal{M}_4$ & 1.408 & 1.7 & 0.1501 & 0.1497 & 0.954 & 1.1 & 0.1412 & 0.1423 & 0.946 & 1.13 \\
\cmidrule(lr){2-13}
& \multirow{4}{*}{24} & $\mathcal{M}_1$ & 1.389 & 1.1 & 0.1265 & 0.1281 & 0.953 & 0.7 & 0.1171 & 0.1165 & 0.944 & 1.17 \\
&  & $\mathcal{M}_2$ & 1.389 & 1.1 & 0.1265 & 0.1281 & 0.953 & $-$0.1 & 0.1169 & 0.1156 & 0.945 & 1.17 \\
&  & $\mathcal{M}_3$ & 1.389 & 1.1 & 0.1244 & 0.1259 & 0.949 & 0.9 & 0.1161 & 0.1160 & 0.942 & 1.15 \\
&  & $\mathcal{M}_4$ & 1.389 & 1.1 & 0.1244 & 0.1259 & 0.949 & 0.1 & 0.1156 & 0.1150 & 0.940 & 1.16 \\
\cmidrule(lr){2-13}
& \multirow{4}{*}{36} & $\mathcal{M}_1$ & 1.363 & 1.0 & 0.1219 & 0.1238 & 0.955 & 0.7 & 0.1080 & 0.1084 & 0.951 & 1.27 \\
&  & $\mathcal{M}_2$ & 1.363 & 1.0 & 0.1219 & 0.1238 & 0.955 & $-$0.5 & 0.1085 & 0.1074 & 0.952 & 1.26 \\
&  & $\mathcal{M}_3$ & 1.363 & 1.1 & 0.1194 & 0.1193 & 0.948 & 1.0 & 0.1073 & 0.1071 & 0.949 & 1.24 \\
&  & $\mathcal{M}_4$ & 1.363 & 1.1 & 0.1194 & 0.1193 & 0.948 & $-$0.1 & 0.1071 & 0.1058 & 0.940 & 1.24 \\
\bottomrule
\end{tabular}}
\end{table}

\begin{table}[htbp]
\centering
\caption{Simulation results for $\WR(\tau)$ under 60\% censoring by $\tau=36$ for the two-component prioritized endpoint. RB: relative bias (\%) for $\widehat{\WR}(\tau)$ on the natural scale; MCSD: Monte Carlo standard deviation of $\log\widehat{\WR}(\tau)$; ASE: average estimated standard error for $\log\widehat{\WR}(\tau)$; Cov: empirical coverage of the 95\% Wald interval based on log-scale inference; RE $=\mathrm{MCSD}^2(C)/\mathrm{MCSD}^2(I)$ computed on the log scale.}
\label{tab:sim-WR60}
\resizebox{\textwidth}{!}{%
\begin{tabular}{ccccrrrrrrrrr}
\toprule
 & & & & \multicolumn{4}{c}{$\widehat{\WR}^{(\text{ipcw})}(\tau)$} & \multicolumn{4}{c}{$\widehat{\WR}^{(\text{ctw})}(\tau)$} & \\
\cmidrule(lr){5-8}\cmidrule(lr){9-12}
$\theta$ & $\tau$ & $\mathcal{M}_k$ & True & RB\% & MCSD & ASE & Cov & RB\% & MCSD & ASE & Cov & RE \\
\midrule
\multirow{12}{*}{1.25}& \multirow{4}{*}{12} & $\mathcal{M}_1$ & 1.370 & 1.3 & 0.1643 & 0.1617 & 0.940 & 1.4 & 0.1480 & 0.1457 & 0.940 & 1.23 \\
&  & $\mathcal{M}_2$ & 1.370 & 1.3 & 0.1643 & 0.1617 & 0.940 & 1.3 & 0.1479 & 0.1454 & 0.943 & 1.23 \\
&  & $\mathcal{M}_3$ & 1.370 & 1.4 & 0.1614 & 0.1583 & 0.943 & 1.5 & 0.1473 & 0.1447 & 0.936 & 1.20 \\
&  & $\mathcal{M}_4$ & 1.370 & 1.4 & 0.1614 & 0.1583 & 0.943 & 1.4 & 0.1470 & 0.1444 & 0.937 & 1.21 \\
\cmidrule(lr){2-13}
& \multirow{4}{*}{24} & $\mathcal{M}_1$ & 1.331 & 2.1 & 0.1556 & 0.1528 & 0.940 & 1.2 & 0.1295 & 0.1245 & 0.945 & 1.44 \\
&  & $\mathcal{M}_2$ & 1.331 & 2.1 & 0.1556 & 0.1528 & 0.940 & 1.1 & 0.1285 & 0.1240 & 0.945 & 1.46 \\
&  & $\mathcal{M}_3$ & 1.331 & 2.2 & 0.1463 & 0.1413 & 0.940 & 1.5 & 0.1259 & 0.1210 & 0.937 & 1.35 \\
&  & $\mathcal{M}_4$ & 1.331 & 2.2 & 0.1463 & 0.1413 & 0.940 & 1.4 & 0.1249 & 0.1203 & 0.936 & 1.37 \\
\cmidrule(lr){2-13}
& \multirow{4}{*}{36} & $\mathcal{M}_1$ & 1.293 & 2.3 & 0.1667 & 0.1635 & 0.951 & 1.2 & 0.1265 & 0.1237 & 0.945 & 1.74 \\
&  & $\mathcal{M}_2$ & 1.293 & 2.3 & 0.1667 & 0.1635 & 0.951 & 1.1 & 0.1252 & 0.1226 & 0.949 & 1.77 \\
&  & $\mathcal{M}_3$ & 1.293 & 2.2 & 0.1400 & 0.1401 & 0.956 & 1.7 & 0.1176 & 0.1153 & 0.950 & 1.42 \\
&  & $\mathcal{M}_4$ & 1.293 & 2.2 & 0.1400 & 0.1401 & 0.956 & 1.6 & 0.1162 & 0.1142 & 0.945 & 1.45 \\
\midrule
\multirow{12}{*}{4.00}& \multirow{4}{*}{12} & $\mathcal{M}_1$ & 1.408 & 2.0 & 0.1670 & 0.1664 & 0.950 & 1.5 & 0.1470 & 0.1504 & 0.962 & 1.29 \\
&  & $\mathcal{M}_2$ & 1.408 & 2.0 & 0.1670 & 0.1664 & 0.950 & 0.9 & 0.1471 & 0.1495 & 0.959 & 1.29 \\
&  & $\mathcal{M}_3$ & 1.408 & 1.8 & 0.1649 & 0.1629 & 0.947 & 1.6 & 0.1475 & 0.1494 & 0.960 & 1.25 \\
&  & $\mathcal{M}_4$ & 1.408 & 1.8 & 0.1649 & 0.1629 & 0.947 & 1.0 & 0.1472 & 0.1484 & 0.957 & 1.25 \\
\cmidrule(lr){2-13}
& \multirow{4}{*}{24} & $\mathcal{M}_1$ & 1.389 & 1.5 & 0.1554 & 0.1570 & 0.945 & 1.1 & 0.1237 & 0.1286 & 0.957 & 1.58 \\
&  & $\mathcal{M}_2$ & 1.389 & 1.5 & 0.1554 & 0.1570 & 0.945 & $-$0.2 & 0.1238 & 0.1273 & 0.953 & 1.58 \\
&  & $\mathcal{M}_3$ & 1.389 & 1.3 & 0.1447 & 0.1454 & 0.953 & 1.3 & 0.1207 & 0.1253 & 0.959 & 1.44 \\
&  & $\mathcal{M}_4$ & 1.389 & 1.3 & 0.1447 & 0.1454 & 0.953 & 0.1 & 0.1198 & 0.1236 & 0.955 & 1.46 \\
\cmidrule(lr){2-13}
& \multirow{4}{*}{36} & $\mathcal{M}_1$ & 1.363 & 1.4 & 0.1664 & 0.1682 & 0.950 & 0.9 & 0.1203 & 0.1271 & 0.971 & 1.91 \\
&  & $\mathcal{M}_2$ & 1.363 & 1.4 & 0.1664 & 0.1682 & 0.950 & $-$0.9 & 0.1200 & 0.1255 & 0.964 & 1.92 \\
&  & $\mathcal{M}_3$ & 1.363 & 1.2 & 0.1400 & 0.1439 & 0.966 & 1.2 & 0.1122 & 0.1188 & 0.971 & 1.56 \\
&  & $\mathcal{M}_4$ & 1.363 & 1.2 & 0.1400 & 0.1439 & 0.966 & $-$0.4 & 0.1109 & 0.1168 & 0.959 & 1.59 \\
\bottomrule
\end{tabular}}
\end{table}

\begin{table}[htbp]
\centering
\caption{Simulation results for $\WR(\tau)$ under 80\% censoring by $\tau=36$ for the two-component prioritized endpoint. RB: relative bias (\%) for $\widehat{\WR}(\tau)$ on the natural scale; MCSD: Monte Carlo standard deviation of $\log\widehat{\WR}(\tau)$; ASE: average estimated standard error for $\log\widehat{\WR}(\tau)$; Cov: empirical coverage of the 95\% Wald interval based on log-scale inference; RE $=\mathrm{MCSD}^2(C)/\mathrm{MCSD}^2(I)$ computed on the log scale.}
\label{tab:sim-WR80}
\resizebox{\textwidth}{!}{%
\begin{tabular}{ccccrrrrrrrrr}
\toprule
 & & & & \multicolumn{4}{c}{$\widehat{\WR}^{(\text{ipcw})}(\tau)$} & \multicolumn{4}{c}{$\widehat{\WR}^{(\text{ctw})}(\tau)$} & \\
\cmidrule(lr){5-8}\cmidrule(lr){9-12}
$\theta$ & $\tau$ & $\mathcal{M}_k$ & True & RB\% & MCSD & ASE & Cov & RB\% & MCSD & ASE & Cov & RE \\
\midrule
\multirow{12}{*}{1.25}& \multirow{4}{*}{12} & $\mathcal{M}_1$ & 1.370 & 2.6 & 0.2135 & 0.2030 & 0.935 & 2.0 & 0.1673 & 0.1631 & 0.946 & 1.63 \\
&  & $\mathcal{M}_2$ & 1.370 & 2.6 & 0.2135 & 0.2030 & 0.935 & 1.8 & 0.1669 & 0.1626 & 0.943 & 1.64 \\
&  & $\mathcal{M}_3$ & 1.370 & 2.0 & 0.1927 & 0.1853 & 0.942 & 1.9 & 0.1631 & 0.1575 & 0.941 & 1.40 \\
&  & $\mathcal{M}_4$ & 1.370 & 2.0 & 0.1927 & 0.1853 & 0.942 & 1.7 & 0.1623 & 0.1568 & 0.939 & 1.41 \\
\cmidrule(lr){2-13}
& \multirow{4}{*}{24} & $\mathcal{M}_1$ & 1.331 & 4.4 & 0.2611 & 0.2416 & 0.944 & 2.7 & 0.1687 & 0.1606 & 0.946 & 2.39 \\
&  & $\mathcal{M}_2$ & 1.331 & 4.4 & 0.2611 & 0.2416 & 0.944 & 2.5 & 0.1669 & 0.1591 & 0.949 & 2.45 \\
&  & $\mathcal{M}_3$ & 1.331 & 3.0 & 0.1873 & 0.1828 & 0.945 & 2.6 & 0.1447 & 0.1400 & 0.938 & 1.68 \\
&  & $\mathcal{M}_4$ & 1.331 & 3.0 & 0.1873 & 0.1828 & 0.945 & 2.4 & 0.1430 & 0.1386 & 0.946 & 1.72 \\
\cmidrule(lr){2-13}
& \multirow{4}{*}{36} & $\mathcal{M}_1$ & 1.293 & 5.5 & 0.3279 & 0.2920 & 0.949 & 3.3 & 0.1998 & 0.1826 & 0.943 & 2.69 \\
&  & $\mathcal{M}_2$ & 1.293 & 5.5 & 0.3279 & 0.2920 & 0.949 & 3.2 & 0.1947 & 0.1789 & 0.943 & 2.84 \\
&  & $\mathcal{M}_3$ & 1.293 & 3.6 & 0.1988 & 0.1962 & 0.946 & 2.6 & 0.1440 & 0.1414 & 0.942 & 1.91 \\
&  & $\mathcal{M}_4$ & 1.293 & 3.6 & 0.1988 & 0.1962 & 0.946 & 2.5 & 0.1409 & 0.1389 & 0.948 & 1.99 \\
\midrule
\multirow{12}{*}{4.00}& \multirow{4}{*}{12} & $\mathcal{M}_1$ & 1.408 & 2.3 & 0.2031 & 0.2089 & 0.962 & 0.9 & 0.1612 & 0.1685 & 0.962 & 1.59 \\
&  & $\mathcal{M}_2$ & 1.408 & 2.3 & 0.2031 & 0.2089 & 0.962 & $-$0.0 & 0.1605 & 0.1668 & 0.956 & 1.60 \\
&  & $\mathcal{M}_3$ & 1.408 & 2.1 & 0.1863 & 0.1905 & 0.959 & 1.1 & 0.1565 & 0.1627 & 0.958 & 1.42 \\
&  & $\mathcal{M}_4$ & 1.408 & 2.1 & 0.1863 & 0.1905 & 0.959 & 0.2 & 0.1553 & 0.1606 & 0.955 & 1.44 \\
\cmidrule(lr){2-13}
& \multirow{4}{*}{24} & $\mathcal{M}_1$ & 1.389 & 2.6 & 0.2587 & 0.2476 & 0.945 & 0.1 & 0.1656 & 0.1673 & 0.953 & 2.44 \\
&  & $\mathcal{M}_2$ & 1.389 & 2.6 & 0.2587 & 0.2476 & 0.945 & $-$1.6 & 0.1626 & 0.1633 & 0.952 & 2.53 \\
&  & $\mathcal{M}_3$ & 1.389 & 1.9 & 0.1853 & 0.1879 & 0.951 & 0.7 & 0.1426 & 0.1456 & 0.950 & 1.69 \\
&  & $\mathcal{M}_4$ & 1.389 & 1.9 & 0.1853 & 0.1879 & 0.951 & $-$0.9 & 0.1399 & 0.1420 & 0.946 & 1.75 \\
\cmidrule(lr){2-13}
& \multirow{4}{*}{36} & $\mathcal{M}_1$ & 1.363 & 4.5 & 0.3295 & 0.2900 & 0.937 & 0.6 & 0.1969 & 0.1888 & 0.966 & 2.80 \\
&  & $\mathcal{M}_2$ & 1.363 & 4.5 & 0.3295 & 0.2900 & 0.937 & $-$1.7 & 0.1893 & 0.1805 & 0.954 & 3.03 \\
&  & $\mathcal{M}_3$ & 1.363 & 2.6 & 0.2020 & 0.2013 & 0.957 & 0.8 & 0.1428 & 0.1464 & 0.954 & 2.00 \\
&  & $\mathcal{M}_4$ & 1.363 & 2.6 & 0.2020 & 0.2013 & 0.957 & $-$1.2 & 0.1383 & 0.1416 & 0.952 & 2.13 \\
\bottomrule
\end{tabular}}
\end{table}

\begin{table}[htbp]
\centering
\caption{Simulation results for $\WO(\tau)$ under 20\% censoring by $\tau=36$ for the two-component prioritized endpoint. RB: relative bias (\%) for $\widehat{\WO}(\tau)$ on the natural scale; MCSD: Monte Carlo standard deviation of $\log\widehat{\WO}(\tau)$; ASE: average estimated standard error for $\log\widehat{\WO}(\tau)$; Cov: empirical coverage of the 95\% Wald interval based on log-scale inference; RE $=\mathrm{MCSD}^2(C)/\mathrm{MCSD}^2(I)$ computed on the log scale.}
\label{tab:sim-WO20}
\resizebox{\textwidth}{!}{%
\begin{tabular}{ccccrrrrrrrrr}
\toprule
 & & & & \multicolumn{4}{c}{$\widehat{\WO}^{(\text{ipcw})}(\tau)$} & \multicolumn{4}{c}{$\widehat{\WO}^{(\text{ctw})}(\tau)$} & \\
\cmidrule(lr){5-8}\cmidrule(lr){9-12}
$\theta$ & $\tau$ & $\mathcal{M}_k$ & True & RB\% & MCSD & ASE & Cov & RB\% & MCSD & ASE & Cov & RE \\
\midrule
\multirow{12}{*}{1.25}& \multirow{4}{*}{12} & $\mathcal{M}_1$ & 1.170 & 0.4 & 0.0673 & 0.0686 & 0.957 & 0.4 & 0.0661 & 0.0671 & 0.956 & 1.04 \\
&  & $\mathcal{M}_2$ & 1.170 & 0.4 & 0.0673 & 0.0686 & 0.957 & 0.4 & 0.0662 & 0.0672 & 0.954 & 1.03 \\
&  & $\mathcal{M}_3$ & 1.170 & 0.3 & 0.0669 & 0.0681 & 0.958 & 0.4 & 0.0659 & 0.0669 & 0.959 & 1.03 \\
&  & $\mathcal{M}_4$ & 1.170 & 0.3 & 0.0669 & 0.0681 & 0.958 & 0.4 & 0.0660 & 0.0670 & 0.958 & 1.03 \\
\cmidrule(lr){2-13}
& \multirow{4}{*}{24} & $\mathcal{M}_1$ & 1.237 & 0.3 & 0.0840 & 0.0825 & 0.946 & 0.3 & 0.0804 & 0.0793 & 0.949 & 1.09 \\
&  & $\mathcal{M}_2$ & 1.237 & 0.3 & 0.0840 & 0.0825 & 0.946 & 0.4 & 0.0806 & 0.0795 & 0.950 & 1.09 \\
&  & $\mathcal{M}_3$ & 1.237 & 0.2 & 0.0835 & 0.0818 & 0.944 & 0.3 & 0.0802 & 0.0790 & 0.948 & 1.08 \\
&  & $\mathcal{M}_4$ & 1.237 & 0.2 & 0.0835 & 0.0818 & 0.944 & 0.4 & 0.0804 & 0.0793 & 0.948 & 1.08 \\
\cmidrule(lr){2-13}
& \multirow{4}{*}{36} & $\mathcal{M}_1$ & 1.249 & 0.3 & 0.0889 & 0.0889 & 0.943 & 0.2 & 0.0844 & 0.0842 & 0.947 & 1.11 \\
&  & $\mathcal{M}_2$ & 1.249 & 0.3 & 0.0889 & 0.0889 & 0.943 & 0.4 & 0.0849 & 0.0847 & 0.947 & 1.10 \\
&  & $\mathcal{M}_3$ & 1.249 & 0.4 & 0.0887 & 0.0882 & 0.947 & 0.3 & 0.0845 & 0.0839 & 0.948 & 1.10 \\
&  & $\mathcal{M}_4$ & 1.249 & 0.4 & 0.0887 & 0.0882 & 0.947 & 0.5 & 0.0849 & 0.0844 & 0.949 & 1.09 \\
\midrule
\multirow{12}{*}{4.00}& \multirow{4}{*}{12} & $\mathcal{M}_1$ & 1.179 & 0.3 & 0.0676 & 0.0675 & 0.950 & 0.2 & 0.0663 & 0.0661 & 0.947 & 1.04 \\
&  & $\mathcal{M}_2$ & 1.179 & 0.3 & 0.0676 & 0.0675 & 0.950 & 0.3 & 0.0666 & 0.0664 & 0.946 & 1.03 \\
&  & $\mathcal{M}_3$ & 1.179 & 0.2 & 0.0672 & 0.0670 & 0.946 & 0.2 & 0.0661 & 0.0659 & 0.947 & 1.03 \\
&  & $\mathcal{M}_4$ & 1.179 & 0.2 & 0.0672 & 0.0670 & 0.946 & 0.2 & 0.0664 & 0.0661 & 0.949 & 1.02 \\
\cmidrule(lr){2-13}
& \multirow{4}{*}{24} & $\mathcal{M}_1$ & 1.264 & 0.4 & 0.0815 & 0.0817 & 0.953 & 0.4 & 0.0801 & 0.0787 & 0.945 & 1.04 \\
&  & $\mathcal{M}_2$ & 1.264 & 0.4 & 0.0815 & 0.0817 & 0.953 & 0.5 & 0.0807 & 0.0794 & 0.945 & 1.02 \\
&  & $\mathcal{M}_3$ & 1.264 & 0.2 & 0.0809 & 0.0809 & 0.943 & 0.3 & 0.0798 & 0.0783 & 0.945 & 1.03 \\
&  & $\mathcal{M}_4$ & 1.264 & 0.2 & 0.0809 & 0.0809 & 0.943 & 0.4 & 0.0804 & 0.0790 & 0.945 & 1.01 \\
\cmidrule(lr){2-13}
& \multirow{4}{*}{36} & $\mathcal{M}_1$ & 1.297 & 0.5 & 0.0891 & 0.0886 & 0.955 & 0.5 & 0.0859 & 0.0841 & 0.946 & 1.08 \\
&  & $\mathcal{M}_2$ & 1.297 & 0.5 & 0.0891 & 0.0886 & 0.955 & 0.6 & 0.0869 & 0.0854 & 0.943 & 1.05 \\
&  & $\mathcal{M}_3$ & 1.297 & 0.4 & 0.0885 & 0.0878 & 0.948 & 0.5 & 0.0858 & 0.0838 & 0.943 & 1.06 \\
&  & $\mathcal{M}_4$ & 1.297 & 0.4 & 0.0885 & 0.0878 & 0.948 & 0.6 & 0.0867 & 0.0850 & 0.944 & 1.04 \\
\bottomrule
\end{tabular}}
\end{table}

\begin{table}[htbp]
\centering
\caption{Simulation results for $\WO(\tau)$ under 40\% censoring by $\tau=36$ for the two-component prioritized endpoint. RB: relative bias (\%) for $\widehat{\WO}(\tau)$ on the natural scale; MCSD: Monte Carlo standard deviation of $\log\widehat{\WO}(\tau)$; ASE: average estimated standard error for $\log\widehat{\WO}(\tau)$; Cov: empirical coverage of the 95\% Wald interval based on log-scale inference; RE $=\mathrm{MCSD}^2(C)/\mathrm{MCSD}^2(I)$ computed on the log scale.}
\label{tab:sim-WO40}
\resizebox{\textwidth}{!}{%
\begin{tabular}{ccccrrrrrrrrr}
\toprule
 & & & & \multicolumn{4}{c}{$\widehat{\WO}^{(\text{ipcw})}(\tau)$} & \multicolumn{4}{c}{$\widehat{\WO}^{(\text{ctw})}(\tau)$} & \\
\cmidrule(lr){5-8}\cmidrule(lr){9-12}
$\theta$ & $\tau$ & $\mathcal{M}_k$ & True & RB\% & MCSD & ASE & Cov & RB\% & MCSD & ASE & Cov & RE \\
\midrule
\multirow{12}{*}{1.25}& \multirow{4}{*}{12} & $\mathcal{M}_1$ & 1.170 & 0.4 & 0.0731 & 0.0729 & 0.957 & 0.4 & 0.0695 & 0.0691 & 0.955 & 1.11 \\
&  & $\mathcal{M}_2$ & 1.170 & 0.4 & 0.0731 & 0.0729 & 0.957 & 0.4 & 0.0698 & 0.0693 & 0.952 & 1.09 \\
&  & $\mathcal{M}_3$ & 1.170 & 0.1 & 0.0721 & 0.0714 & 0.956 & 0.3 & 0.0690 & 0.0684 & 0.957 & 1.09 \\
&  & $\mathcal{M}_4$ & 1.170 & 0.1 & 0.0721 & 0.0714 & 0.956 & 0.3 & 0.0694 & 0.0687 & 0.955 & 1.08 \\
\cmidrule(lr){2-13}
& \multirow{4}{*}{24} & $\mathcal{M}_1$ & 1.237 & 0.7 & 0.0924 & 0.0929 & 0.945 & 0.5 & 0.0826 & 0.0837 & 0.954 & 1.25 \\
&  & $\mathcal{M}_2$ & 1.237 & 0.7 & 0.0924 & 0.0929 & 0.945 & 0.7 & 0.0833 & 0.0844 & 0.953 & 1.23 \\
&  & $\mathcal{M}_3$ & 1.237 & 0.5 & 0.0898 & 0.0896 & 0.952 & 0.5 & 0.0820 & 0.0826 & 0.955 & 1.20 \\
&  & $\mathcal{M}_4$ & 1.237 & 0.5 & 0.0898 & 0.0896 & 0.952 & 0.7 & 0.0826 & 0.0833 & 0.954 & 1.18 \\
\cmidrule(lr){2-13}
& \multirow{4}{*}{36} & $\mathcal{M}_1$ & 1.249 & 0.9 & 0.1047 & 0.1054 & 0.959 & 0.7 & 0.0903 & 0.0914 & 0.958 & 1.34 \\
&  & $\mathcal{M}_2$ & 1.249 & 0.9 & 0.1047 & 0.1054 & 0.959 & 1.2 & 0.0919 & 0.0927 & 0.957 & 1.30 \\
&  & $\mathcal{M}_3$ & 1.249 & 1.1 & 0.1006 & 0.1003 & 0.959 & 1.0 & 0.0889 & 0.0895 & 0.957 & 1.28 \\
&  & $\mathcal{M}_4$ & 1.249 & 1.1 & 0.1006 & 0.1003 & 0.959 & 1.4 & 0.0903 & 0.0907 & 0.954 & 1.24 \\
\midrule
\multirow{12}{*}{4.00}& \multirow{4}{*}{12} & $\mathcal{M}_1$ & 1.179 & 0.4 & 0.0719 & 0.0718 & 0.954 & 0.3 & 0.0674 & 0.0683 & 0.946 & 1.14 \\
&  & $\mathcal{M}_2$ & 1.179 & 0.4 & 0.0719 & 0.0718 & 0.954 & 0.4 & 0.0684 & 0.0688 & 0.941 & 1.11 \\
&  & $\mathcal{M}_3$ & 1.179 & 0.2 & 0.0703 & 0.0702 & 0.952 & 0.3 & 0.0667 & 0.0676 & 0.947 & 1.11 \\
&  & $\mathcal{M}_4$ & 1.179 & 0.2 & 0.0703 & 0.0702 & 0.952 & 0.3 & 0.0675 & 0.0681 & 0.946 & 1.08 \\
\cmidrule(lr){2-13}
& \multirow{4}{*}{24} & $\mathcal{M}_1$ & 1.264 & 0.5 & 0.0903 & 0.0920 & 0.954 & 0.4 & 0.0836 & 0.0833 & 0.949 & 1.17 \\
&  & $\mathcal{M}_2$ & 1.264 & 0.5 & 0.0903 & 0.0920 & 0.954 & 0.5 & 0.0859 & 0.0851 & 0.945 & 1.11 \\
&  & $\mathcal{M}_3$ & 1.264 & 0.1 & 0.0871 & 0.0886 & 0.950 & 0.3 & 0.0821 & 0.0821 & 0.946 & 1.13 \\
&  & $\mathcal{M}_4$ & 1.264 & 0.1 & 0.0871 & 0.0886 & 0.950 & 0.4 & 0.0840 & 0.0838 & 0.941 & 1.07 \\
\cmidrule(lr){2-13}
& \multirow{4}{*}{36} & $\mathcal{M}_1$ & 1.297 & 0.7 & 0.1021 & 0.1052 & 0.966 & 0.5 & 0.0907 & 0.0914 & 0.952 & 1.27 \\
&  & $\mathcal{M}_2$ & 1.297 & 0.7 & 0.1021 & 0.1052 & 0.966 & 0.8 & 0.0953 & 0.0948 & 0.946 & 1.15 \\
&  & $\mathcal{M}_3$ & 1.297 & 0.4 & 0.0987 & 0.0997 & 0.957 & 0.6 & 0.0894 & 0.0895 & 0.951 & 1.22 \\
&  & $\mathcal{M}_4$ & 1.297 & 0.4 & 0.0987 & 0.0997 & 0.957 & 0.8 & 0.0934 & 0.0926 & 0.948 & 1.12 \\
\bottomrule
\end{tabular}}
\end{table}

\begin{table}[htbp]
\centering
\caption{Simulation results for $\WO(\tau)$ under 60\% censoring by $\tau=36$ for the two-component prioritized endpoint. RB: relative bias (\%) for $\widehat{\WO}(\tau)$ on the natural scale; MCSD: Monte Carlo standard deviation of $\log\widehat{\WO}(\tau)$; ASE: average estimated standard error for $\log\widehat{\WO}(\tau)$; Cov: empirical coverage of the 95\% Wald interval based on log-scale inference; RE $=\mathrm{MCSD}^2(C)/\mathrm{MCSD}^2(I)$ computed on the log scale.}
\label{tab:sim-WO60}
\resizebox{\textwidth}{!}{%
\begin{tabular}{ccccrrrrrrrrr}
\toprule
 & & & & \multicolumn{4}{c}{$\widehat{\WO}^{(\text{ipcw})}(\tau)$} & \multicolumn{4}{c}{$\widehat{\WO}^{(\text{ctw})}(\tau)$} & \\
\cmidrule(lr){5-8}\cmidrule(lr){9-12}
$\theta$ & $\tau$ & $\mathcal{M}_k$ & True & RB\% & MCSD & ASE & Cov & RB\% & MCSD & ASE & Cov & RE \\
\midrule
\multirow{12}{*}{1.25}& \multirow{4}{*}{12} & $\mathcal{M}_1$ & 1.170 & 0.2 & 0.0809 & 0.0808 & 0.942 & 0.4 & 0.0733 & 0.0725 & 0.937 & 1.22 \\
&  & $\mathcal{M}_2$ & 1.170 & 0.2 & 0.0809 & 0.0808 & 0.942 & 0.5 & 0.0739 & 0.0730 & 0.940 & 1.20 \\
&  & $\mathcal{M}_3$ & 1.170 & $-$0.2 & 0.0771 & 0.0766 & 0.945 & 0.2 & 0.0719 & 0.0709 & 0.939 & 1.15 \\
&  & $\mathcal{M}_4$ & 1.170 & $-$0.2 & 0.0771 & 0.0766 & 0.945 & 0.3 & 0.0724 & 0.0714 & 0.940 & 1.14 \\
\cmidrule(lr){2-13}
& \multirow{4}{*}{24} & $\mathcal{M}_1$ & 1.237 & 1.2 & 0.1151 & 0.1146 & 0.945 & 0.7 & 0.0957 & 0.0924 & 0.947 & 1.45 \\
&  & $\mathcal{M}_2$ & 1.237 & 1.2 & 0.1151 & 0.1146 & 0.945 & 1.0 & 0.0970 & 0.0939 & 0.951 & 1.41 \\
&  & $\mathcal{M}_3$ & 1.237 & 0.6 & 0.1049 & 0.1021 & 0.944 & 0.6 & 0.0918 & 0.0883 & 0.940 & 1.31 \\
&  & $\mathcal{M}_4$ & 1.237 & 0.6 & 0.1049 & 0.1021 & 0.944 & 0.9 & 0.0929 & 0.0897 & 0.940 & 1.27 \\
\cmidrule(lr){2-13}
& \multirow{4}{*}{36} & $\mathcal{M}_1$ & 1.249 & 1.6 & 0.1458 & 0.1457 & 0.971 & 0.8 & 0.1101 & 0.1077 & 0.950 & 1.75 \\
&  & $\mathcal{M}_2$ & 1.249 & 1.6 & 0.1458 & 0.1457 & 0.971 & 1.6 & 0.1127 & 0.1106 & 0.953 & 1.67 \\
&  & $\mathcal{M}_3$ & 1.249 & 1.1 & 0.1181 & 0.1201 & 0.960 & 1.1 & 0.1004 & 0.0990 & 0.951 & 1.39 \\
&  & $\mathcal{M}_4$ & 1.249 & 1.1 & 0.1181 & 0.1201 & 0.960 & 1.8 & 0.1025 & 0.1013 & 0.947 & 1.33 \\
\midrule
\multirow{12}{*}{4.00}& \multirow{4}{*}{12} & $\mathcal{M}_1$ & 1.179 & 0.4 & 0.0780 & 0.0793 & 0.958 & 0.3 & 0.0694 & 0.0717 & 0.955 & 1.26 \\
&  & $\mathcal{M}_2$ & 1.179 & 0.4 & 0.0780 & 0.0793 & 0.958 & 0.4 & 0.0707 & 0.0728 & 0.958 & 1.22 \\
&  & $\mathcal{M}_3$ & 1.179 & $-$0.2 & 0.0743 & 0.0752 & 0.951 & 0.1 & 0.0682 & 0.0701 & 0.956 & 1.19 \\
&  & $\mathcal{M}_4$ & 1.179 & $-$0.2 & 0.0743 & 0.0752 & 0.951 & 0.2 & 0.0694 & 0.0711 & 0.952 & 1.15 \\
\cmidrule(lr){2-13}
& \multirow{4}{*}{24} & $\mathcal{M}_1$ & 1.264 & 0.6 & 0.1100 & 0.1135 & 0.956 & 0.6 & 0.0882 & 0.0922 & 0.962 & 1.56 \\
&  & $\mathcal{M}_2$ & 1.264 & 0.6 & 0.1100 & 0.1135 & 0.956 & 0.9 & 0.0923 & 0.0958 & 0.953 & 1.42 \\
&  & $\mathcal{M}_3$ & 1.264 & $-$0.3 & 0.0981 & 0.1008 & 0.959 & 0.3 & 0.0840 & 0.0880 & 0.960 & 1.36 \\
&  & $\mathcal{M}_4$ & 1.264 & $-$0.3 & 0.0981 & 0.1008 & 0.959 & 0.6 & 0.0873 & 0.0911 & 0.961 & 1.26 \\
\cmidrule(lr){2-13}
& \multirow{4}{*}{36} & $\mathcal{M}_1$ & 1.297 & 0.8 & 0.1401 & 0.1456 & 0.961 & 0.7 & 0.1017 & 0.1080 & 0.974 & 1.90 \\
&  & $\mathcal{M}_2$ & 1.297 & 0.8 & 0.1401 & 0.1456 & 0.961 & 1.2 & 0.1091 & 0.1147 & 0.964 & 1.65 \\
&  & $\mathcal{M}_3$ & 1.297 & 0.1 & 0.1136 & 0.1194 & 0.965 & 0.6 & 0.0931 & 0.0990 & 0.970 & 1.49 \\
&  & $\mathcal{M}_4$ & 1.297 & 0.1 & 0.1136 & 0.1194 & 0.965 & 1.2 & 0.0990 & 0.1047 & 0.968 & 1.32 \\
\bottomrule
\end{tabular}}
\end{table}

\begin{table}[htbp]
\centering
\caption{Simulation results for $\WO(\tau)$ under 80\% censoring by $\tau=36$ for the two-component prioritized endpoint. RB: relative bias (\%) for $\widehat{\WO}(\tau)$ on the natural scale; MCSD: Monte Carlo standard deviation of $\log\widehat{\WO}(\tau)$; ASE: average estimated standard error for $\log\widehat{\WO}(\tau)$; Cov: empirical coverage of the 95\% Wald interval based on log-scale inference; RE $=\mathrm{MCSD}^2(C)/\mathrm{MCSD}^2(I)$ computed on the log scale.}
\label{tab:sim-WO80}
\resizebox{\textwidth}{!}{%
\begin{tabular}{ccccrrrrrrrrr}
\toprule
 & & & & \multicolumn{4}{c}{$\widehat{\WO}^{(\text{ipcw})}(\tau)$} & \multicolumn{4}{c}{$\widehat{\WO}^{(\text{ctw})}(\tau)$} & \\
\cmidrule(lr){5-8}\cmidrule(lr){9-12}
$\theta$ & $\tau$ & $\mathcal{M}_k$ & True & RB\% & MCSD & ASE & Cov & RB\% & MCSD & ASE & Cov & RE \\
\midrule
\multirow{12}{*}{1.25}& \multirow{4}{*}{12} & $\mathcal{M}_1$ & 1.170 & 0.5 & 0.1056 & 0.1018 & 0.946 & 0.5 & 0.0830 & 0.0811 & 0.951 & 1.62 \\
&  & $\mathcal{M}_2$ & 1.170 & 0.5 & 0.1056 & 0.1018 & 0.946 & 0.7 & 0.0841 & 0.0821 & 0.951 & 1.58 \\
&  & $\mathcal{M}_3$ & 1.170 & $-$0.5 & 0.0900 & 0.0872 & 0.944 & 0.0 & 0.0786 & 0.0760 & 0.935 & 1.31 \\
&  & $\mathcal{M}_4$ & 1.170 & $-$0.5 & 0.0900 & 0.0872 & 0.944 & 0.2 & 0.0795 & 0.0769 & 0.938 & 1.28 \\
\cmidrule(lr){2-13}
& \multirow{4}{*}{24} & $\mathcal{M}_1$ & 1.237 & 2.1 & 0.2052 & 0.1930 & 0.969 & 1.6 & 0.1273 & 0.1219 & 0.955 & 2.60 \\
&  & $\mathcal{M}_2$ & 1.237 & 2.1 & 0.2052 & 0.1930 & 0.969 & 2.3 & 0.1305 & 0.1250 & 0.957 & 2.47 \\
&  & $\mathcal{M}_3$ & 1.237 & 0.3 & 0.1298 & 0.1292 & 0.949 & 0.9 & 0.1041 & 0.1013 & 0.942 & 1.56 \\
&  & $\mathcal{M}_4$ & 1.237 & 0.3 & 0.1298 & 0.1292 & 0.949 & 1.6 & 0.1064 & 0.1037 & 0.949 & 1.49 \\
\cmidrule(lr){2-13}
& \multirow{4}{*}{36} & $\mathcal{M}_1$ & 1.249 & 5.4 & 0.3193 & 0.3143 & 0.969 & 3.7 & 0.1997 & 0.1757 & 0.962 & 2.56 \\
&  & $\mathcal{M}_2$ & 1.249 & 5.4 & 0.3193 & 0.3143 & 0.969 & 5.5 & 0.2068 & 0.1833 & 0.964 & 2.38 \\
&  & $\mathcal{M}_3$ & 1.249 & 1.6 & 0.1634 & 0.1656 & 0.953 & 1.5 & 0.1214 & 0.1205 & 0.950 & 1.81 \\
&  & $\mathcal{M}_4$ & 1.249 & 1.6 & 0.1634 & 0.1656 & 0.953 & 2.8 & 0.1253 & 0.1249 & 0.952 & 1.70 \\
\midrule
\multirow{12}{*}{4.00}& \multirow{4}{*}{12} & $\mathcal{M}_1$ & 1.179 & 0.4 & 0.0962 & 0.1001 & 0.969 & 0.1 & 0.0777 & 0.0808 & 0.963 & 1.53 \\
&  & $\mathcal{M}_2$ & 1.179 & 0.4 & 0.0962 & 0.1001 & 0.969 & 0.2 & 0.0800 & 0.0829 & 0.962 & 1.45 \\
&  & $\mathcal{M}_3$ & 1.179 & $-$0.6 & 0.0837 & 0.0858 & 0.954 & $-$0.3 & 0.0732 & 0.0755 & 0.955 & 1.31 \\
&  & $\mathcal{M}_4$ & 1.179 & $-$0.6 & 0.0837 & 0.0858 & 0.954 & $-$0.2 & 0.0751 & 0.0772 & 0.956 & 1.24 \\
\cmidrule(lr){2-13}
& \multirow{4}{*}{24} & $\mathcal{M}_1$ & 1.264 & 1.6 & 0.2052 & 0.1926 & 0.968 & $-$0.1 & 0.1213 & 0.1227 & 0.960 & 2.86 \\
&  & $\mathcal{M}_2$ & 1.264 & 1.6 & 0.2052 & 0.1926 & 0.968 & 0.4 & 0.1283 & 0.1291 & 0.965 & 2.56 \\
&  & $\mathcal{M}_3$ & 1.264 & $-$0.7 & 0.1233 & 0.1277 & 0.951 & $-$0.4 & 0.0993 & 0.1014 & 0.946 & 1.54 \\
&  & $\mathcal{M}_4$ & 1.264 & $-$0.7 & 0.1233 & 0.1277 & 0.951 & 0.2 & 0.1049 & 0.1067 & 0.953 & 1.38 \\
\cmidrule(lr){2-13}
& \multirow{4}{*}{36} & $\mathcal{M}_1$ & 1.297 & 7.0 & 0.3419 & 0.3083 & 0.960 & 0.1 & 0.1802 & 0.1703 & 0.973 & 3.60 \\
&  & $\mathcal{M}_2$ & 1.297 & 7.0 & 0.3419 & 0.3083 & 0.960 & 1.2 & 0.1912 & 0.1808 & 0.975 & 3.20 \\
&  & $\mathcal{M}_3$ & 1.297 & 0.3 & 0.1612 & 0.1647 & 0.963 & $-$0.2 & 0.1174 & 0.1213 & 0.958 & 1.88 \\
&  & $\mathcal{M}_4$ & 1.297 & 0.3 & 0.1612 & 0.1647 & 0.963 & 1.0 & 0.1267 & 0.1305 & 0.961 & 1.62 \\
\bottomrule
\end{tabular}}
\end{table}

\subsection{Nuisance model misspecification study: three-component prioritized endpoint}
\label{supp:sim3}

We also evaluated the original nuisance-model misspecification study under a three-component prioritized endpoint. In this setting, a moderate event was added as the third-priority component after death and the serious nonfatal event. The treatment-arm sample size, covariate distribution, censoring calibration, restriction times, dependence strengths, and nuisance configurations are otherwise the same as in the two-component study described in the main text. The purpose of this analysis is to assess whether the efficiency advantage of conditional tie weighting persists when the hierarchy contains an additional lower-priority component.

The three-component results reinforce the conclusions from the two-component setting. The proposed estimator remains more efficient than IPCW estimator of \citet{cui2025ipcw} across censoring levels and restriction times, with the largest gains under heavier censoring and longer restriction horizons. The gain can be larger than in the two-component setting because the third-priority component creates another opportunity to recover partially observed information when higher-priority ties are unresolved. As before, the correctly specified configuration $\mathcal{M}_1$ gives the most stable performance. Event-time model misspecification reduces the efficiency gain, and censoring-model misspecification has the greatest impact on finite-sample bias and coverage.

\begin{table}[htbp]
\centering
\caption{Simulation results for $\NB(\tau)$ under 20\% censoring by $\tau=36$ for the three-component prioritized endpoint. Same layout as the main-text simulation tables. RB: relative bias (\%); MCSD: Monte Carlo standard deviation; ASE: average estimated standard error; Cov: empirical coverage of the 95\% Wald interval; RE $=\mathrm{MCSD}^2(C)/\mathrm{MCSD}^2(I)$.}
\label{tab:sim3-NB20}
\resizebox{\textwidth}{!}{%
\begin{tabular}{ccccrrrrrrrrr}
\toprule
 & & & & \multicolumn{4}{c}{$\widehat{\NB}^{(\text{ipcw})}(\tau)$} & \multicolumn{4}{c}{$\widehat{\NB}^{(\text{ctw})}(\tau)$} & \\
\cmidrule(lr){5-8}\cmidrule(lr){9-12}
$\theta$ & $\tau$ & $\mathcal{M}_k$ & True & RB\% & MCSD & ASE & Cov & RB\% & MCSD & ASE & Cov & RE \\
\midrule
\multirow{12}{*}{1.25}& \multirow{4}{*}{12} & $\mathcal{M}_1$ & 0.149 & $-$0.0 & 0.0398 & 0.0394 & 0.951 & 0.0 & 0.0391 & 0.0384 & 0.940 & 1.04 \\
&  & $\mathcal{M}_2$ & 0.149 & $-$0.0 & 0.0398 & 0.0394 & 0.951 & 0.2 & 0.0392 & 0.0385 & 0.941 & 1.03 \\
&  & $\mathcal{M}_3$ & 0.149 & $-$0.2 & 0.0398 & 0.0392 & 0.952 & $-$0.0 & 0.0391 & 0.0384 & 0.942 & 1.04 \\
&  & $\mathcal{M}_4$ & 0.149 & $-$0.2 & 0.0398 & 0.0392 & 0.952 & 0.2 & 0.0391 & 0.0384 & 0.942 & 1.03 \\
\cmidrule(lr){2-13}
& \multirow{4}{*}{24} & $\mathcal{M}_1$ & 0.153 & $-$0.5 & 0.0435 & 0.0430 & 0.944 & $-$0.5 & 0.0416 & 0.0411 & 0.945 & 1.10 \\
&  & $\mathcal{M}_2$ & 0.153 & $-$0.5 & 0.0435 & 0.0430 & 0.944 & $-$0.2 & 0.0417 & 0.0413 & 0.948 & 1.09 \\
&  & $\mathcal{M}_3$ & 0.153 & $-$0.1 & 0.0437 & 0.0428 & 0.943 & $-$0.1 & 0.0416 & 0.0411 & 0.943 & 1.10 \\
&  & $\mathcal{M}_4$ & 0.153 & $-$0.1 & 0.0437 & 0.0428 & 0.943 & 0.1 & 0.0418 & 0.0412 & 0.947 & 1.10 \\
\cmidrule(lr){2-13}
& \multirow{4}{*}{36} & $\mathcal{M}_1$ & 0.139 & $-$0.3 & 0.0444 & 0.0447 & 0.956 & 0.1 & 0.0425 & 0.0421 & 0.951 & 1.09 \\
&  & $\mathcal{M}_2$ & 0.139 & $-$0.3 & 0.0444 & 0.0447 & 0.956 & 0.6 & 0.0429 & 0.0424 & 0.951 & 1.07 \\
&  & $\mathcal{M}_3$ & 0.139 & 1.1 & 0.0450 & 0.0445 & 0.946 & 1.0 & 0.0427 & 0.0421 & 0.951 & 1.11 \\
&  & $\mathcal{M}_4$ & 0.139 & 1.1 & 0.0450 & 0.0445 & 0.946 & 1.5 & 0.0430 & 0.0424 & 0.949 & 1.10 \\
\midrule
\multirow{12}{*}{4.00}& \multirow{4}{*}{12} & $\mathcal{M}_1$ & 0.139 & 0.1 & 0.0357 & 0.0369 & 0.951 & 0.1 & 0.0348 & 0.0362 & 0.954 & 1.05 \\
&  & $\mathcal{M}_2$ & 0.139 & 0.1 & 0.0357 & 0.0369 & 0.951 & 0.8 & 0.0351 & 0.0363 & 0.953 & 1.03 \\
&  & $\mathcal{M}_3$ & 0.139 & $-$0.2 & 0.0357 & 0.0367 & 0.952 & $-$0.1 & 0.0348 & 0.0361 & 0.953 & 1.05 \\
&  & $\mathcal{M}_4$ & 0.139 & $-$0.2 & 0.0357 & 0.0367 & 0.952 & 0.7 & 0.0351 & 0.0362 & 0.950 & 1.04 \\
\cmidrule(lr){2-13}
& \multirow{4}{*}{24} & $\mathcal{M}_1$ & 0.158 & 0.3 & 0.0419 & 0.0417 & 0.950 & 0.6 & 0.0401 & 0.0402 & 0.949 & 1.09 \\
&  & $\mathcal{M}_2$ & 0.158 & 0.3 & 0.0419 & 0.0417 & 0.950 & 1.7 & 0.0406 & 0.0405 & 0.950 & 1.06 \\
&  & $\mathcal{M}_3$ & 0.158 & 0.3 & 0.0418 & 0.0415 & 0.948 & 0.8 & 0.0400 & 0.0401 & 0.947 & 1.09 \\
&  & $\mathcal{M}_4$ & 0.158 & 0.3 & 0.0418 & 0.0415 & 0.948 & 1.9 & 0.0405 & 0.0404 & 0.951 & 1.07 \\
\cmidrule(lr){2-13}
& \multirow{4}{*}{36} & $\mathcal{M}_1$ & 0.156 & 0.5 & 0.0442 & 0.0441 & 0.947 & 0.8 & 0.0423 & 0.0419 & 0.948 & 1.09 \\
&  & $\mathcal{M}_2$ & 0.156 & 0.5 & 0.0442 & 0.0441 & 0.947 & 1.8 & 0.0429 & 0.0425 & 0.941 & 1.06 \\
&  & $\mathcal{M}_3$ & 0.156 & 1.1 & 0.0439 & 0.0438 & 0.948 & 1.4 & 0.0421 & 0.0418 & 0.945 & 1.09 \\
&  & $\mathcal{M}_4$ & 0.156 & 1.1 & 0.0439 & 0.0438 & 0.948 & 2.5 & 0.0426 & 0.0424 & 0.949 & 1.06 \\
\bottomrule
\end{tabular}}
\end{table}

\begin{table}[htbp]
\centering
\caption{Simulation results for $\NB(\tau)$ under 40\% censoring by $\tau=36$ for the three-component prioritized endpoint. Same layout as the main-text simulation tables. RB: relative bias (\%); MCSD: Monte Carlo standard deviation; ASE: average estimated standard error; Cov: empirical coverage of the 95\% Wald interval; RE $=\mathrm{MCSD}^2(C)/\mathrm{MCSD}^2(I)$.}
\label{tab:sim3-NB40}
\resizebox{\textwidth}{!}{%
\begin{tabular}{ccccrrrrrrrrr}
\toprule
 & & & & \multicolumn{4}{c}{$\widehat{\NB}^{(\text{ipcw})}(\tau)$} & \multicolumn{4}{c}{$\widehat{\NB}^{(\text{ctw})}(\tau)$} & \\
\cmidrule(lr){5-8}\cmidrule(lr){9-12}
$\theta$ & $\tau$ & $\mathcal{M}_k$ & True & RB\% & MCSD & ASE & Cov & RB\% & MCSD & ASE & Cov & RE \\
\midrule
\multirow{12}{*}{1.25}& \multirow{4}{*}{12} & $\mathcal{M}_1$ & 0.149 & $-$1.2 & 0.0414 & 0.0418 & 0.948 & $-$1.0 & 0.0393 & 0.0394 & 0.950 & 1.11 \\
&  & $\mathcal{M}_2$ & 0.149 & $-$1.2 & 0.0414 & 0.0418 & 0.948 & $-$0.6 & 0.0395 & 0.0395 & 0.946 & 1.10 \\
&  & $\mathcal{M}_3$ & 0.149 & $-$1.8 & 0.0410 & 0.0413 & 0.944 & $-$1.1 & 0.0391 & 0.0392 & 0.946 & 1.10 \\
&  & $\mathcal{M}_4$ & 0.149 & $-$1.8 & 0.0410 & 0.0413 & 0.944 & $-$0.7 & 0.0392 & 0.0393 & 0.947 & 1.09 \\
\cmidrule(lr){2-13}
& \multirow{4}{*}{24} & $\mathcal{M}_1$ & 0.153 & $-$0.9 & 0.0473 & 0.0482 & 0.954 & $-$0.3 & 0.0439 & 0.0431 & 0.946 & 1.16 \\
&  & $\mathcal{M}_2$ & 0.153 & $-$0.9 & 0.0473 & 0.0482 & 0.954 & 0.3 & 0.0443 & 0.0435 & 0.946 & 1.14 \\
&  & $\mathcal{M}_3$ & 0.153 & $-$0.4 & 0.0468 & 0.0471 & 0.950 & 0.0 & 0.0434 & 0.0428 & 0.946 & 1.16 \\
&  & $\mathcal{M}_4$ & 0.153 & $-$0.4 & 0.0468 & 0.0471 & 0.950 & 0.6 & 0.0437 & 0.0432 & 0.946 & 1.15 \\
\cmidrule(lr){2-13}
& \multirow{4}{*}{36} & $\mathcal{M}_1$ & 0.139 & $-$1.3 & 0.0522 & 0.0528 & 0.943 & $-$0.5 & 0.0465 & 0.0453 & 0.949 & 1.26 \\
&  & $\mathcal{M}_2$ & 0.139 & $-$1.3 & 0.0522 & 0.0528 & 0.943 & 0.6 & 0.0473 & 0.0462 & 0.945 & 1.22 \\
&  & $\mathcal{M}_3$ & 0.139 & 1.0 & 0.0508 & 0.0508 & 0.941 & 0.8 & 0.0454 & 0.0447 & 0.946 & 1.25 \\
&  & $\mathcal{M}_4$ & 0.139 & 1.0 & 0.0508 & 0.0508 & 0.941 & 1.9 & 0.0460 & 0.0455 & 0.945 & 1.22 \\
\midrule
\multirow{12}{*}{4.00}& \multirow{4}{*}{12} & $\mathcal{M}_1$ & 0.139 & 0.1 & 0.0394 & 0.0393 & 0.948 & 0.1 & 0.0369 & 0.0373 & 0.950 & 1.14 \\
&  & $\mathcal{M}_2$ & 0.139 & 0.1 & 0.0394 & 0.0393 & 0.948 & 2.0 & 0.0375 & 0.0378 & 0.944 & 1.10 \\
&  & $\mathcal{M}_3$ & 0.139 & $-$0.8 & 0.0386 & 0.0386 & 0.947 & $-$0.2 & 0.0365 & 0.0371 & 0.953 & 1.12 \\
&  & $\mathcal{M}_4$ & 0.139 & $-$0.8 & 0.0386 & 0.0386 & 0.947 & 1.6 & 0.0371 & 0.0375 & 0.949 & 1.08 \\
\cmidrule(lr){2-13}
& \multirow{4}{*}{24} & $\mathcal{M}_1$ & 0.158 & 0.1 & 0.0469 & 0.0471 & 0.950 & 0.1 & 0.0416 & 0.0426 & 0.949 & 1.27 \\
&  & $\mathcal{M}_2$ & 0.158 & 0.1 & 0.0469 & 0.0471 & 0.950 & 2.5 & 0.0430 & 0.0435 & 0.945 & 1.19 \\
&  & $\mathcal{M}_3$ & 0.158 & $-$0.5 & 0.0456 & 0.0456 & 0.950 & 0.2 & 0.0410 & 0.0421 & 0.947 & 1.23 \\
&  & $\mathcal{M}_4$ & 0.158 & $-$0.5 & 0.0456 & 0.0456 & 0.950 & 2.6 & 0.0423 & 0.0430 & 0.942 & 1.16 \\
\cmidrule(lr){2-13}
& \multirow{4}{*}{36} & $\mathcal{M}_1$ & 0.156 & $-$0.0 & 0.0518 & 0.0524 & 0.958 & 0.1 & 0.0453 & 0.0455 & 0.954 & 1.31 \\
&  & $\mathcal{M}_2$ & 0.156 & $-$0.0 & 0.0518 & 0.0524 & 0.958 & 2.4 & 0.0477 & 0.0472 & 0.943 & 1.18 \\
&  & $\mathcal{M}_3$ & 0.156 & 0.6 & 0.0500 & 0.0500 & 0.952 & 1.2 & 0.0447 & 0.0447 & 0.950 & 1.25 \\
&  & $\mathcal{M}_4$ & 0.156 & 0.6 & 0.0500 & 0.0500 & 0.952 & 3.8 & 0.0467 & 0.0463 & 0.949 & 1.14 \\
\bottomrule
\end{tabular}}
\end{table}

\begin{table}[htbp]
\centering
\caption{Simulation results for $\NB(\tau)$ under 60\% censoring by $\tau=36$ for the three-component prioritized endpoint. Same layout as the main-text simulation tables. RB: relative bias (\%); MCSD: Monte Carlo standard deviation; ASE: average estimated standard error; Cov: empirical coverage of the 95\% Wald interval; RE $=\mathrm{MCSD}^2(C)/\mathrm{MCSD}^2(I)$.}
\label{tab:sim3-NB60}
\resizebox{\textwidth}{!}{%
\begin{tabular}{ccccrrrrrrrrr}
\toprule
 & & & & \multicolumn{4}{c}{$\widehat{\NB}^{(\text{ipcw})}(\tau)$} & \multicolumn{4}{c}{$\widehat{\NB}^{(\text{ctw})}(\tau)$} & \\
\cmidrule(lr){5-8}\cmidrule(lr){9-12}
$\theta$ & $\tau$ & $\mathcal{M}_k$ & True & RB\% & MCSD & ASE & Cov & RB\% & MCSD & ASE & Cov & RE \\
\midrule
\multirow{12}{*}{1.25}& \multirow{4}{*}{12} & $\mathcal{M}_1$ & 0.149 & $-$0.5 & 0.0458 & 0.0464 & 0.950 & $-$0.1 & 0.0411 & 0.0412 & 0.947 & 1.24 \\
&  & $\mathcal{M}_2$ & 0.149 & $-$0.5 & 0.0458 & 0.0464 & 0.950 & 0.6 & 0.0414 & 0.0415 & 0.947 & 1.22 \\
&  & $\mathcal{M}_3$ & 0.149 & $-$1.4 & 0.0453 & 0.0447 & 0.943 & $-$0.2 & 0.0409 & 0.0406 & 0.948 & 1.22 \\
&  & $\mathcal{M}_4$ & 0.149 & $-$1.4 & 0.0453 & 0.0447 & 0.943 & 0.5 & 0.0412 & 0.0409 & 0.947 & 1.21 \\
\cmidrule(lr){2-13}
& \multirow{4}{*}{24} & $\mathcal{M}_1$ & 0.153 & $-$0.5 & 0.0548 & 0.0593 & 0.962 & $-$0.4 & 0.0457 & 0.0473 & 0.960 & 1.44 \\
&  & $\mathcal{M}_2$ & 0.153 & $-$0.5 & 0.0548 & 0.0593 & 0.962 & 0.7 & 0.0466 & 0.0482 & 0.958 & 1.38 \\
&  & $\mathcal{M}_3$ & 0.153 & 0.6 & 0.0530 & 0.0543 & 0.946 & 0.8 & 0.0453 & 0.0458 & 0.949 & 1.37 \\
&  & $\mathcal{M}_4$ & 0.153 & 0.6 & 0.0530 & 0.0543 & 0.946 & 1.9 & 0.0460 & 0.0466 & 0.947 & 1.33 \\
\cmidrule(lr){2-13}
& \multirow{4}{*}{36} & $\mathcal{M}_1$ & 0.139 & $-$1.2 & 0.0678 & 0.0723 & 0.970 & 1.0 & 0.0519 & 0.0530 & 0.955 & 1.71 \\
&  & $\mathcal{M}_2$ & 0.139 & $-$1.2 & 0.0678 & 0.0723 & 0.970 & 3.3 & 0.0538 & 0.0547 & 0.953 & 1.58 \\
&  & $\mathcal{M}_3$ & 0.139 & 4.4 & 0.0591 & 0.0614 & 0.953 & 3.7 & 0.0483 & 0.0494 & 0.951 & 1.50 \\
&  & $\mathcal{M}_4$ & 0.139 & 4.4 & 0.0591 & 0.0614 & 0.953 & 5.9 & 0.0500 & 0.0508 & 0.949 & 1.40 \\
\midrule
\multirow{12}{*}{4.00}& \multirow{4}{*}{12} & $\mathcal{M}_1$ & 0.139 & $-$1.0 & 0.0433 & 0.0436 & 0.955 & $-$0.1 & 0.0396 & 0.0394 & 0.954 & 1.20 \\
&  & $\mathcal{M}_2$ & 0.139 & $-$1.0 & 0.0433 & 0.0436 & 0.955 & 3.5 & 0.0409 & 0.0402 & 0.943 & 1.12 \\
&  & $\mathcal{M}_3$ & 0.139 & $-$2.6 & 0.0413 & 0.0416 & 0.955 & $-$0.9 & 0.0389 & 0.0386 & 0.948 & 1.13 \\
&  & $\mathcal{M}_4$ & 0.139 & $-$2.6 & 0.0413 & 0.0416 & 0.955 & 2.6 & 0.0401 & 0.0394 & 0.939 & 1.06 \\
\cmidrule(lr){2-13}
& \multirow{4}{*}{24} & $\mathcal{M}_1$ & 0.158 & $-$0.6 & 0.0551 & 0.0581 & 0.970 & 0.1 & 0.0455 & 0.0473 & 0.971 & 1.47 \\
&  & $\mathcal{M}_2$ & 0.158 & $-$0.6 & 0.0551 & 0.0581 & 0.970 & 4.2 & 0.0481 & 0.0491 & 0.964 & 1.31 \\
&  & $\mathcal{M}_3$ & 0.158 & $-$1.3 & 0.0501 & 0.0523 & 0.962 & 0.5 & 0.0437 & 0.0454 & 0.962 & 1.31 \\
&  & $\mathcal{M}_4$ & 0.158 & $-$1.3 & 0.0501 & 0.0523 & 0.962 & 4.7 & 0.0459 & 0.0471 & 0.961 & 1.19 \\
\cmidrule(lr){2-13}
& \multirow{4}{*}{36} & $\mathcal{M}_1$ & 0.156 & $-$1.5 & 0.0716 & 0.0726 & 0.965 & $-$0.3 & 0.0524 & 0.0537 & 0.962 & 1.86 \\
&  & $\mathcal{M}_2$ & 0.156 & $-$1.5 & 0.0716 & 0.0726 & 0.965 & 3.5 & 0.0569 & 0.0569 & 0.956 & 1.58 \\
&  & $\mathcal{M}_3$ & 0.156 & 0.3 & 0.0589 & 0.0603 & 0.954 & 2.0 & 0.0485 & 0.0497 & 0.960 & 1.48 \\
&  & $\mathcal{M}_4$ & 0.156 & 0.3 & 0.0589 & 0.0603 & 0.954 & 6.3 & 0.0520 & 0.0525 & 0.953 & 1.28 \\
\bottomrule
\end{tabular}}
\end{table}

\begin{table}[htbp]
\centering
\caption{Simulation results for $\NB(\tau)$ under 80\% censoring by $\tau=36$ for the three-component prioritized endpoint. Same layout as the main-text simulation tables. RB: relative bias (\%); MCSD: Monte Carlo standard deviation; ASE: average estimated standard error; Cov: empirical coverage of the 95\% Wald interval; RE $=\mathrm{MCSD}^2(C)/\mathrm{MCSD}^2(I)$.}
\label{tab:sim3-NB80}
\resizebox{\textwidth}{!}{%
\begin{tabular}{ccccrrrrrrrrr}
\toprule
 & & & & \multicolumn{4}{c}{$\widehat{\NB}^{(\text{ipcw})}(\tau)$} & \multicolumn{4}{c}{$\widehat{\NB}^{(\text{ctw})}(\tau)$} & \\
\cmidrule(lr){5-8}\cmidrule(lr){9-12}
$\theta$ & $\tau$ & $\mathcal{M}_k$ & True & RB\% & MCSD & ASE & Cov & RB\% & MCSD & ASE & Cov & RE \\
\midrule
\multirow{12}{*}{1.25}& \multirow{4}{*}{12} & $\mathcal{M}_1$ & 0.149 & $-$0.0 & 0.0569 & 0.0585 & 0.961 & 1.2 & 0.0457 & 0.0456 & 0.957 & 1.55 \\
&  & $\mathcal{M}_2$ & 0.149 & $-$0.0 & 0.0569 & 0.0585 & 0.961 & 2.5 & 0.0465 & 0.0462 & 0.957 & 1.50 \\
&  & $\mathcal{M}_3$ & 0.149 & $-$1.4 & 0.0513 & 0.0516 & 0.953 & 0.4 & 0.0440 & 0.0436 & 0.955 & 1.36 \\
&  & $\mathcal{M}_4$ & 0.149 & $-$1.4 & 0.0513 & 0.0516 & 0.953 & 1.7 & 0.0445 & 0.0440 & 0.954 & 1.33 \\
\cmidrule(lr){2-13}
& \multirow{4}{*}{24} & $\mathcal{M}_1$ & 0.153 & $-$1.1 & 0.0976 & 0.0974 & 0.958 & $-$0.1 & 0.0616 & 0.0608 & 0.959 & 2.51 \\
&  & $\mathcal{M}_2$ & 0.153 & $-$1.1 & 0.0976 & 0.0974 & 0.958 & 2.0 & 0.0636 & 0.0626 & 0.958 & 2.35 \\
&  & $\mathcal{M}_3$ & 0.153 & 2.3 & 0.0676 & 0.0697 & 0.955 & 2.2 & 0.0521 & 0.0525 & 0.959 & 1.68 \\
&  & $\mathcal{M}_4$ & 0.153 & 2.3 & 0.0676 & 0.0697 & 0.955 & 4.2 & 0.0535 & 0.0539 & 0.957 & 1.59 \\
\cmidrule(lr){2-13}
& \multirow{4}{*}{36} & $\mathcal{M}_1$ & 0.139 & $-$7.0 & 0.1584 & 0.1404 & 0.963 & $-$2.9 & 0.0938 & 0.0808 & 0.962 & 2.85 \\
&  & $\mathcal{M}_2$ & 0.139 & $-$7.0 & 0.1584 & 0.1404 & 0.963 & 1.4 & 0.0967 & 0.0840 & 0.964 & 2.68 \\
&  & $\mathcal{M}_3$ & 0.139 & 6.4 & 0.0813 & 0.0858 & 0.962 & 5.2 & 0.0582 & 0.0601 & 0.957 & 1.95 \\
&  & $\mathcal{M}_4$ & 0.139 & 6.4 & 0.0813 & 0.0858 & 0.962 & 9.1 & 0.0609 & 0.0628 & 0.956 & 1.78 \\
\midrule
\multirow{12}{*}{4.00}& \multirow{4}{*}{12} & $\mathcal{M}_1$ & 0.139 & 0.5 & 0.0564 & 0.0552 & 0.959 & 0.7 & 0.0447 & 0.0444 & 0.962 & 1.59 \\
&  & $\mathcal{M}_2$ & 0.139 & 0.5 & 0.0564 & 0.0552 & 0.959 & 7.0 & 0.0470 & 0.0459 & 0.953 & 1.43 \\
&  & $\mathcal{M}_3$ & 0.139 & $-$2.7 & 0.0489 & 0.0479 & 0.949 & $-$0.7 & 0.0419 & 0.0418 & 0.956 & 1.36 \\
&  & $\mathcal{M}_4$ & 0.139 & $-$2.7 & 0.0489 & 0.0479 & 0.949 & 5.5 & 0.0438 & 0.0432 & 0.954 & 1.25 \\
\cmidrule(lr){2-13}
& \multirow{4}{*}{24} & $\mathcal{M}_1$ & 0.158 & $-$2.1 & 0.1009 & 0.0973 & 0.954 & 0.6 & 0.0625 & 0.0627 & 0.956 & 2.61 \\
&  & $\mathcal{M}_2$ & 0.158 & $-$2.1 & 0.1009 & 0.0973 & 0.954 & 7.9 & 0.0676 & 0.0662 & 0.943 & 2.23 \\
&  & $\mathcal{M}_3$ & 0.158 & $-$1.6 & 0.0668 & 0.0669 & 0.947 & 1.4 & 0.0523 & 0.0527 & 0.949 & 1.63 \\
&  & $\mathcal{M}_4$ & 0.158 & $-$1.6 & 0.0668 & 0.0669 & 0.947 & 9.1 & 0.0565 & 0.0557 & 0.936 & 1.40 \\
\cmidrule(lr){2-13}
& \multirow{4}{*}{36} & $\mathcal{M}_1$ & 0.156 & $-$6.4 & 0.1586 & 0.1394 & 0.944 & $-$0.5 & 0.0913 & 0.0853 & 0.968 & 3.01 \\
&  & $\mathcal{M}_2$ & 0.156 & $-$6.4 & 0.1586 & 0.1394 & 0.944 & 6.9 & 0.0980 & 0.0908 & 0.959 & 2.62 \\
&  & $\mathcal{M}_3$ & 0.156 & 0.2 & 0.0801 & 0.0837 & 0.955 & 4.0 & 0.0605 & 0.0612 & 0.959 & 1.75 \\
&  & $\mathcal{M}_4$ & 0.156 & 0.2 & 0.0801 & 0.0837 & 0.955 & 12.4 & 0.0668 & 0.0664 & 0.936 & 1.44 \\
\bottomrule
\end{tabular}}
\end{table}

\begin{table}[htbp]
\centering
\caption{Simulation results for $\WR(\tau)$ under 20\% censoring by $\tau=36$ for the three-component prioritized endpoint. Same layout as the main-text simulation tables. RB: relative bias (\%) for $\widehat{\WR}(\tau)$ on the natural scale; MCSD: Monte Carlo standard deviation of $\log\widehat{\WR}(\tau)$; ASE: average estimated standard error for $\log\widehat{\WR}(\tau)$; Cov: empirical coverage of the 95\% Wald interval based on log-scale inference; RE $=\mathrm{MCSD}^2(C)/\mathrm{MCSD}^2(I)$ computed on the log scale.}
\label{tab:sim3-WR20}
\resizebox{\textwidth}{!}{%
\begin{tabular}{ccccrrrrrrrrr}
\toprule
 & & & & \multicolumn{4}{c}{$\widehat{\WR}^{(\text{ipcw})}(\tau)$} & \multicolumn{4}{c}{$\widehat{\WR}^{(\text{ctw})}(\tau)$} & \\
\cmidrule(lr){5-8}\cmidrule(lr){9-12}
$\theta$ & $\tau$ & $\mathcal{M}_k$ & True & RB\% & MCSD & ASE & Cov & RB\% & MCSD & ASE & Cov & RE \\
\midrule
\multirow{12}{*}{1.25}& \multirow{4}{*}{12} & $\mathcal{M}_1$ & 1.490 & 0.6 & 0.1089 & 0.1075 & 0.953 & 0.6 & 0.1068 & 0.1051 & 0.943 & 1.04 \\
&  & $\mathcal{M}_2$ & 1.490 & 0.6 & 0.1089 & 0.1075 & 0.953 & 0.6 & 0.1067 & 0.1050 & 0.945 & 1.04 \\
&  & $\mathcal{M}_3$ & 1.490 & 0.7 & 0.1092 & 0.1075 & 0.953 & 0.7 & 0.1069 & 0.1051 & 0.946 & 1.04 \\
&  & $\mathcal{M}_4$ & 1.490 & 0.7 & 0.1092 & 0.1075 & 0.953 & 0.7 & 0.1068 & 0.1050 & 0.944 & 1.05 \\
\cmidrule(lr){2-13}
& \multirow{4}{*}{24} & $\mathcal{M}_1$ & 1.402 & 0.4 & 0.0976 & 0.0959 & 0.943 & 0.3 & 0.0931 & 0.0920 & 0.946 & 1.10 \\
&  & $\mathcal{M}_2$ & 1.402 & 0.4 & 0.0976 & 0.0959 & 0.943 & 0.3 & 0.0930 & 0.0919 & 0.946 & 1.10 \\
&  & $\mathcal{M}_3$ & 1.402 & 0.6 & 0.0985 & 0.0958 & 0.944 & 0.5 & 0.0934 & 0.0920 & 0.944 & 1.11 \\
&  & $\mathcal{M}_4$ & 1.402 & 0.6 & 0.0985 & 0.0958 & 0.944 & 0.4 & 0.0933 & 0.0919 & 0.947 & 1.11 \\
\cmidrule(lr){2-13}
& \multirow{4}{*}{36} & $\mathcal{M}_1$ & 1.335 & 0.4 & 0.0936 & 0.0938 & 0.957 & 0.5 & 0.0897 & 0.0888 & 0.951 & 1.09 \\
&  & $\mathcal{M}_2$ & 1.335 & 0.4 & 0.0936 & 0.0938 & 0.957 & 0.4 & 0.0897 & 0.0888 & 0.952 & 1.09 \\
&  & $\mathcal{M}_3$ & 1.335 & 0.8 & 0.0950 & 0.0934 & 0.945 & 0.8 & 0.0902 & 0.0887 & 0.950 & 1.11 \\
&  & $\mathcal{M}_4$ & 1.335 & 0.8 & 0.0950 & 0.0934 & 0.945 & 0.7 & 0.0901 & 0.0886 & 0.949 & 1.11 \\
\midrule
\multirow{12}{*}{4.00}& \multirow{4}{*}{12} & $\mathcal{M}_1$ & 1.568 & 0.9 & 0.1187 & 0.1224 & 0.954 & 0.8 & 0.1155 & 0.1200 & 0.950 & 1.06 \\
&  & $\mathcal{M}_2$ & 1.568 & 0.9 & 0.1187 & 0.1224 & 0.954 & 0.7 & 0.1154 & 0.1194 & 0.948 & 1.06 \\
&  & $\mathcal{M}_3$ & 1.568 & 1.0 & 0.1191 & 0.1224 & 0.949 & 0.9 & 0.1158 & 0.1200 & 0.953 & 1.06 \\
&  & $\mathcal{M}_4$ & 1.568 & 1.0 & 0.1191 & 0.1224 & 0.949 & 0.8 & 0.1156 & 0.1194 & 0.951 & 1.06 \\
\cmidrule(lr){2-13}
& \multirow{4}{*}{24} & $\mathcal{M}_1$ & 1.480 & 0.8 & 0.1063 & 0.1053 & 0.947 & 0.8 & 0.1017 & 0.1016 & 0.947 & 1.09 \\
&  & $\mathcal{M}_2$ & 1.480 & 0.8 & 0.1063 & 0.1053 & 0.947 & 0.5 & 0.1012 & 0.1008 & 0.946 & 1.10 \\
&  & $\mathcal{M}_3$ & 1.480 & 1.0 & 0.1066 & 0.1052 & 0.942 & 1.0 & 0.1018 & 0.1016 & 0.946 & 1.10 \\
&  & $\mathcal{M}_4$ & 1.480 & 1.0 & 0.1066 & 0.1052 & 0.942 & 0.7 & 0.1012 & 0.1008 & 0.947 & 1.11 \\
\cmidrule(lr){2-13}
& \multirow{4}{*}{36} & $\mathcal{M}_1$ & 1.417 & 0.8 & 0.1007 & 0.0999 & 0.946 & 0.8 & 0.0962 & 0.0952 & 0.951 & 1.09 \\
&  & $\mathcal{M}_2$ & 1.417 & 0.8 & 0.1007 & 0.0999 & 0.946 & 0.3 & 0.0953 & 0.0943 & 0.948 & 1.12 \\
&  & $\mathcal{M}_3$ & 1.417 & 1.1 & 0.1004 & 0.0995 & 0.951 & 1.1 & 0.0959 & 0.0951 & 0.948 & 1.09 \\
&  & $\mathcal{M}_4$ & 1.417 & 1.1 & 0.1004 & 0.0995 & 0.951 & 0.5 & 0.0948 & 0.0942 & 0.947 & 1.12 \\
\bottomrule
\end{tabular}}
\end{table}

\begin{table}[htbp]
\centering
\caption{Simulation results for $\WR(\tau)$ under 40\% censoring by $\tau=36$ for the three-component prioritized endpoint. Same layout as the main-text simulation tables. RB: relative bias (\%) for $\widehat{\WR}(\tau)$ on the natural scale; MCSD: Monte Carlo standard deviation of $\log\widehat{\WR}(\tau)$; ASE: average estimated standard error for $\log\widehat{\WR}(\tau)$; Cov: empirical coverage of the 95\% Wald interval based on log-scale inference; RE $=\mathrm{MCSD}^2(C)/\mathrm{MCSD}^2(I)$ computed on the log scale.}
\label{tab:sim3-WR40}
\resizebox{\textwidth}{!}{%
\begin{tabular}{ccccrrrrrrrrr}
\toprule
 & & & & \multicolumn{4}{c}{$\widehat{\WR}^{(\text{ipcw})}(\tau)$} & \multicolumn{4}{c}{$\widehat{\WR}^{(\text{ctw})}(\tau)$} & \\
\cmidrule(lr){5-8}\cmidrule(lr){9-12}
$\theta$ & $\tau$ & $\mathcal{M}_k$ & True & RB\% & MCSD & ASE & Cov & RB\% & MCSD & ASE & Cov & RE \\
\midrule
\multirow{12}{*}{1.25}& \multirow{4}{*}{12} & $\mathcal{M}_1$ & 1.490 & 0.3 & 0.1136 & 0.1139 & 0.943 & 0.3 & 0.1080 & 0.1078 & 0.946 & 1.11 \\
&  & $\mathcal{M}_2$ & 1.490 & 0.3 & 0.1136 & 0.1139 & 0.943 & 0.2 & 0.1078 & 0.1076 & 0.945 & 1.11 \\
&  & $\mathcal{M}_3$ & 1.490 & 0.4 & 0.1136 & 0.1136 & 0.939 & 0.4 & 0.1079 & 0.1078 & 0.943 & 1.11 \\
&  & $\mathcal{M}_4$ & 1.490 & 0.4 & 0.1136 & 0.1136 & 0.939 & 0.4 & 0.1076 & 0.1075 & 0.945 & 1.11 \\
\cmidrule(lr){2-13}
& \multirow{4}{*}{24} & $\mathcal{M}_1$ & 1.402 & 0.4 & 0.1065 & 0.1070 & 0.949 & 0.5 & 0.0985 & 0.0964 & 0.946 & 1.17 \\
&  & $\mathcal{M}_2$ & 1.402 & 0.4 & 0.1065 & 0.1070 & 0.949 & 0.3 & 0.0983 & 0.0963 & 0.947 & 1.17 \\
&  & $\mathcal{M}_3$ & 1.402 & 0.7 & 0.1059 & 0.1053 & 0.946 & 0.7 & 0.0977 & 0.0960 & 0.944 & 1.17 \\
&  & $\mathcal{M}_4$ & 1.402 & 0.7 & 0.1059 & 0.1053 & 0.946 & 0.5 & 0.0973 & 0.0959 & 0.947 & 1.18 \\
\cmidrule(lr){2-13}
& \multirow{4}{*}{36} & $\mathcal{M}_1$ & 1.335 & 0.4 & 0.1105 & 0.1098 & 0.938 & 0.5 & 0.0984 & 0.0954 & 0.949 & 1.26 \\
&  & $\mathcal{M}_2$ & 1.335 & 0.4 & 0.1105 & 0.1098 & 0.938 & 0.3 & 0.0983 & 0.0955 & 0.944 & 1.26 \\
&  & $\mathcal{M}_3$ & 1.335 & 0.9 & 0.1071 & 0.1058 & 0.938 & 0.8 & 0.0959 & 0.0941 & 0.946 & 1.25 \\
&  & $\mathcal{M}_4$ & 1.335 & 0.9 & 0.1071 & 0.1058 & 0.938 & 0.6 & 0.0956 & 0.0941 & 0.943 & 1.26 \\
\midrule
\multirow{12}{*}{4.00}& \multirow{4}{*}{12} & $\mathcal{M}_1$ & 1.568 & 1.0 & 0.1312 & 0.1295 & 0.945 & 0.8 & 0.1230 & 0.1233 & 0.946 & 1.14 \\
&  & $\mathcal{M}_2$ & 1.568 & 1.0 & 0.1312 & 0.1295 & 0.945 & 0.7 & 0.1221 & 0.1221 & 0.943 & 1.16 \\
&  & $\mathcal{M}_3$ & 1.568 & 1.2 & 0.1303 & 0.1290 & 0.943 & 1.0 & 0.1227 & 0.1232 & 0.948 & 1.13 \\
&  & $\mathcal{M}_4$ & 1.568 & 1.2 & 0.1303 & 0.1290 & 0.943 & 0.8 & 0.1216 & 0.1220 & 0.947 & 1.15 \\
\cmidrule(lr){2-13}
& \multirow{4}{*}{24} & $\mathcal{M}_1$ & 1.480 & 0.9 & 0.1191 & 0.1176 & 0.947 & 0.6 & 0.1055 & 0.1071 & 0.951 & 1.27 \\
&  & $\mathcal{M}_2$ & 1.480 & 0.9 & 0.1191 & 0.1176 & 0.947 & 0.0 & 0.1049 & 0.1054 & 0.943 & 1.29 \\
&  & $\mathcal{M}_3$ & 1.480 & 1.1 & 0.1173 & 0.1156 & 0.947 & 0.9 & 0.1049 & 0.1066 & 0.951 & 1.25 \\
&  & $\mathcal{M}_4$ & 1.480 & 1.1 & 0.1173 & 0.1156 & 0.947 & 0.3 & 0.1040 & 0.1049 & 0.943 & 1.27 \\
\cmidrule(lr){2-13}
& \multirow{4}{*}{36} & $\mathcal{M}_1$ & 1.417 & 0.9 & 0.1180 & 0.1173 & 0.950 & 0.6 & 0.1029 & 0.1028 & 0.949 & 1.31 \\
&  & $\mathcal{M}_2$ & 1.417 & 0.9 & 0.1180 & 0.1173 & 0.950 & $-$0.4 & 0.1027 & 0.1012 & 0.945 & 1.32 \\
&  & $\mathcal{M}_3$ & 1.417 & 1.2 & 0.1147 & 0.1130 & 0.949 & 1.1 & 0.1021 & 0.1015 & 0.951 & 1.26 \\
&  & $\mathcal{M}_4$ & 1.417 & 1.2 & 0.1147 & 0.1130 & 0.949 & 0.1 & 0.1012 & 0.0997 & 0.948 & 1.28 \\
\bottomrule
\end{tabular}}
\end{table}

\begin{table}[htbp]
\centering
\caption{Simulation results for $\WR(\tau)$ under 60\% censoring by $\tau=36$ for the three-component prioritized endpoint. Same layout as the main-text simulation tables. RB: relative bias (\%) for $\widehat{\WR}(\tau)$ on the natural scale; MCSD: Monte Carlo standard deviation of $\log\widehat{\WR}(\tau)$; ASE: average estimated standard error for $\log\widehat{\WR}(\tau)$; Cov: empirical coverage of the 95\% Wald interval based on log-scale inference; RE $=\mathrm{MCSD}^2(C)/\mathrm{MCSD}^2(I)$ computed on the log scale.}
\label{tab:sim3-WR60}
\resizebox{\textwidth}{!}{%
\begin{tabular}{ccccrrrrrrrrr}
\toprule
 & & & & \multicolumn{4}{c}{$\widehat{\WR}^{(\text{ipcw})}(\tau)$} & \multicolumn{4}{c}{$\widehat{\WR}^{(\text{ctw})}(\tau)$} & \\
\cmidrule(lr){5-8}\cmidrule(lr){9-12}
$\theta$ & $\tau$ & $\mathcal{M}_k$ & True & RB\% & MCSD & ASE & Cov & RB\% & MCSD & ASE & Cov & RE \\
\midrule
\multirow{12}{*}{1.25}& \multirow{4}{*}{12} & $\mathcal{M}_1$ & 1.490 & 0.7 & 0.1260 & 0.1256 & 0.950 & 0.7 & 0.1129 & 0.1124 & 0.949 & 1.25 \\
&  & $\mathcal{M}_2$ & 1.490 & 0.7 & 0.1260 & 0.1256 & 0.950 & 0.6 & 0.1124 & 0.1120 & 0.945 & 1.26 \\
&  & $\mathcal{M}_3$ & 1.490 & 1.1 & 0.1269 & 0.1233 & 0.939 & 1.0 & 0.1132 & 0.1118 & 0.950 & 1.26 \\
&  & $\mathcal{M}_4$ & 1.490 & 1.1 & 0.1269 & 0.1233 & 0.939 & 0.8 & 0.1126 & 0.1114 & 0.950 & 1.27 \\
\cmidrule(lr){2-13}
& \multirow{4}{*}{24} & $\mathcal{M}_1$ & 1.402 & 1.0 & 0.1238 & 0.1307 & 0.958 & 0.5 & 0.1026 & 0.1056 & 0.958 & 1.46 \\
&  & $\mathcal{M}_2$ & 1.402 & 1.0 & 0.1238 & 0.1307 & 0.958 & 0.2 & 0.1025 & 0.1055 & 0.958 & 1.46 \\
&  & $\mathcal{M}_3$ & 1.402 & 1.5 & 0.1205 & 0.1211 & 0.946 & 1.1 & 0.1021 & 0.1028 & 0.948 & 1.39 \\
&  & $\mathcal{M}_4$ & 1.402 & 1.5 & 0.1205 & 0.1211 & 0.946 & 0.8 & 0.1018 & 0.1025 & 0.949 & 1.40 \\
\cmidrule(lr){2-13}
& \multirow{4}{*}{36} & $\mathcal{M}_1$ & 1.335 & 1.1 & 0.1430 & 0.1487 & 0.956 & 1.1 & 0.1097 & 0.1112 & 0.951 & 1.70 \\
&  & $\mathcal{M}_2$ & 1.335 & 1.1 & 0.1430 & 0.1487 & 0.956 & 0.9 & 0.1103 & 0.1113 & 0.950 & 1.68 \\
&  & $\mathcal{M}_3$ & 1.335 & 2.3 & 0.1249 & 0.1271 & 0.951 & 1.8 & 0.1025 & 0.1037 & 0.950 & 1.49 \\
&  & $\mathcal{M}_4$ & 1.335 & 2.3 & 0.1249 & 0.1271 & 0.951 & 1.5 & 0.1028 & 0.1037 & 0.953 & 1.48 \\
\midrule
\multirow{12}{*}{4.00}& \multirow{4}{*}{12} & $\mathcal{M}_1$ & 1.568 & 0.8 & 0.1440 & 0.1429 & 0.955 & 0.7 & 0.1304 & 0.1295 & 0.949 & 1.22 \\
&  & $\mathcal{M}_2$ & 1.568 & 0.8 & 0.1440 & 0.1429 & 0.955 & 0.6 & 0.1298 & 0.1272 & 0.945 & 1.23 \\
&  & $\mathcal{M}_3$ & 1.568 & 1.0 & 0.1405 & 0.1400 & 0.950 & 0.9 & 0.1298 & 0.1287 & 0.948 & 1.17 \\
&  & $\mathcal{M}_4$ & 1.568 & 1.0 & 0.1405 & 0.1400 & 0.950 & 0.8 & 0.1289 & 0.1263 & 0.947 & 1.19 \\
\cmidrule(lr){2-13}
& \multirow{4}{*}{24} & $\mathcal{M}_1$ & 1.480 & 1.0 & 0.1397 & 0.1437 & 0.961 & 0.6 & 0.1147 & 0.1180 & 0.966 & 1.48 \\
&  & $\mathcal{M}_2$ & 1.480 & 1.0 & 0.1397 & 0.1437 & 0.961 & $-$0.4 & 0.1137 & 0.1151 & 0.961 & 1.51 \\
&  & $\mathcal{M}_3$ & 1.480 & 1.4 & 0.1297 & 0.1330 & 0.962 & 1.2 & 0.1115 & 0.1148 & 0.963 & 1.35 \\
&  & $\mathcal{M}_4$ & 1.480 & 1.4 & 0.1297 & 0.1330 & 0.962 & 0.3 & 0.1098 & 0.1118 & 0.961 & 1.39 \\
\cmidrule(lr){2-13}
& \multirow{4}{*}{36} & $\mathcal{M}_1$ & 1.417 & 1.3 & 0.1623 & 0.1593 & 0.956 & 0.6 & 0.1186 & 0.1200 & 0.952 & 1.87 \\
&  & $\mathcal{M}_2$ & 1.417 & 1.3 & 0.1623 & 0.1593 & 0.956 & $-$1.0 & 0.1185 & 0.1173 & 0.942 & 1.87 \\
&  & $\mathcal{M}_3$ & 1.417 & 1.8 & 0.1365 & 0.1360 & 0.947 & 1.6 & 0.1107 & 0.1122 & 0.960 & 1.52 \\
&  & $\mathcal{M}_4$ & 1.417 & 1.8 & 0.1365 & 0.1360 & 0.947 & 0.0 & 0.1095 & 0.1094 & 0.958 & 1.55 \\
\bottomrule
\end{tabular}}
\end{table}

\begin{table}[htbp]
\centering
\caption{Simulation results for $\WR(\tau)$ under 80\% censoring by $\tau=36$ for the three-component prioritized endpoint. Same layout as the main-text simulation tables. RB: relative bias (\%) for $\widehat{\WR}(\tau)$ on the natural scale; MCSD: Monte Carlo standard deviation of $\log\widehat{\WR}(\tau)$; ASE: average estimated standard error for $\log\widehat{\WR}(\tau)$; Cov: empirical coverage of the 95\% Wald interval based on log-scale inference; RE $=\mathrm{MCSD}^2(C)/\mathrm{MCSD}^2(I)$ computed on the log scale.}
\label{tab:sim3-WR80}
\resizebox{\textwidth}{!}{%
\begin{tabular}{ccccrrrrrrrrr}
\toprule
 & & & & \multicolumn{4}{c}{$\widehat{\WR}^{(\text{ipcw})}(\tau)$} & \multicolumn{4}{c}{$\widehat{\WR}^{(\text{ctw})}(\tau)$} & \\
\cmidrule(lr){5-8}\cmidrule(lr){9-12}
$\theta$ & $\tau$ & $\mathcal{M}_k$ & True & RB\% & MCSD & ASE & Cov & RB\% & MCSD & ASE & Cov & RE \\
\midrule
\multirow{12}{*}{1.25}& \multirow{4}{*}{12} & $\mathcal{M}_1$ & 1.490 & 1.7 & 0.1576 & 0.1572 & 0.952 & 1.5 & 0.1257 & 0.1242 & 0.957 & 1.57 \\
&  & $\mathcal{M}_2$ & 1.490 & 1.7 & 0.1576 & 0.1572 & 0.952 & 1.3 & 0.1257 & 0.1236 & 0.954 & 1.57 \\
&  & $\mathcal{M}_3$ & 1.490 & 2.2 & 0.1460 & 0.1441 & 0.950 & 1.8 & 0.1225 & 0.1206 & 0.955 & 1.42 \\
&  & $\mathcal{M}_4$ & 1.490 & 2.2 & 0.1460 & 0.1441 & 0.950 & 1.6 & 0.1220 & 0.1199 & 0.957 & 1.43 \\
\cmidrule(lr){2-13}
& \multirow{4}{*}{24} & $\mathcal{M}_1$ & 1.402 & 3.0 & 0.2120 & 0.2062 & 0.947 & 1.3 & 0.1369 & 0.1343 & 0.955 & 2.40 \\
&  & $\mathcal{M}_2$ & 1.402 & 3.0 & 0.2120 & 0.2062 & 0.947 & 0.9 & 0.1369 & 0.1339 & 0.952 & 2.40 \\
&  & $\mathcal{M}_3$ & 1.402 & 3.3 & 0.1558 & 0.1564 & 0.943 & 2.0 & 0.1181 & 0.1179 & 0.956 & 1.74 \\
&  & $\mathcal{M}_4$ & 1.402 & 3.3 & 0.1558 & 0.1564 & 0.943 & 1.6 & 0.1177 & 0.1176 & 0.958 & 1.75 \\
\cmidrule(lr){2-13}
& \multirow{4}{*}{36} & $\mathcal{M}_1$ & 1.335 & 4.1 & 0.2897 & 0.2623 & 0.954 & 1.6 & 0.1797 & 0.1615 & 0.951 & 2.60 \\
&  & $\mathcal{M}_2$ & 1.335 & 4.1 & 0.2897 & 0.2623 & 0.954 & 1.2 & 0.1771 & 0.1601 & 0.952 & 2.68 \\
&  & $\mathcal{M}_3$ & 1.335 & 4.1 & 0.1743 & 0.1776 & 0.948 & 2.6 & 0.1238 & 0.1261 & 0.956 & 1.98 \\
&  & $\mathcal{M}_4$ & 1.335 & 4.1 & 0.1743 & 0.1776 & 0.948 & 2.3 & 0.1236 & 0.1258 & 0.957 & 1.99 \\
\midrule
\multirow{12}{*}{4.00}& \multirow{4}{*}{12} & $\mathcal{M}_1$ & 1.568 & 2.6 & 0.1891 & 0.1793 & 0.943 & 1.4 & 0.1476 & 0.1451 & 0.956 & 1.64 \\
&  & $\mathcal{M}_2$ & 1.568 & 2.6 & 0.1891 & 0.1793 & 0.943 & 1.0 & 0.1451 & 0.1404 & 0.953 & 1.70 \\
&  & $\mathcal{M}_3$ & 1.568 & 2.6 & 0.1710 & 0.1636 & 0.951 & 1.7 & 0.1416 & 0.1402 & 0.957 & 1.46 \\
&  & $\mathcal{M}_4$ & 1.568 & 2.6 & 0.1710 & 0.1636 & 0.951 & 1.3 & 0.1385 & 0.1357 & 0.953 & 1.52 \\
\cmidrule(lr){2-13}
& \multirow{4}{*}{24} & $\mathcal{M}_1$ & 1.480 & 3.2 & 0.2487 & 0.2283 & 0.928 & 1.6 & 0.1555 & 0.1536 & 0.945 & 2.56 \\
&  & $\mathcal{M}_2$ & 1.480 & 3.2 & 0.2487 & 0.2283 & 0.928 & 0.1 & 0.1522 & 0.1466 & 0.941 & 2.67 \\
&  & $\mathcal{M}_3$ & 1.480 & 3.0 & 0.1777 & 0.1720 & 0.943 & 2.3 & 0.1345 & 0.1334 & 0.945 & 1.74 \\
&  & $\mathcal{M}_4$ & 1.480 & 3.0 & 0.1777 & 0.1720 & 0.943 & 1.0 & 0.1312 & 0.1278 & 0.948 & 1.83 \\
\cmidrule(lr){2-13}
& \multirow{4}{*}{36} & $\mathcal{M}_1$ & 1.417 & 5.1 & 0.3227 & 0.2783 & 0.931 & 2.2 & 0.1916 & 0.1809 & 0.955 & 2.84 \\
&  & $\mathcal{M}_2$ & 1.417 & 5.1 & 0.3227 & 0.2783 & 0.931 & $-$0.1 & 0.1834 & 0.1706 & 0.946 & 3.09 \\
&  & $\mathcal{M}_3$ & 1.417 & 3.5 & 0.1893 & 0.1901 & 0.951 & 3.0 & 0.1386 & 0.1382 & 0.954 & 1.87 \\
&  & $\mathcal{M}_4$ & 1.417 & 3.5 & 0.1893 & 0.1901 & 0.951 & 1.0 & 0.1348 & 0.1323 & 0.952 & 1.97 \\
\bottomrule
\end{tabular}}
\end{table}

\begin{table}[htbp]
\centering
\caption{Simulation results for $\WO(\tau)$ under 20\% censoring by $\tau=36$ for the three-component prioritized endpoint. Same layout as the main-text simulation tables. RB: relative bias (\%) for $\widehat{\WO}(\tau)$ on the natural scale; MCSD: Monte Carlo standard deviation of $\log\widehat{\WO}(\tau)$; ASE: average estimated standard error for $\log\widehat{\WO}(\tau)$; Cov: empirical coverage of the 95\% Wald interval based on log-scale inference; RE $=\mathrm{MCSD}^2(C)/\mathrm{MCSD}^2(I)$ computed on the log scale.}
\label{tab:sim3-WO20}
\resizebox{\textwidth}{!}{%
\begin{tabular}{ccccrrrrrrrrr}
\toprule
 & & & & \multicolumn{4}{c}{$\widehat{\WO}^{(\text{ipcw})}(\tau)$} & \multicolumn{4}{c}{$\widehat{\WO}^{(\text{ctw})}(\tau)$} & \\
\cmidrule(lr){5-8}\cmidrule(lr){9-12}
$\theta$ & $\tau$ & $\mathcal{M}_k$ & True & RB\% & MCSD & ASE & Cov & RB\% & MCSD & ASE & Cov & RE \\
\midrule
\multirow{12}{*}{1.25}& \multirow{4}{*}{12} & $\mathcal{M}_1$ & 1.349 & 0.4 & 0.0816 & 0.0807 & 0.953 & 0.4 & 0.0801 & 0.0787 & 0.941 & 1.04 \\
&  & $\mathcal{M}_2$ & 1.349 & 0.4 & 0.0816 & 0.0807 & 0.953 & 0.4 & 0.0802 & 0.0788 & 0.942 & 1.03 \\
&  & $\mathcal{M}_3$ & 1.349 & 0.3 & 0.0815 & 0.0804 & 0.951 & 0.4 & 0.0800 & 0.0786 & 0.941 & 1.04 \\
&  & $\mathcal{M}_4$ & 1.349 & 0.3 & 0.0815 & 0.0804 & 0.951 & 0.4 & 0.0801 & 0.0787 & 0.942 & 1.03 \\
\cmidrule(lr){2-13}
& \multirow{4}{*}{24} & $\mathcal{M}_1$ & 1.362 & 0.3 & 0.0892 & 0.0882 & 0.946 & 0.3 & 0.0852 & 0.0843 & 0.946 & 1.10 \\
&  & $\mathcal{M}_2$ & 1.362 & 0.3 & 0.0892 & 0.0882 & 0.946 & 0.4 & 0.0856 & 0.0847 & 0.947 & 1.09 \\
&  & $\mathcal{M}_3$ & 1.362 & 0.4 & 0.0898 & 0.0879 & 0.945 & 0.4 & 0.0854 & 0.0842 & 0.944 & 1.10 \\
&  & $\mathcal{M}_4$ & 1.362 & 0.4 & 0.0898 & 0.0879 & 0.945 & 0.5 & 0.0858 & 0.0846 & 0.948 & 1.10 \\
\cmidrule(lr){2-13}
& \multirow{4}{*}{36} & $\mathcal{M}_1$ & 1.322 & 0.4 & 0.0906 & 0.0913 & 0.959 & 0.5 & 0.0869 & 0.0861 & 0.953 & 1.09 \\
&  & $\mathcal{M}_2$ & 1.322 & 0.4 & 0.0906 & 0.0913 & 0.959 & 0.6 & 0.0876 & 0.0868 & 0.952 & 1.07 \\
&  & $\mathcal{M}_3$ & 1.322 & 0.8 & 0.0920 & 0.0911 & 0.950 & 0.7 & 0.0873 & 0.0860 & 0.952 & 1.11 \\
&  & $\mathcal{M}_4$ & 1.322 & 0.8 & 0.0920 & 0.0911 & 0.950 & 0.9 & 0.0880 & 0.0866 & 0.950 & 1.09 \\
\midrule
\multirow{12}{*}{4.00}& \multirow{4}{*}{12} & $\mathcal{M}_1$ & 1.322 & 0.3 & 0.0729 & 0.0754 & 0.951 & 0.3 & 0.0711 & 0.0738 & 0.957 & 1.05 \\
&  & $\mathcal{M}_2$ & 1.322 & 0.3 & 0.0729 & 0.0754 & 0.951 & 0.5 & 0.0717 & 0.0742 & 0.954 & 1.03 \\
&  & $\mathcal{M}_3$ & 1.322 & 0.2 & 0.0728 & 0.0749 & 0.954 & 0.3 & 0.0711 & 0.0736 & 0.955 & 1.05 \\
&  & $\mathcal{M}_4$ & 1.322 & 0.2 & 0.0728 & 0.0749 & 0.954 & 0.5 & 0.0716 & 0.0740 & 0.950 & 1.03 \\
\cmidrule(lr){2-13}
& \multirow{4}{*}{24} & $\mathcal{M}_1$ & 1.376 & 0.5 & 0.0861 & 0.0858 & 0.949 & 0.6 & 0.0825 & 0.0825 & 0.949 & 1.09 \\
&  & $\mathcal{M}_2$ & 1.376 & 0.5 & 0.0861 & 0.0858 & 0.949 & 0.9 & 0.0835 & 0.0834 & 0.950 & 1.06 \\
&  & $\mathcal{M}_3$ & 1.376 & 0.5 & 0.0859 & 0.0852 & 0.948 & 0.7 & 0.0823 & 0.0824 & 0.949 & 1.09 \\
&  & $\mathcal{M}_4$ & 1.376 & 0.5 & 0.0859 & 0.0852 & 0.948 & 1.0 & 0.0833 & 0.0832 & 0.951 & 1.06 \\
\cmidrule(lr){2-13}
& \multirow{4}{*}{36} & $\mathcal{M}_1$ & 1.369 & 0.7 & 0.0908 & 0.0907 & 0.948 & 0.7 & 0.0869 & 0.0860 & 0.951 & 1.09 \\
&  & $\mathcal{M}_2$ & 1.369 & 0.7 & 0.0908 & 0.0907 & 0.948 & 1.0 & 0.0882 & 0.0874 & 0.945 & 1.06 \\
&  & $\mathcal{M}_3$ & 1.369 & 0.8 & 0.0903 & 0.0901 & 0.951 & 0.9 & 0.0866 & 0.0859 & 0.946 & 1.09 \\
&  & $\mathcal{M}_4$ & 1.369 & 0.8 & 0.0903 & 0.0901 & 0.951 & 1.3 & 0.0876 & 0.0872 & 0.948 & 1.06 \\
\bottomrule
\end{tabular}}
\end{table}

\begin{table}[htbp]
\centering
\caption{Simulation results for $\WO(\tau)$ under 40\% censoring by $\tau=36$ for the three-component prioritized endpoint. Same layout as the main-text simulation tables. RB: relative bias (\%) for $\widehat{\WO}(\tau)$ on the natural scale; MCSD: Monte Carlo standard deviation of $\log\widehat{\WO}(\tau)$; ASE: average estimated standard error for $\log\widehat{\WO}(\tau)$; Cov: empirical coverage of the 95\% Wald interval based on log-scale inference; RE $=\mathrm{MCSD}^2(C)/\mathrm{MCSD}^2(I)$ computed on the log scale.}
\label{tab:sim3-WO40}
\resizebox{\textwidth}{!}{%
\begin{tabular}{ccccrrrrrrrrr}
\toprule
 & & & & \multicolumn{4}{c}{$\widehat{\WO}^{(\text{ipcw})}(\tau)$} & \multicolumn{4}{c}{$\widehat{\WO}^{(\text{ctw})}(\tau)$} & \\
\cmidrule(lr){5-8}\cmidrule(lr){9-12}
$\theta$ & $\tau$ & $\mathcal{M}_k$ & True & RB\% & MCSD & ASE & Cov & RB\% & MCSD & ASE & Cov & RE \\
\midrule
\multirow{12}{*}{1.25}& \multirow{4}{*}{12} & $\mathcal{M}_1$ & 1.349 & 0.1 & 0.0847 & 0.0857 & 0.950 & 0.1 & 0.0805 & 0.0807 & 0.948 & 1.11 \\
&  & $\mathcal{M}_2$ & 1.349 & 0.1 & 0.0847 & 0.0857 & 0.950 & 0.2 & 0.0808 & 0.0810 & 0.949 & 1.10 \\
&  & $\mathcal{M}_3$ & 1.349 & $-$0.2 & 0.0839 & 0.0845 & 0.945 & 0.0 & 0.0801 & 0.0803 & 0.949 & 1.10 \\
&  & $\mathcal{M}_4$ & 1.349 & $-$0.2 & 0.0839 & 0.0845 & 0.945 & 0.1 & 0.0803 & 0.0806 & 0.946 & 1.09 \\
\cmidrule(lr){2-13}
& \multirow{4}{*}{24} & $\mathcal{M}_1$ & 1.362 & 0.3 & 0.0971 & 0.0990 & 0.958 & 0.4 & 0.0901 & 0.0885 & 0.948 & 1.16 \\
&  & $\mathcal{M}_2$ & 1.362 & 0.3 & 0.0971 & 0.0990 & 0.958 & 0.6 & 0.0909 & 0.0894 & 0.949 & 1.14 \\
&  & $\mathcal{M}_3$ & 1.362 & 0.4 & 0.0961 & 0.0968 & 0.947 & 0.5 & 0.0891 & 0.0879 & 0.949 & 1.16 \\
&  & $\mathcal{M}_4$ & 1.362 & 0.4 & 0.0961 & 0.0968 & 0.947 & 0.7 & 0.0898 & 0.0887 & 0.949 & 1.15 \\
\cmidrule(lr){2-13}
& \multirow{4}{*}{36} & $\mathcal{M}_1$ & 1.322 & 0.3 & 0.1069 & 0.1079 & 0.943 & 0.4 & 0.0951 & 0.0927 & 0.948 & 1.26 \\
&  & $\mathcal{M}_2$ & 1.322 & 0.3 & 0.1069 & 0.1079 & 0.943 & 0.7 & 0.0967 & 0.0944 & 0.944 & 1.22 \\
&  & $\mathcal{M}_3$ & 1.322 & 0.9 & 0.1040 & 0.1040 & 0.942 & 0.7 & 0.0928 & 0.0914 & 0.946 & 1.25 \\
&  & $\mathcal{M}_4$ & 1.322 & 0.9 & 0.1040 & 0.1040 & 0.942 & 1.1 & 0.0942 & 0.0930 & 0.945 & 1.22 \\
\midrule
\multirow{12}{*}{4.00}& \multirow{4}{*}{12} & $\mathcal{M}_1$ & 1.322 & 0.4 & 0.0805 & 0.0803 & 0.949 & 0.4 & 0.0754 & 0.0763 & 0.951 & 1.14 \\
&  & $\mathcal{M}_2$ & 1.322 & 0.4 & 0.0805 & 0.0803 & 0.949 & 0.9 & 0.0768 & 0.0772 & 0.946 & 1.10 \\
&  & $\mathcal{M}_3$ & 1.322 & 0.1 & 0.0789 & 0.0789 & 0.945 & 0.2 & 0.0746 & 0.0757 & 0.952 & 1.12 \\
&  & $\mathcal{M}_4$ & 1.322 & 0.1 & 0.0789 & 0.0789 & 0.945 & 0.8 & 0.0759 & 0.0766 & 0.953 & 1.08 \\
\cmidrule(lr){2-13}
& \multirow{4}{*}{24} & $\mathcal{M}_1$ & 1.376 & 0.6 & 0.0963 & 0.0968 & 0.952 & 0.5 & 0.0854 & 0.0875 & 0.950 & 1.27 \\
&  & $\mathcal{M}_2$ & 1.376 & 0.6 & 0.0963 & 0.0968 & 0.952 & 1.3 & 0.0885 & 0.0895 & 0.945 & 1.18 \\
&  & $\mathcal{M}_3$ & 1.376 & 0.4 & 0.0936 & 0.0938 & 0.950 & 0.5 & 0.0843 & 0.0865 & 0.951 & 1.23 \\
&  & $\mathcal{M}_4$ & 1.376 & 0.4 & 0.0936 & 0.0938 & 0.950 & 1.3 & 0.0871 & 0.0885 & 0.944 & 1.16 \\
\cmidrule(lr){2-13}
& \multirow{4}{*}{36} & $\mathcal{M}_1$ & 1.369 & 0.6 & 0.1064 & 0.1078 & 0.959 & 0.5 & 0.0930 & 0.0935 & 0.954 & 1.31 \\
&  & $\mathcal{M}_2$ & 1.369 & 0.6 & 0.1064 & 0.1078 & 0.959 & 1.3 & 0.0980 & 0.0970 & 0.946 & 1.18 \\
&  & $\mathcal{M}_3$ & 1.369 & 0.8 & 0.1027 & 0.1029 & 0.954 & 0.9 & 0.0918 & 0.0920 & 0.952 & 1.25 \\
&  & $\mathcal{M}_4$ & 1.369 & 0.8 & 0.1027 & 0.1029 & 0.954 & 1.8 & 0.0962 & 0.0953 & 0.951 & 1.14 \\
\bottomrule
\end{tabular}}
\end{table}

\begin{table}[htbp]
\centering
\caption{Simulation results for $\WO(\tau)$ under 60\% censoring by $\tau=36$ for the three-component prioritized endpoint. Same layout as the main-text simulation tables. RB: relative bias (\%) for $\widehat{\WO}(\tau)$ on the natural scale; MCSD: Monte Carlo standard deviation of $\log\widehat{\WO}(\tau)$; ASE: average estimated standard error for $\log\widehat{\WO}(\tau)$; Cov: empirical coverage of the 95\% Wald interval based on log-scale inference; RE $=\mathrm{MCSD}^2(C)/\mathrm{MCSD}^2(I)$ computed on the log scale.}
\label{tab:sim3-WO60}
\resizebox{\textwidth}{!}{%
\begin{tabular}{ccccrrrrrrrrr}
\toprule
 & & & & \multicolumn{4}{c}{$\widehat{\WO}^{(\text{ipcw})}(\tau)$} & \multicolumn{4}{c}{$\widehat{\WO}^{(\text{ctw})}(\tau)$} & \\
\cmidrule(lr){5-8}\cmidrule(lr){9-12}
$\theta$ & $\tau$ & $\mathcal{M}_k$ & True & RB\% & MCSD & ASE & Cov & RB\% & MCSD & ASE & Cov & RE \\
\midrule
\multirow{12}{*}{1.25}& \multirow{4}{*}{12} & $\mathcal{M}_1$ & 1.349 & 0.3 & 0.0938 & 0.0951 & 0.950 & 0.4 & 0.0843 & 0.0844 & 0.948 & 1.24 \\
&  & $\mathcal{M}_2$ & 1.349 & 0.3 & 0.0938 & 0.0951 & 0.950 & 0.6 & 0.0848 & 0.0850 & 0.950 & 1.22 \\
&  & $\mathcal{M}_3$ & 1.349 & 0.1 & 0.0927 & 0.0915 & 0.944 & 0.3 & 0.0839 & 0.0833 & 0.948 & 1.22 \\
&  & $\mathcal{M}_4$ & 1.349 & 0.1 & 0.0927 & 0.0915 & 0.944 & 0.6 & 0.0843 & 0.0838 & 0.947 & 1.21 \\
\cmidrule(lr){2-13}
& \multirow{4}{*}{24} & $\mathcal{M}_1$ & 1.362 & 0.6 & 0.1125 & 0.1219 & 0.964 & 0.4 & 0.0937 & 0.0971 & 0.959 & 1.44 \\
&  & $\mathcal{M}_2$ & 1.362 & 0.6 & 0.1125 & 0.1219 & 0.964 & 0.7 & 0.0957 & 0.0989 & 0.961 & 1.38 \\
&  & $\mathcal{M}_3$ & 1.362 & 0.9 & 0.1089 & 0.1115 & 0.951 & 0.8 & 0.0930 & 0.0941 & 0.949 & 1.37 \\
&  & $\mathcal{M}_4$ & 1.362 & 0.9 & 0.1089 & 0.1115 & 0.951 & 1.1 & 0.0945 & 0.0957 & 0.947 & 1.33 \\
\cmidrule(lr){2-13}
& \multirow{4}{*}{36} & $\mathcal{M}_1$ & 1.322 & 0.8 & 0.1389 & 0.1483 & 0.971 & 0.9 & 0.1061 & 0.1085 & 0.956 & 1.72 \\
&  & $\mathcal{M}_2$ & 1.322 & 0.8 & 0.1389 & 0.1483 & 0.971 & 1.6 & 0.1103 & 0.1121 & 0.956 & 1.59 \\
&  & $\mathcal{M}_3$ & 1.322 & 2.1 & 0.1213 & 0.1260 & 0.956 & 1.6 & 0.0989 & 0.1011 & 0.956 & 1.50 \\
&  & $\mathcal{M}_4$ & 1.322 & 2.1 & 0.1213 & 0.1260 & 0.956 & 2.3 & 0.1024 & 0.1041 & 0.953 & 1.40 \\
\midrule
\multirow{12}{*}{4.00}& \multirow{4}{*}{12} & $\mathcal{M}_1$ & 1.322 & 0.2 & 0.0884 & 0.0890 & 0.955 & 0.3 & 0.0808 & 0.0805 & 0.954 & 1.20 \\
&  & $\mathcal{M}_2$ & 1.322 & 0.2 & 0.0884 & 0.0890 & 0.955 & 1.4 & 0.0837 & 0.0821 & 0.946 & 1.12 \\
&  & $\mathcal{M}_3$ & 1.322 & $-$0.3 & 0.0843 & 0.0850 & 0.953 & 0.1 & 0.0794 & 0.0789 & 0.948 & 1.13 \\
&  & $\mathcal{M}_4$ & 1.322 & $-$0.3 & 0.0843 & 0.0850 & 0.953 & 1.1 & 0.0821 & 0.0805 & 0.940 & 1.05 \\
\cmidrule(lr){2-13}
& \multirow{4}{*}{24} & $\mathcal{M}_1$ & 1.376 & 0.5 & 0.1134 & 0.1196 & 0.972 & 0.5 & 0.0936 & 0.0972 & 0.970 & 1.47 \\
&  & $\mathcal{M}_2$ & 1.376 & 0.5 & 0.1134 & 0.1196 & 0.972 & 1.9 & 0.0991 & 0.1013 & 0.966 & 1.31 \\
&  & $\mathcal{M}_3$ & 1.376 & 0.2 & 0.1030 & 0.1076 & 0.967 & 0.6 & 0.0899 & 0.0933 & 0.960 & 1.31 \\
&  & $\mathcal{M}_4$ & 1.376 & 0.2 & 0.1030 & 0.1076 & 0.967 & 2.1 & 0.0946 & 0.0971 & 0.964 & 1.19 \\
\cmidrule(lr){2-13}
& \multirow{4}{*}{36} & $\mathcal{M}_1$ & 1.369 & 0.8 & 0.1476 & 0.1498 & 0.968 & 0.6 & 0.1078 & 0.1104 & 0.961 & 1.87 \\
&  & $\mathcal{M}_2$ & 1.369 & 0.8 & 0.1476 & 0.1498 & 0.968 & 1.9 & 0.1173 & 0.1174 & 0.960 & 1.58 \\
&  & $\mathcal{M}_3$ & 1.369 & 1.0 & 0.1212 & 0.1242 & 0.960 & 1.2 & 0.0997 & 0.1022 & 0.960 & 1.48 \\
&  & $\mathcal{M}_4$ & 1.369 & 1.0 & 0.1212 & 0.1242 & 0.960 & 2.7 & 0.1073 & 0.1084 & 0.955 & 1.28 \\
\bottomrule
\end{tabular}}
\end{table}

\begin{table}[htbp]
\centering
\caption{Simulation results for $\WO(\tau)$ under 80\% censoring by $\tau=36$ for the three-component prioritized endpoint. Same layout as the main-text simulation tables. RB: relative bias (\%) for $\widehat{\WO}(\tau)$ on the natural scale; MCSD: Monte Carlo standard deviation of $\log\widehat{\WO}(\tau)$; ASE: average estimated standard error for $\log\widehat{\WO}(\tau)$; Cov: empirical coverage of the 95\% Wald interval based on log-scale inference; RE $=\mathrm{MCSD}^2(C)/\mathrm{MCSD}^2(I)$ computed on the log scale.}
\label{tab:sim3-WO80}
\resizebox{\textwidth}{!}{%
\begin{tabular}{ccccrrrrrrrrr}
\toprule
 & & & & \multicolumn{4}{c}{$\widehat{\WO}^{(\text{ipcw})}(\tau)$} & \multicolumn{4}{c}{$\widehat{\WO}^{(\text{ctw})}(\tau)$} & \\
\cmidrule(lr){5-8}\cmidrule(lr){9-12}
$\theta$ & $\tau$ & $\mathcal{M}_k$ & True & RB\% & MCSD & ASE & Cov & RB\% & MCSD & ASE & Cov & RE \\
\midrule
\multirow{12}{*}{1.25}& \multirow{4}{*}{12} & $\mathcal{M}_1$ & 1.349 & 0.8 & 0.1169 & 0.1202 & 0.963 & 0.9 & 0.0938 & 0.0935 & 0.957 & 1.55 \\
&  & $\mathcal{M}_2$ & 1.349 & 0.8 & 0.1169 & 0.1202 & 0.963 & 1.3 & 0.0954 & 0.0947 & 0.956 & 1.50 \\
&  & $\mathcal{M}_3$ & 1.349 & 0.2 & 0.1052 & 0.1058 & 0.955 & 0.6 & 0.0902 & 0.0893 & 0.954 & 1.36 \\
&  & $\mathcal{M}_4$ & 1.349 & 0.2 & 0.1052 & 0.1058 & 0.955 & 1.0 & 0.0914 & 0.0903 & 0.954 & 1.32 \\
\cmidrule(lr){2-13}
& \multirow{4}{*}{24} & $\mathcal{M}_1$ & 1.362 & 2.4 & 0.2063 & 0.2044 & 0.962 & 0.9 & 0.1270 & 0.1252 & 0.963 & 2.64 \\
&  & $\mathcal{M}_2$ & 1.362 & 2.4 & 0.2063 & 0.2044 & 0.962 & 1.6 & 0.1313 & 0.1290 & 0.961 & 2.47 \\
&  & $\mathcal{M}_3$ & 1.362 & 1.9 & 0.1393 & 0.1438 & 0.959 & 1.4 & 0.1071 & 0.1080 & 0.960 & 1.69 \\
&  & $\mathcal{M}_4$ & 1.362 & 1.9 & 0.1393 & 0.1438 & 0.959 & 2.1 & 0.1101 & 0.1110 & 0.961 & 1.60 \\
\cmidrule(lr){2-13}
& \multirow{4}{*}{36} & $\mathcal{M}_1$ & 1.322 & 8.6 & 0.3548 & 0.3338 & 0.966 & 4.0 & 0.2184 & 0.1855 & 0.962 & 2.64 \\
&  & $\mathcal{M}_2$ & 1.322 & 8.6 & 0.3548 & 0.3338 & 0.966 & 6.5 & 0.2285 & 0.1988 & 0.969 & 2.41 \\
&  & $\mathcal{M}_3$ & 1.322 & 3.5 & 0.1676 & 0.1768 & 0.966 & 2.3 & 0.1196 & 0.1233 & 0.958 & 1.96 \\
&  & $\mathcal{M}_4$ & 1.322 & 3.5 & 0.1676 & 0.1768 & 0.966 & 3.5 & 0.1253 & 0.1291 & 0.959 & 1.79 \\
\midrule
\multirow{12}{*}{4.00}& \multirow{4}{*}{12} & $\mathcal{M}_1$ & 1.322 & 0.9 & 0.1154 & 0.1130 & 0.961 & 0.7 & 0.0914 & 0.0907 & 0.965 & 1.59 \\
&  & $\mathcal{M}_2$ & 1.322 & 0.9 & 0.1154 & 0.1130 & 0.961 & 2.5 & 0.0965 & 0.0942 & 0.955 & 1.43 \\
&  & $\mathcal{M}_3$ & 1.322 & $-$0.2 & 0.0998 & 0.0979 & 0.949 & 0.2 & 0.0855 & 0.0854 & 0.957 & 1.36 \\
&  & $\mathcal{M}_4$ & 1.322 & $-$0.2 & 0.0998 & 0.0979 & 0.949 & 2.0 & 0.0897 & 0.0886 & 0.957 & 1.24 \\
\cmidrule(lr){2-13}
& \multirow{4}{*}{24} & $\mathcal{M}_1$ & 1.376 & 2.0 & 0.2111 & 0.2035 & 0.956 & 1.2 & 0.1290 & 0.1293 & 0.961 & 2.68 \\
&  & $\mathcal{M}_2$ & 1.376 & 2.0 & 0.2111 & 0.2035 & 0.956 & 3.8 & 0.1401 & 0.1372 & 0.952 & 2.27 \\
&  & $\mathcal{M}_3$ & 1.376 & 0.6 & 0.1377 & 0.1379 & 0.948 & 1.1 & 0.1077 & 0.1084 & 0.950 & 1.64 \\
&  & $\mathcal{M}_4$ & 1.376 & 0.6 & 0.1377 & 0.1379 & 0.948 & 3.8 & 0.1168 & 0.1154 & 0.942 & 1.39 \\
\cmidrule(lr){2-13}
& \multirow{4}{*}{36} & $\mathcal{M}_1$ & 1.369 & 5.8 & 0.3410 & 0.3105 & 0.945 & 4.5 & 0.2128 & 0.1958 & 0.967 & 2.57 \\
&  & $\mathcal{M}_2$ & 1.369 & 5.8 & 0.3410 & 0.3105 & 0.945 & 6.9 & 0.2239 & 0.2055 & 0.965 & 2.32 \\
&  & $\mathcal{M}_3$ & 1.369 & 1.7 & 0.1659 & 0.1729 & 0.962 & 2.2 & 0.1248 & 0.1264 & 0.960 & 1.77 \\
&  & $\mathcal{M}_4$ & 1.369 & 1.7 & 0.1659 & 0.1729 & 0.962 & 5.2 & 0.1384 & 0.1377 & 0.944 & 1.44 \\
\bottomrule
\end{tabular}}
\end{table}

\subsection{Copula sensitivity study: supplementary results}
\label{supp:sim_copula}

This subsection reports results from the copula-sensitivity study. In contrast to the original nuisance-model misspecification study, the marginal event-time models are correctly specified here through separate Cox proportional hazards models within each treatment arm and component. The sensitivity analysis therefore isolates the effect of the dependence model used to estimate the conditional tie probabilities. The data-generating copula and the working copula are varied over Gumbel--Hougaard, Clayton, Frank, and Plackett, while the IPCW estimator of \citet{cui2025ipcw} is included as a reference because it does not use a working copula. The four copulas used in this study are summarized in Table \ref{tab:copula-main} of the main text, and the broader copula library is described in Table \ref{tab:supp-copula}.

In this section, we provide the corresponding results under 40\% censoring and the analogous results for $\WR(\tau)$ and $\WO(\tau)$. Across these additional settings, the proposed estimator remains substantially more efficient than the IPCW estimator of \citet{cui2025ipcw}. under all working copulas considered. The efficiency gain is smaller under 40\% censoring than under 80\% censoring, as expected, because fewer lower-priority comparisons are blocked by unresolved higher-priority ties when censoring is moderate. Within each data-generating copula, the four proposed-method columns are similar in RBias, COV, and RE, indicating that the efficiency gain is not driven by exact matching of the working copula to the data-generating copula. Instead, when the Cox marginal models are correctly specified, the dominant source of improvement is from conditional tie probabilities.

\begin{figure}[htbp!]
\centering
\includegraphics[width=\textwidth]{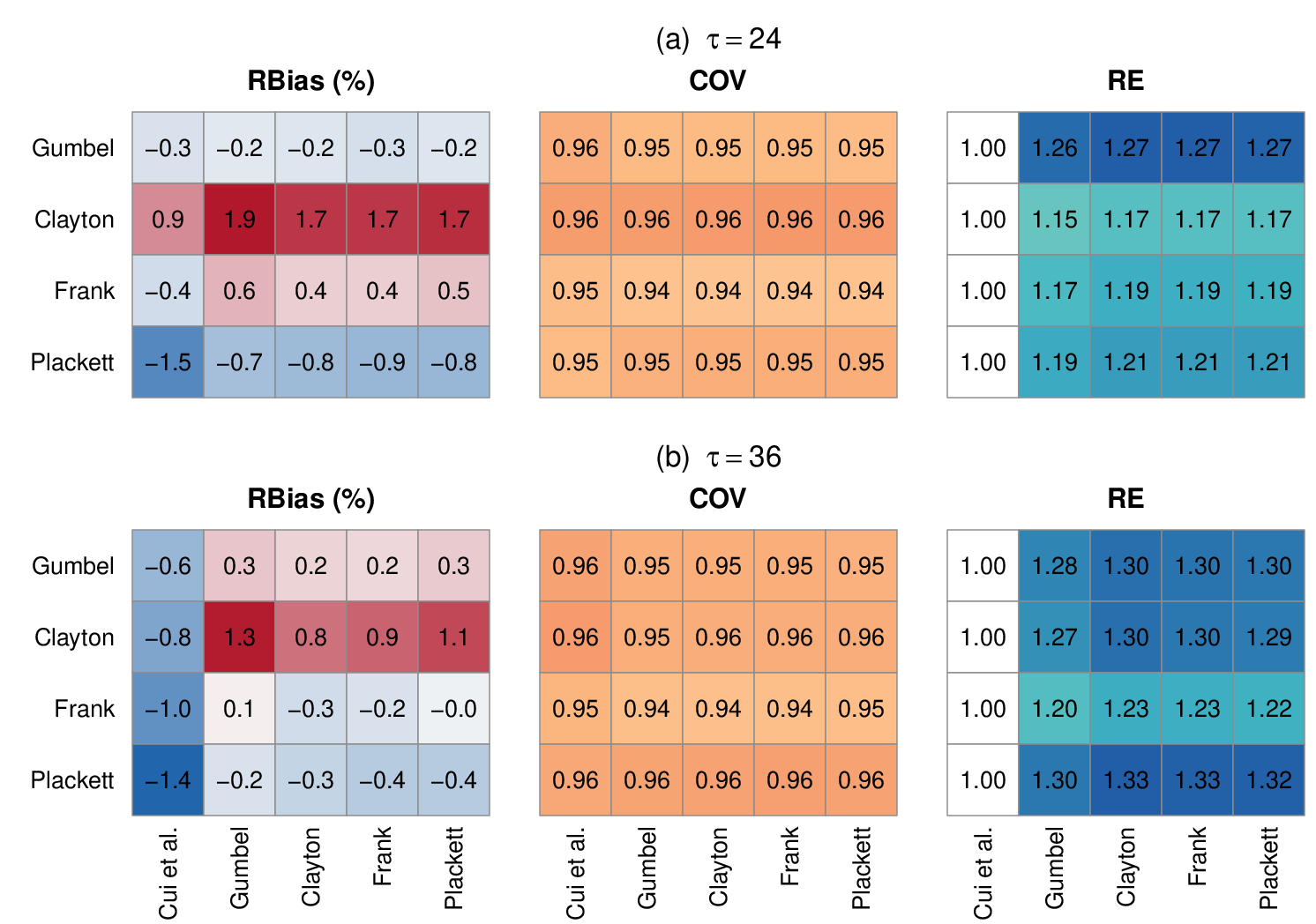}
\caption{Copula sensitivity of $\NB(\tau)$ under 40\% censoring. The two stacked blocks correspond to (a) $\tau=12$, (b) $\tau=24$, and (c) $\tau=36$ months. Within each block, rows represent the data-generating copula and columns represent \citet{cui2025ipcw} or the working copula used by the proposed estimator. The three panels display relative bias (RBias), empirical coverage of the nominal 95\% confidence interval (COV), and efficiency gain relative to \citet{cui2025ipcw} (RE), respectively.}
\label{fig:supp-copula-NB40}
\end{figure}

\begin{figure}[htbp!]
\centering
\includegraphics[width=\textwidth]{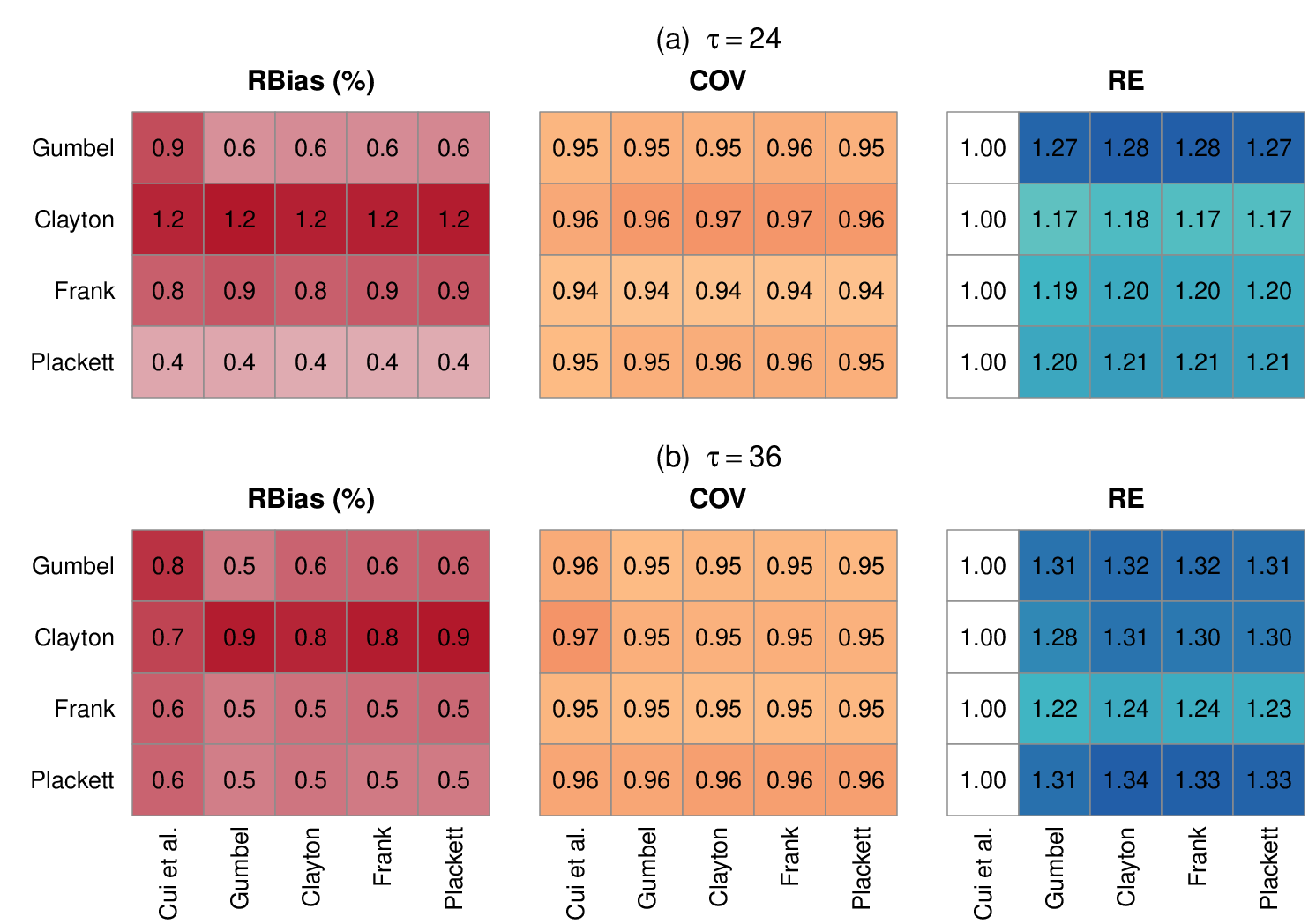}
\caption{Copula sensitivity of $\WR(\tau)$ under 40\% censoring. The two stacked blocks correspond to (a) $\tau=12$, (b) $\tau=24$, and (c) $\tau=36$ months. Within each block, rows represent the data-generating copula and columns represent \citet{cui2025ipcw} or the working copula used by the proposed estimator. The three panels display relative bias (RBias), empirical coverage of the nominal 95\% confidence interval (COV), and efficiency gain relative to \citet{cui2025ipcw} (RE), respectively.}
\label{fig:supp-copula-WR40}
\end{figure}

\begin{figure}[htbp!]
\centering
\includegraphics[width=\textwidth]{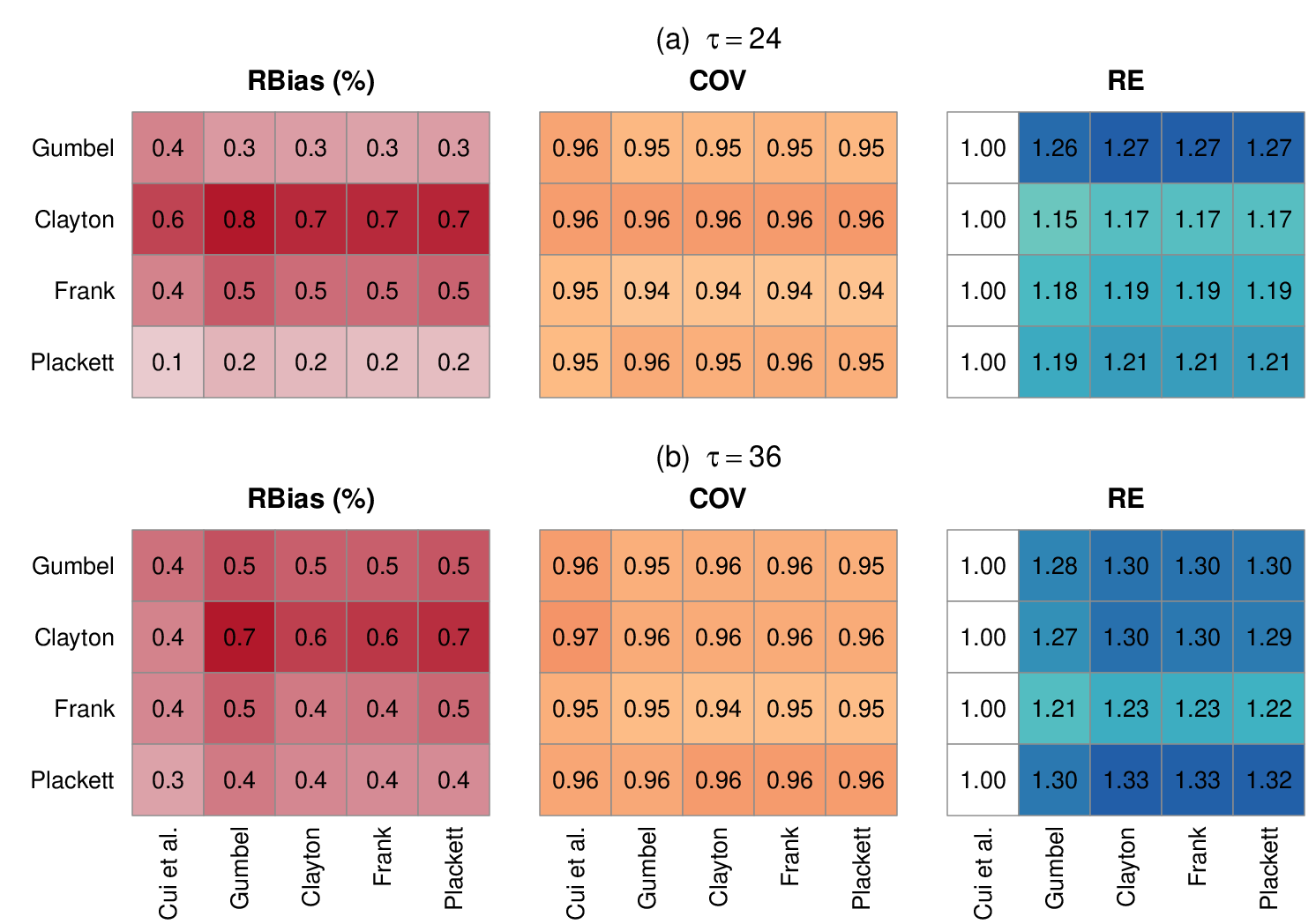}
\caption{Copula sensitivity of $\WO(\tau)$ under 40\% censoring. The two stacked blocks correspond to (a) $\tau=12$, (b) $\tau=24$, and (c) $\tau=36$ months. Within each block, rows represent the data-generating copula and columns represent \citet{cui2025ipcw} or the working copula used by the proposed estimator. The three panels display relative bias (RBias), empirical coverage of the nominal 95\% confidence interval (COV), and efficiency gain relative to \citet{cui2025ipcw} (RE), respectively.}
\label{fig:supp-copula-WO40}
\end{figure}

\begin{figure}[htbp!]
\centering
\includegraphics[width=\textwidth]{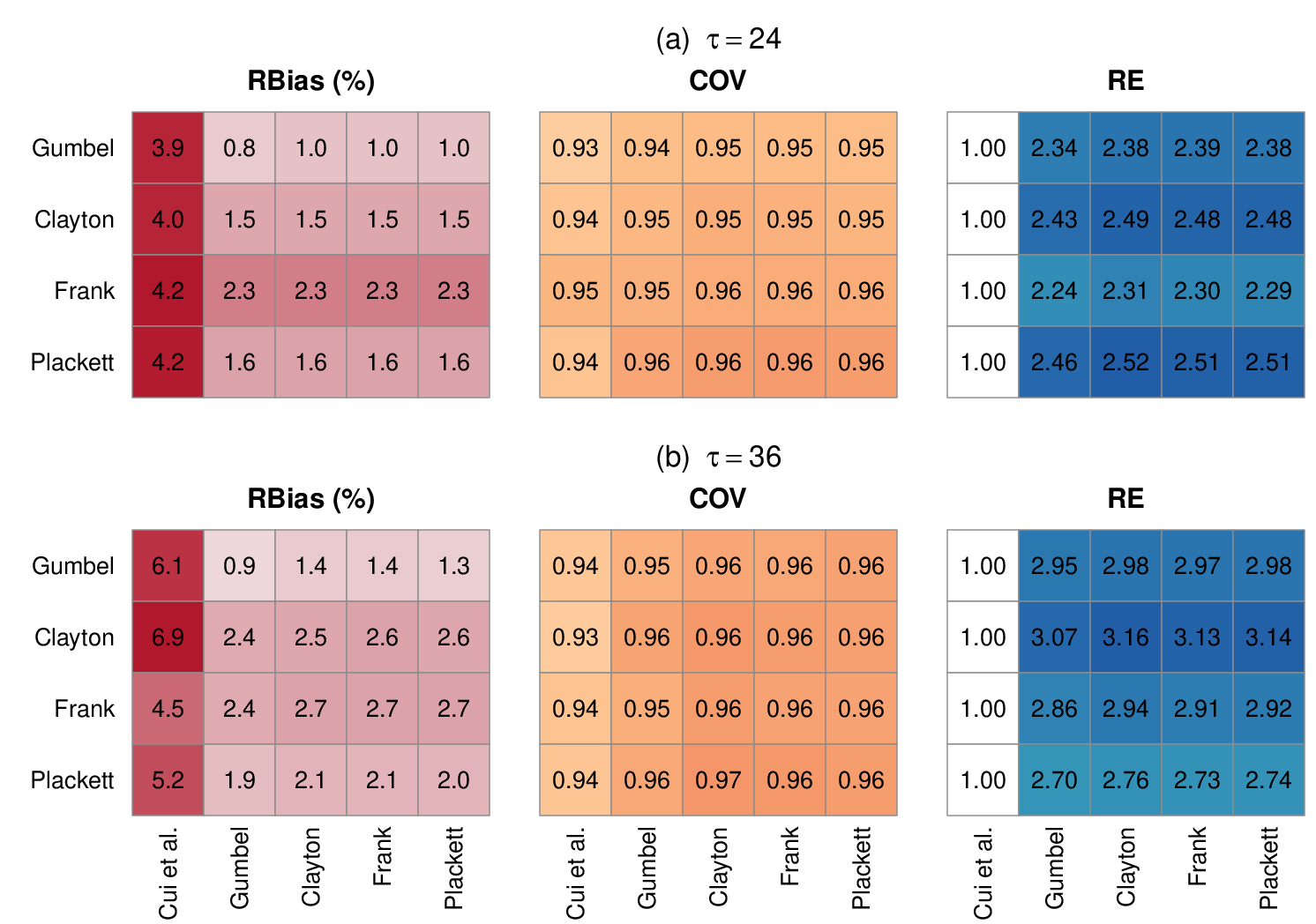}
\caption{Copula sensitivity of $\WR(\tau)$ under 80\% censoring. The two stacked blocks correspond to (a) $\tau=12$, (b) $\tau=24$, and (c) $\tau=36$ months. Within each block, rows represent the data-generating copula and columns represent \citet{cui2025ipcw} or the working copula used by the proposed estimator. The three panels display relative bias (RBias), empirical coverage of the nominal 95\% confidence interval (COV), and efficiency gain relative to \citet{cui2025ipcw} (RE), respectively.}
\label{fig:supp-copula-WR80}
\end{figure}

\begin{figure}[htbp!]
\centering
\includegraphics[width=\textwidth]{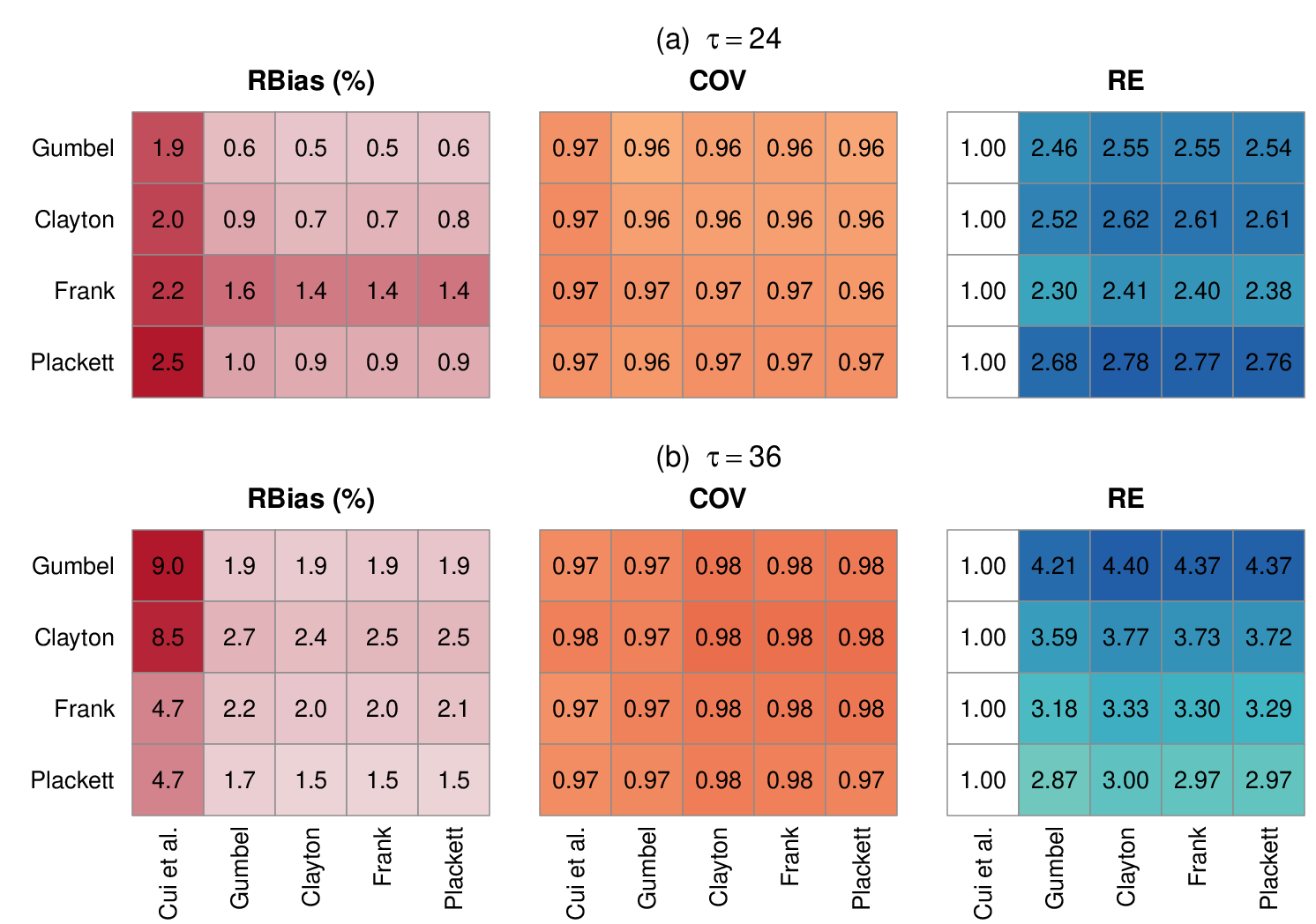}
\caption{Copula sensitivity of $\WO(\tau)$ under 80\% censoring. The three stacked blocks correspond to (a) $\tau=12$, (b) $\tau=24$, and (c) $\tau=36$ months. Within each block, rows represent the data-generating copula and columns represent \citet{cui2025ipcw} or the working copula used by the proposed estimator. The three panels display relative bias (RBias), empirical coverage of the nominal 95\% confidence interval (COV), and efficiency gain relative to \citet{cui2025ipcw} (RE), respectively.}
\label{fig:supp-copula-WO80}
\end{figure}

Taken together, the supplementary results support the conclusions in the main text. In the original nuisance-model misspecification study, the proposed estimator improves efficiency relative to IPCW estimator of \citet{cui2025ipcw} across censoring levels, restriction times, win-statistic summaries, and endpoint hierarchies, with the largest gains under heavier censoring and longer restriction horizons. Accurate censoring-model specification remains important for bias and interval performance because both estimators rely on inverse probability of censoring weighting. In the copula-sensitivity study, when the event-time margins are correctly specified by Cox models, the proposed estimator remains substantially more efficient than \citet{cui2025ipcw} across data-generating and working copula families, and the working copula choice has a smaller impact than the transition from IPCW estimator to conditional tie weighting.

\section{Computation of true values in the simulation study} \label{supp:truevalues}

For each simulation setting, the true component specific win probabilities were computed deterministically under the uncensored data-generating mechanism. This section describes the calculation used for both the two-component and three-component prioritized endpoints considered in Section \ref{sec:simulation} and  \ref{supp:sim3}. Let $Q$ denote the number of prioritized event types. For component $q=1,\dots,Q$, let
\[
\Lambda_{q,a}(t\mid \bZ)=\lambda_{q0}\,t^{\rho_q}\exp\!\left(\beta_{q1}Z_1+\beta_{q2}Z_2+\beta_{q3}Z_3+\beta_{qa}a\right), \qquad a\in\{0,1\},
\]
be the marginal cumulative hazard under arm $a$, with corresponding marginal survival function and density
\[
S_{q,a}(t\mid \bZ)=\exp\{-\Lambda_{q,a}(t\mid \bZ)\},
\]
and 
\[f_{q,a}(t\mid \bZ)=\lambda_{q0}\rho_q t^{\rho_q-1}\exp\!\left(\beta_{q1}Z_1+\beta_{q2}Z_2+\beta_{q3}Z_3+\beta_{qa}a\right)
S_{q,a}(t\mid \bZ).\]
Dependence among the $Q$ latent event times was introduced through an exchangeable Gumbel copula with parameter $\theta\ge1$. For a fixed priority level $q$, define
\[
A_{q,a}(t,\tau\mid \bZ)=\sum_{k=1}^{q-1}\Lambda_{k,a}(\tau\mid \bZ)^\theta+\Lambda_{q,a}(t\mid \bZ)^\theta.
\]
Under the exchangeable Gumbel copula, the joint survival function needed for the $q$th priority comparison takes the form
\begin{align} \label{eq:supp-true-S}
\mathcal{S}_{q,a}(\tau,t\mid \bZ)&=\Pr\!\bigl(T_{1}> \tau,\dots,T_{q-1}>\tau,\;T_q>t \mid A=a,\bZ\bigr) \nonumber\\
& =\exp\!\left\{-A_{q,a}(t,\tau\mid \bZ)^{1/\theta}\right\}.  
\end{align}
Differentiating \eqref{eq:supp-true-S} with respect to $t$ gives the corresponding prefix subdensity
\begin{align*}
    \mathcal{H}_{q,a}(\tau,t\mid \bZ)&=-\frac{\partial}{\partial t}\mathcal{S}_{q,a}(\tau,t\mid \bZ)\\
    & =\mathcal{S}_{q,a}(\tau,t\mid \bZ)\Lambda_{q,a}(t\mid \bZ)^{\theta-1}A_{q,a}(t,\tau\mid \bZ)^{1/\theta-1}\\
    &\times \lambda_{q0}\rho_q t^{\rho_q-1}\exp\!\left(\beta_{q1}Z_1+\beta_{q2}Z_2+\beta_{q3}Z_3+\beta_{qa}a\right),
\end{align*}
with the convention that when $\theta=1$ this reduces to
\[
\mathcal{H}_{q,a}(\tau,t\mid \bZ)=\mathcal{S}_{q,a}(\tau,t\mid \bZ)\,\lambda_{q0}\rho_q t^{\rho_q-1}\exp\!\left(\beta_{q1}Z_1+\beta_{q2}Z_2+\beta_{q3}Z_3+\beta_{qa}a\right).
\]
For a treated subject with covariates $\bZ_i$ and a control subject with covariates $\bZ_j$, the conditional treatment-win probability contributed by the $q$th priority level is
\begin{equation}\label{eq:supp-true-pitq-cond}
\pi_{tq}(\tau\mid \bZ_i,\bZ_j)=\int_0^\tau\mathcal{S}_{q,1}(\tau,t\mid \bZ_i)\,\mathcal{H}_{q,0}(\tau,t\mid \bZ_j)\,\mathrm{d}t.
\end{equation}
This quantity is the probability that the treated subject remains event-free through $\tau$ on all higher-priority components and beyond time $t$ on component $q$, while the control subject remains event-free through $\tau$ on all higher-priority components and experiences the component-$q$ event at time $t$. The corresponding conditional control-win probability is obtained symmetrically by interchanging the treatment and control arms,
\begin{equation}\label{eq:supp-true-picq-cond}
\pi_{cq}(\tau\mid \bZ_i,\bZ_j)=
\int_0^\tau\mathcal{S}_{q,0}(\tau,t\mid \bZ_j)\,\mathcal{H}_{q,1}(\tau,t\mid \bZ_i)\,\mathrm{d}t.
\end{equation}
The marginal component-specific win probabilities were then obtained by integrating
\eqref{eq:supp-true-pitq-cond}--\eqref{eq:supp-true-picq-cond} over the joint distribution of $(\bZ_i,\bZ_j)$. Because the baseline covariates were generated independently with
\[
Z_1\sim \mathrm{Bernoulli}(0.5),\qquad Z_2\sim \mathrm{Uniform}(0,1),\qquad Z_3\sim \mathrm{Bernoulli}(0.4),
\]
the outer expectation factors across arms. In implementation, the expectation over the binary covariates $(Z_1,Z_3)$ was evaluated exactly by summation over the four possible values with probability
\[
\Pr(Z_1=z_1,Z_3=z_3)=0.5\times 0.4^{z_3}0.6^{1-z_3},
\]
and the expectation over the continuous covariate $Z_2$ in each arm was approximated by 24-point Gauss--Legendre quadrature on $[0,1]$. Thus,
\begin{align}
\pi_{tq}(\tau)
&=\sum_{z_{1i}=0}^1\sum_{z_{3i}=0}^1\sum_{z_{1j}=0}^1\sum_{z_{3j}=0}^1\Pr(Z_{1i}=z_{1i},Z_{3i}=z_{3i})\Pr(Z_{1j}=z_{1j},Z_{3j}=z_{3j})\nonumber\\
&\quad\times\int_0^1\int_0^1\pi_{tq}\!\left(\tau\mid z_{1i},z_{2i},z_{3i},z_{1j},z_{2j},z_{3j}\right)\,\mathrm{d}z_{2i}\,\mathrm{d}z_{2j},
\label{eq:supp-true-pitq}
\end{align}
and $\pi_{cq}(\tau)$ was computed analogously. The inner integral over $t\in[0,\tau]$ in \eqref{eq:supp-true-pitq-cond} and \eqref{eq:supp-true-picq-cond} was evaluated by 80-point Gauss--Legendre quadrature on $[0,\tau]$. After obtaining the component-specific win probabilities, the overall treatment win and control win probabilities were formed as
\[
\pi_t(\tau)=\sum_{q=1}^Q \pi_{tq}(\tau),
\qquad
\pi_c(\tau)=\sum_{q=1}^Q \pi_{cq}(\tau).
\]
The truth values of the summary measures reported in the simulation study were then computed by direct transformation:
\[
\NB(\tau)=\pi_t(\tau)-\pi_c(\tau),\qquad\WR(\tau)=\frac{\pi_t(\tau)}{\pi_c(\tau)},\qquad\WO(\tau)=\frac{1+\NB(\tau)}{1-\NB(\tau)}.
\]
By construction, these truth calculations depend only on the latent event-time model, the copula dependence parameter, the endpoint hierarchy, and the restriction time $\tau$.

\section{Tutorial for the \texttt{winIPCW} R package}
\label{supp:package_tutorial}

This section provides a brief tutorial for applying the \texttt{winIPCW} R package at \\\url{https://github.com/fancy575/winIPCW} to estimate restricted win statistics for prioritized time-to-event composite endpoints under right censoring. The package implements the IPCW estimator of \citet{cui2025ipcw} and the proposed modified IPCW estimator with conditional tie weighting described in the main text. The example below uses the simulated data set \texttt{windat}, which is included with the package. The data are stored in long format, with one follow-up or censoring row per subject coded by \texttt{event\_type = 0}, and additional rows for observed prioritized endpoint components. In this example, \texttt{event\_type = 1} is treated as the first-priority endpoint and \texttt{event\_type = 2} as the second-priority endpoint.

\begin{verbatim}
## Install winIPCW from GitHub
if (!requireNamespace("remotes", quietly = TRUE)) {
  install.packages("remotes")
}

remotes::install_github("fancy575/winIPCW",
                        dependencies = TRUE,
                        upgrade = "never")

## Load packages
library(winIPCW)
library(survival)

## Load example data
data("windat", package = "winIPCW")

## Fit restricted win statistics using the full example data set
fit <- winIPCW(
  Surv(time, status) ~ Z1 + Z2 + Z3,
  data = windat,
  id = "id",
  trt = "A",
  endpoint = "event_type",
  tau = 36,
  method = c("ipcw", "ctw"),
  estimand = c("NB", "WR", "WO"),
  priority = c(1, 2),
  censor = "cox",
  margin = "cox",
  copula = "gumbel",
  eps = 0.01,
  conf.level = 0.95
)

## Extract and print results
out <- results(fit)
print(out)

\end{verbatim}

In this example, the formula \texttt{Surv(time, status) \textasciitilde{} Z1 + Z2 + Z3} specifies the working Cox model covariates for the censoring model and for the event-time margins used by the proposed estimator. The argument \texttt{id} identifies subjects, \texttt{trt} specifies the treatment indicator, and \texttt{endpoint} identifies the follow-up row and endpoint component rows. The treatment indicator is coded as \texttt{A = 1} for the treatment arm and \texttt{A = 0} for the control arm. The argument \texttt{tau = 36} requests estimates at the restriction time of 36 months, and \texttt{priority = c(1, 2)} specifies that component 1 has higher clinical priority than component 2. The option \texttt{method = c("ipcw", "ctw")} computes both the IPCW estimator of \citet{cui2025ipcw} and the proposed modified IPCW estimator with conditional tie weighting. For the proposed estimator, \texttt{margin = "cox"} uses Cox proportional hazards models for the event-time margins and \texttt{copula = "gumbel"} uses a Gumbel--Hougaard working copula. The resulting output is summarized in Table \ref{tab:supp-winipcw-example}.

\begin{table}[htbp]
\centering
\caption{Example output from the \texttt{winIPCW} package using the full included \texttt{windat} data set. The analysis reports the net benefit (NB), win ratio (WR), and win odds (WO) at $\tau=36$ months. Confidence intervals are Wald-type 95\% intervals; WR and WO intervals are constructed on the log scale.}
\label{tab:supp-winipcw-example}
\begin{tabular}{rlllrrrr}
\toprule
$\tau$ & Method & Copula & Estimand & Estimate & SE & Lower & Upper \\
\midrule
36 & IPCW   & --     & NB & $-0.014$ & $0.023$ & $-0.059$ & $0.032$ \\
36 & IPCW   & --     & WR & $0.851$  & $0.275$ & $0.496$  & $1.459$ \\
36 & IPCW   & --     & WO & $0.973$  & $0.047$ & $0.888$  & $1.066$ \\
36 & ctw & Gumbel & NB & $0.012$  & $0.042$ & $-0.071$ & $0.094$ \\
36 & ctw & Gumbel & WR & $1.032$  & $0.116$ & $0.823$  & $1.296$ \\
36 & ctw & Gumbel & WO & $1.023$  & $0.084$ & $0.868$  & $1.207$ \\
\bottomrule
\end{tabular}
\end{table}

The same syntax can be used to request estimates at multiple restriction times. For example, replacing \texttt{tau = 36} by \texttt{tau = c(12, 24, 36)} returns the requested win statistics at each restriction time. Additional working copulas for the modified IPCW estimator can be requested by changing the \texttt{copula} argument, for example \texttt{copula = c("gumbel", "clayton", "frank", "plackett")}.

\end{document}